\documentclass[11pt]{amsart}

\usepackage{subcaption}
\usepackage{amsmath}
\usepackage{amssymb}
\usepackage{amsthm}
\usepackage{graphicx}
\usepackage{mathtools}
\usepackage{url}
\usepackage{enumerate}
\usepackage{color}
\usepackage{hyperref}
\usepackage{color}

\usepackage{tikz}
\usepackage{pgfplots}
\usepgfplotslibrary{fillbetween}
\pgfplotsset{compat=1.11}
\usetikzlibrary{patterns}

\mathtoolsset{showonlyrefs}

\usepackage[margin=1.35in]{geometry}
\def\pr{\mathbb{P}}

\def\Lam{\Lambda}

\def\R{\mathbb{R}}
\def\P{\mathbb{P}}
\def\Z{\mathbb{Z}}
\def\eps{\epsilon}
\def\del{\delta}

\def\lam{\lambda}

\def\cG{\mathcal G}
\def\cC{\mathcal C}

\def\cI{\mathcal I}

\def\cL{\mathcal L}
\def\cS{\mathcal S}

\def\Ostab{\Xi^{\text{stab}}}

\def\bI{\mathbf I}

\def\cov{\text{cov}}

\def\tor{\mathbb T^d_n}
\def\torr{\mathbb T_n}

\newcommand{\cc}{\mathcal}

\newtheorem*{theorem*}{Theorem}
\newtheorem{theorem}{Theorem}[section]
\newtheorem{lemma}[theorem]{Lemma}

\newtheorem{defn}{Definition}

\newtheorem*{prop*}{Proposition}

\newtheorem*{claim*}{Claim}
\newtheorem*{fact*}{Fact}
\newtheorem{remark}{Remark}
\newtheorem*{defn*}{Definition}
\newtheorem*{openprob}{Open Problem}
\newtheorem{algorithm}{Algorithm}
\newtheorem{assumption}{Assumption}

\theoremstyle{definition}
\newtheorem{example}{Example}

\newcommand{\bydef}{\coloneqq}

\tikzset{ boundary/.style={ very thin, gray} }
\tikzset{ boundaryfill/.style={ very thin, gray, pattern=north west lines,
    opacity=0.4} }
\tikzset{ lam/.style={very thick,black} }

\tikzset{ bg/.style={very thick, green} }
\tikzset{ br/.style={very thick, red} }
\tikzset{ bb/.style={very thick, blue} }

\tikzset{ bv/.style={circle,fill=blue!50,draw,thick,minimum size=10pt,
    inner sep=0}}

\tikzset{ bvb/.style={circle,preaction={fill=blue!50},draw,pattern=north west
    lines,thick, minimum size=10pt, inner sep=0}}

\tikzset{ rv/.style={circle,fill=red!50,draw,thick,minimum
    size=10pt,inner sep=0}}

\tikzset{rvb/.style={circle,preaction={fill=red!50},draw,pattern=north west
    lines,thick, minimum size=10pt, inner sep=0}}

\tikzset{ gv/.style={circle,fill=green!50,draw,thick, minimum size=10pt,
    inner sep=0}}

\tikzset{ gvb/.style={circle,preaction={fill=green!50},draw,pattern=north west
    lines,thick, minimum size=10pt, inner sep=0}}

\tikzset{ ov/.style={circle,fill=black!50,draw,thick, minimum size=10pt,
    inner sep=0}}

\tikzset{ ovb/.style={circle,preaction={fill=black!50},draw,pattern=north west
    lines,thick, minimum size=10pt, inner sep=0}}

\tikzset{ ev/.style={circle,fill=white!50,draw,thick, minimum size=10pt,
    inner sep=0}}
\tikzset{ evb/.style={circle,preaction={fill=white!50},draw,pattern=north west
    lines,thick, minimum size=10pt, inner sep=0}}

\tikzset{empv/.style={minimum size=10pt, inner sep=0}}

\tikzset{unk/.style={circle,draw,black,fill=black!95, minimum size=5pt,
    inner sep=0} }

\begin{document}
\title{Algorithmic Pirogov-Sinai theory}
\author{Tyler Helmuth}
\author{Will Perkins}
\author{Guus Regts}


\address{University of Bristol}
\email{th17948@bristol.ac.uk}
\address{University of Illinois at Chicago}
\email{math@willperkins.org}
\address{University of Amsterdam}
\email{guusregts@gmail.com}

\begin{abstract}
    We develop an efficient algorithmic approach for approximate counting
  and sampling in the low-temperature regime of a broad class of
  statistical physics models on finite subsets of the lattice $\Z^d$
  and on the torus $(\Z/n\Z)^d$.  Our approach is based on combining
  contour representations from Pirogov--Sinai theory with Barvinok's
  approach to approximate counting using truncated Taylor series.
  Some consequences of our main results include an FPTAS for
  approximating the partition function of the hard-core model at
  sufficiently high fugacity on subsets of $\Z^d$ with appropriate
  boundary conditions and an efficient sampling algorithm for the
  ferromagnetic Potts model on the discrete torus $(\Z/n\Z)^d$ at
  sufficiently low temperature.
\end{abstract}

\date{6/16/2022}
\maketitle
\tableofcontents


\section{Introduction}
\label{secIntro}

For a wide class of equilibrium lattice statistical mechanics models
it is known that there is a phase transition from a high-temperature
disordered state to a low-temperature ordered state. In many cases
this transition is reflected in the dynamical and algorithmic behavior
of these models.  For example, a simple Markov chain (the Glauber
dynamics) provides an efficient means of sampling from many models on
finite subsets of $\Z^d$ at high temperatures but 
is often known to be inefficient at low
temperatures~\cite{borgs1999torpid}. For many models there are no
known efficient sampling algorithms at low temperatures, e.g., this is
the case for the well-studied hard-core model and for the ferromagnetic
$q$-state Potts model when $q$ and $d$ are greater than $2$. See
Sections~\ref{secPottsIntro} and~\ref{secIntroHC} for definitions of
these models.

Our main contribution is to rectify this by providing efficient
approximate counting and sampling algorithms at low temperatures on
subsets of $\Z^d $ and on the torus $\tor = (\Z/n\Z)^{d}$. Our results
apply to a wide class of statistical mechanics models, including the
hard-core and ferromagnetic Potts models. The following theorem is
representative of our results.
\begin{theorem}
  \label{thm:salespitch}
  For all $d\geq 2$ and $q\geq 2$ there exists
  $\beta^{\star}=\beta^{\star}(d,q)$ such that for all inverse temperatures
  $\beta>\beta^{\star}$ and all $c>0$ there is a polynomial-time
  algorithm to sample from the $q$-state Potts model on $\tor$ within
  $n^{-c}$ total variation distance.
\end{theorem}

To the best of our knowledge, this is the first provably efficient
sampling algorithm for the $q$-state Potts model on the torus $\tor$
below the critical temperature for $q, d \ge 3$. We are also able to
give an efficient algorithm to approximate the partition function of
the model, see Theorem~\ref{PottsMainThm} below.

Before describing our full results for the Potts and hard-core models
we briefly recall the motivation for, and intuition behind, our work.

There are two natural computational problems associated to the Potts
model and other discrete models from statistical physics. Given a
graph $G$ and an inverse temperature $\beta$ the \emph{counting}
problem is to compute the partition function $Z(G,\beta)$ of the
model, and the \emph{sampling} problem is to produce a sample
distributed according to the probability law of the model on $G$. If
we take the graph $G$ as our input, the algorithmic problem of
computing $Z(G,\beta)$ can be $\# P$-hard in general, and so research
has focused on providing \emph{approximate counting algorithms} that
return values close to $Z(G,\beta)$ and \emph{approximate sampling
  algorithms} that produce samples close in distribution to the given
model. For many problems, namely those that are
\emph{self-reducible}~\cite{jerrum1986random,sinclair1989approximate},
the existence of an efficient approximate counting algorithm implies
the existence of an efficient approximate sampling algorithm, and
vice-versa.

The existence of efficient algorithms for these
computational tasks is often known in the high-temperature
regime of statistical physics models. In contrast, algorithms are
often lacking in the low-temperature regime, even on restricted
classes of graphs like lattices. This often reflects the existence of
\emph{phase transitions} in these models on certain infinite graphs,
e.g., the infinite regular tree or $\Z^d$. 

At the same time, the low-temperature regime of many discrete
statistical physics models is fairly well-understood at a
probabilistic level when the graph considered is a nice subset of
$\Z^{d}$ or the torus, see,
e.g.,~\cite[Chapter~7]{friedli2017statistical}.  One might therefore
hope that the algorithmic tasks of sampling and counting are tractable
when restricted to these settings.  Theorem~\ref{thm:salespitch} and
our other results confirm that this is the case. While we focus in
this paper on the Potts and hard-core models as they are two of the
most studied lattice spin models, our main results
(Theorems~\ref{thmContourModelCount} and~\ref{mainThmSample}) are much
more general and apply to many discrete statistical physics models
e.g., the Widom--Rowlinson model, the Blume--Capel model, and many of
the $H$-coloring models described in~\cite{chayes2004sampling}.

The most systematic probabilistic understanding of the low-temperature
regime of discrete lattice spin models is based on Pirogov--Sinai
theory. Roughly speaking, this is a significantly more sophisticated
development of the Peierls' contour argument. The main idea of our
algorithms is to make use of Pirogov--Sinai theory to express the
logarithm of the partition function as a convergent cluster expansion,
where terms of the expansion correspond to overlapping clusters of
contours.  We then use the approach of Barvinok to approximate the
logarithm of the partition function, i.e., we truncate its Taylor
series expansion and compute the initial coefficients exactly by using
the cluster expansion representation. We describe this in more detail
in Section~\ref{secOverview} below.

Contour arguments have also been used to prove the slow mixing of
Markov chains on
lattices~\cite{borgs1999torpid,randall2006slow,borgs2012tight,blanca2016phase},
and our results can counterintuitively be phrased as saying that a
contour-based proof that a Markov chain on $\Z^{d}$ mixes slowly
implies the existence of an efficient sampling
algorithm at low enough temperatures.

In the next two sections we present our results for the Potts and
hard-core models in detail, but first we give precise definitions for
our notions of approximation.  In this introduction we only define
approximation for non-negative parameters though our main counting
algorithms (Theorems~\ref{PottsMainThm}--\ref{HCTorus} below) in fact
apply for complex parameters. Readers interested in complex parameters
should consult the more general Theorem~\ref{thmContourModelCount}.

We define fully polynomial-time approximation schemes in terms of the
approximate evaluation of polynomials since many counting problems can
be recast as the evaluation of a univariate polynomial.  For a
positive number $p$, we say $\hat p$ is an \emph{$\eps$-relative
  approximation} to $p$ if
$e^{-\eps} \hat p \le p \le e^{\eps} \hat p$.

\begin{defn}
  A \emph{fully polynomial time approximation scheme (FPTAS)} for
  approximating the evaluation of a polynomial $p(z)$ with nonnegative
  coefficients at $z >0$ is an algorithm that for any $\eps>0$
  produces an $\eps$-relative approximation to $p(z)$ and runs in time
  bounded by a polynomial in $\deg(p)$ and $1/\eps$.
\end{defn}

We use the total variation distance to measure the quality of an approximate sample.
\begin{defn}
  An \emph{$\eps$-approximate sample} from a probability measure $\mu$ is a configuration drawn according to a probability measure $\hat \mu$ with
\begin{equation*}
\| \hat \mu - \mu \|_{TV} < \eps \,.
\end{equation*}
\end{defn}

\begin{defn}
  Suppose $(\mu_n)$ is a sequence of probability measures indexed by
  $n$. An \emph{efficient sampling algorithm} is a randomized
  algorithm that returns an $\eps$-approximate sample to $\mu_n$ and
  runs in time polynomial in $n$ and $1/\eps$.
\end{defn}

\subsection{The Potts model}
\label{secPottsIntro}
The \emph{$q$-state Potts model} on a finite graph $G=(V,E)$ is the probability
distribution over assignments of $q$ colors to the vertices $V$ of $G$
given by
\begin{equation*}
  \mu_{G,q,\beta}(\sigma) \bydef  \frac{ \exp \left[\beta \sum_{\{i,j\}\in E} \mathbf 1_{\sigma_i = \sigma_j}  \right ]  }{  Z_{G,q}(\beta)}
\end{equation*}
where 
\begin{equation*}
  Z_{G,q}(\beta) \bydef \sum_{\sigma \in [q]^{V}} \exp \left[\beta \sum_{\{i,j\}\in E} \mathbf 1_{\sigma_i = \sigma_j}  \right ] \,
\end{equation*}
is the \emph{partition function}. We have written
$[q]\bydef\{1,2,\dots,q\}$ for the set of colours.  In what follows
we assume $\beta>0$, i.e., that the model is \emph{ferromagnetic},
meaning that it prefers configurations with more monochromatic edges.
The case $q=2$ of the Potts model is also called the \emph{Ising
  model}.

The Potts model is a simple model of a magnetic material and in
classical statistical physics it is studied on the $d$-dimensional
lattice $\Z^{d}$. For the remainder of this discussion we will
consider $d$ and the number of colors $q$ to be fixed. The Potts model
on $\Z^d$ is defined by taking a sequence of finite graphs
$\Lam_n \subset \Z^d$ so that $\Lam \to \Z^d$, and infinite volume
measures are obtained as weak limits of finite volume measures
$\mu_{\Lam_n,q,\beta}$.  If for a given choice of $\beta$ only one
infinite volume measure exists the model is said to be in the
\emph{uniqueness regime}. Otherwise, when multiple infinite volume
measures are possible, the Potts model is said to exhibit \emph{phase
  coexistence}.  The transition between uniqueness and coexistence as
$\beta$ changes is a \emph{phase transition} and occurs at a critical
point $\beta_{c}(d,q)$ (see, e.g.,~\cite{grimmett2006random}).
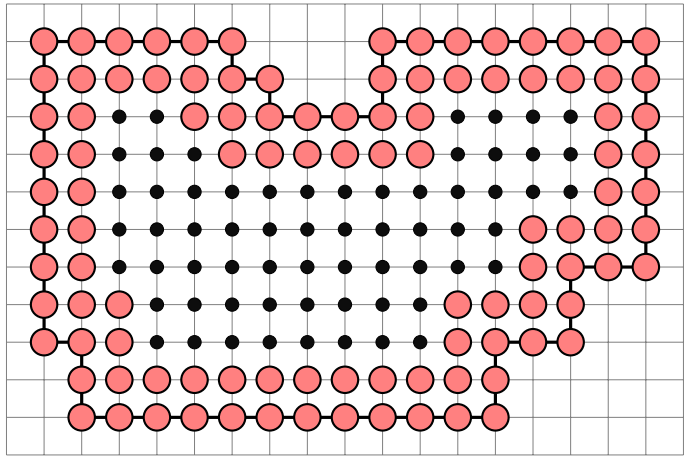
\begin{figure}
  \centering
  \begin{tikzpicture}[scale=.5]
    \draw[boundary] (-1,-1) grid (17,11);
    \draw[lam] (0,2) -- (0,10) -- (5,10) -- (5,9) -- (6,9) -- (6,8) --
    (9,8) -- (9,10) -- (16,10) -- (16,4) -- (14,4) -- (14,2) -- (12,2)
    -- (12,0) -- (1,0) -- (1,2) -- (0,2);

    \node[empv] at (0,0) {};
    \node[empv] at (0,1) {};
    \node[rv] at (0,2) {};
    \node[rv] at (0,3) {};
    \node[rv] at (0,4) {};
    \node[rv] at (0,5) {};
    \node[rv] at (0,6) {};
    \node[rv] at (0,7) {};
    \node[rv] at (0,8) {};
    \node[rv] at (0,9) {};
    \node[rv] at (0,10) {};

    \node[rv] at (1,0) {};
    \node[rv] at (1,1) {};
    \node[rv] at (1,2) {};
    \node[rv] at (1,3) {};
    \node[rv] at (1,4) {};
    \node[rv] at (1,5) {};
    \node[rv] at (1,6) {};
    \node[rv] at (1,7) {};
    \node[rv] at (1,8) {};
    \node[rv] at (1,9) {};
    \node[rv] at (1,10) {};

    \node[rv] at (2,0) {};
    \node[rv] at (2,1) {};
    \node[rv] at (2,2) {};
    \node[rv] at (2,3) {};
    \node[unk] at (2,4) {};
    \node[unk] at (2,5) {};
    \node[unk] at (2,6) {};
    \node[unk] at (2,7) {};
    \node[unk] at (2,8) {};
    \node[rv] at (2,9) {};
    \node[rv] at (2,10) {};

    \node[rv] at (3,0) {};
    \node[rv] at (3,1) {};
    \node[unk] at (3,2) {};
    \node[unk] at (3,3) {};
    \node[unk] at (3,4) {};
    \node[unk] at (3,5) {};
    \node[unk] at (3,6) {};
    \node[unk] at (3,7) {};
    \node[unk] at (3,8) {};
    \node[rv] at (3,9) {};
    \node[rv] at (3,10) {};

    \node[rv] at (4,0) {};
    \node[rv] at (4,1) {};
    \node[unk] at (4,2) {};
    \node[unk] at (4,3) {};
    \node[unk] at (4,4) {};
    \node[unk] at (4,5) {};
    \node[unk] at (4,6) {};
    \node[unk] at (4,7) {};
    \node[rv] at (4,8) {};
    \node[rv] at (4,9) {};
    \node[rv] at (4,10) {};

    \node[rv] at (5,0) {};
    \node[rv] at (5,1) {};
    \node[unk] at (5,2) {};
    \node[unk] at (5,3) {};
    \node[unk] at (5,4) {};
    \node[unk] at (5,5) {};
    \node[unk] at (5,6) {};
    \node[rv] at (5,7) {};
    \node[rv] at (5,8) {};
    \node[rv] at (5,9) {};
    \node[rv] at (5,10) {};

    \node[rv] at (6,0) {};
    \node[rv] at (6,1) {};
    \node[unk] at (6,2) {};
    \node[unk] at (6,3) {};
    \node[unk] at (6,4) {};
    \node[unk] at (6,5) {};
    \node[unk] at (6,6) {};
    \node[rv] at (6,7) {};
    \node[rv] at (6,8) {};
    \node[rv] at (6,9) {};
    \node[empv] at (6,10) {};

    \node[rv] at (7,0) {};
    \node[rv] at (7,1) {};
    \node[unk] at (7,2) {};
    \node[unk] at (7,3) {};
    \node[unk] at (7,4) {};
    \node[unk] at (7,5) {};
    \node[unk] at (7,6) {};
    \node[rv] at (7,7) {};
    \node[rv] at (7,8) {};
    \node[empv] at (7,9) {};
    \node[empv] at (7,10) {};

    \node[rv] at (8,0) {};
    \node[rv] at (8,1) {};
    \node[unk] at (8,2) {};
    \node[unk] at (8,3) {};
    \node[unk] at (8,4) {};
    \node[unk] at (8,5) {};
    \node[unk] at (8,6) {};
    \node[rv] at (8,7) {};
    \node[rv] at (8,8) {};
    \node[empv] at (8,9) {};
    \node[empv] at (8,10) {};

    \node[rv] at (9,0) {};
    \node[rv] at (9,1) {};
    \node[unk] at (9,2) {};
    \node[unk] at (9,3) {};
    \node[unk] at (9,4) {};
    \node[unk] at (9,5) {};
    \node[unk] at (9,6) {};
    \node[rv] at (9,7) {};
    \node[rv] at (9,8) {};
    \node[rv] at (9,9) {};
    \node[rv] at (9,10) {};

    \node[rv] at (10,0) {};
    \node[rv] at (10,1) {};
    \node[unk] at (10,2) {};
    \node[unk] at (10,3) {};
    \node[unk] at (10,4) {};
    \node[unk] at (10,5) {};
    \node[unk] at (10,6) {};
    \node[rv] at (10,7) {};
    \node[rv] at (10,8) {};
    \node[rv] at (10,9) {};
    \node[rv] at (10,10) {};

    \node[rv] at (11,0) {};
    \node[rv] at (11,1) {};
    \node[rv] at (11,2) {};
    \node[rv] at (11,3) {};
    \node[unk] at (11,4) {};
    \node[unk] at (11,5) {};
    \node[unk] at (11,6) {};
    \node[unk] at (11,7) {};
    \node[unk] at (11,8) {};
    \node[rv] at (11,9) {};
    \node[rv] at (11,10) {};

    \node[rv] at (12,0) {};
    \node[rv] at (12,1) {};
    \node[rv] at (12,2) {};
    \node[rv] at (12,3) {};
    \node[unk] at (12,4) {};
    \node[unk] at (12,5) {};
    \node[unk] at (12,6) {};
    \node[unk] at (12,7) {};
    \node[unk] at (12,8) {};
    \node[rv] at (12,9) {};
    \node[rv] at (12,10) {};

    \node[empv] at (13,0) {};
    \node[empv] at (13,1) {};
    \node[rv] at (13,2) {};
    \node[rv] at (13,3) {};
    \node[rv] at (13,4) {};
    \node[rv] at (13,5) {};
    \node[unk] at (13,6) {};
    \node[unk] at (13,7) {};
    \node[unk] at (13,8) {};
    \node[rv] at (13,9) {};
    \node[rv] at (13,10) {};

    \node[empv] at (14,0) {};
    \node[empv] at (14,1) {};
    \node[rv] at (14,2) {};
    \node[rv] at (14,3) {};
    \node[rv] at (14,4) {};
    \node[rv] at (14,5) {};
    \node[unk] at (14,6) {};
    \node[unk] at (14,7) {};
    \node[unk] at (14,8) {};
    \node[rv] at (14,9) {};
    \node[rv] at (14,10) {};

    \node[empv] at (15,0) {};
    \node[empv] at (15,1) {};
    \node[empv] at (15,2) {};
    \node[empv] at (15,3) {};
    \node[rv] at (15,4) {};
    \node[rv] at (15,5) {};
    \node[rv] at (15,6) {};
    \node[rv] at (15,7) {};
    \node[rv] at (15,8) {};
    \node[rv] at (15,9) {};
    \node[rv] at (15,10) {};

    \node[empv] at (16,0) {};
    \node[empv] at (16,1) {};
    \node[empv] at (16,2) {};
    \node[empv] at (16,3) {};
    \node[rv] at (16,4) {};
    \node[rv] at (16,5) {};
    \node[rv] at (16,6) {};
    \node[rv] at (16,7) {};
    \node[rv] at (16,8) {};
    \node[rv] at (16,9) {};
    \node[rv] at (16,10) {};
  \end{tikzpicture}
  \caption{Red padded boundary conditions for the Potts model on a
    region $\Lambda$. The thick black line passes through the interior
    vertex boundary $\partial^{\text{in}}\Lambda$ of
    $\Lambda$. Vertices determined by the boundary condition have been
    drawn red. Solid black vertices indicate where the configuration
    is not determined by the boundary conditions.}
  \label{fig:pad-potts}
\end{figure}

To state our results precisely requires two definitions.  Let $\Lam$
be a subgraph of $\Z^{d}$. We write $E(\Lam)\subset E(\Z^{d})$ for the
edge set of $\Lam$, and by a slight abuse of notation, we write
$\Lam$ in place of $V(\Lam)$ for the vertex set of $\Lam$. A finite
subgraph $\Lam$ is a \emph{region} if $\Lam^c$ is connected under the
adjacency relation derived from the distance function
$d_{\infty}(x,y) \bydef \max_{i=1}^d |x_i-y_i|$.  For a color
$\varphi \in [q]$, the set of allowed configurations with \emph{padded
  monochromatic boundary conditions} are:
\begin{equation*}
  \Omega_\Lam^\varphi \bydef 
  \{  \sigma \in [q]^\Lam :  d_{\infty}(i,\Lam^{c})\leq 2 \implies \sigma_{i}=\varphi
  \} \,.
\end{equation*}
See Figure~\ref{fig:pad-potts}. The corresponding partition function is
\begin{equation}
\label{eqPottsPartition}
  Z^\varphi_{q,\Lam}(\beta) \bydef \sum_{\omega \in
    \Omega_\Lam^\varphi} \exp \left[\beta \sum_{\{i,j\}
      \in E(\Lam)} \mathbf 1_{\sigma_i = \sigma_j}  \right ]. 
\end{equation}

\begin{theorem}
\label{PottsMainThm}
For all $d\ge2, q\ge2$, there exists
$\beta^{\star}=\beta^{\star}(d,q)>0$ so that for all
$\beta > \beta^{\star}$, there is an efficient sampling algorithm and
an FPTAS for the $q$-state Potts model on any finite region
$\Lam$ of $\Z^d$ with padded monochromatic boundary conditions.
\end{theorem}

The running time of these algorithms is $(n/\eps)^{O(\log d)}$ where $n$ is the number of vertices in the region $\Lam$.  While this is polynomial in $n$ and $1/\eps$, it would be desirable to improve the running time, perhaps to something close to linear in $n$. See Section~\ref{sec:MCMC} for more.

On the torus $\tor$ a great deal of work has gone into understanding
the mixing times of different Markov chains.  When $d=2$ a great deal
is known: the Glauber dynamics and Swendsen--Wang dynamics mix rapidly
(in polynomial time) for $\beta < \beta_c$, and the Swendsen--Wang
dynamics mix rapidly for
$\beta>\beta_c$~\cite{borgs1999torpid,borgs2012tight,ullrich2014rapid,gheissari2016mixing,blanca2017random}.
More generally the Swendesen--Wang dynamics are thought to be rapidly
mixing for all $d$ and $q$ when $\beta \ne \beta_c$.

Our results hold on the torus for a slightly weaker notion of approximation.
\begin{theorem}
  \label{PottsTorus}
  For all $d\ge2$ and $q\ge2$ there exists
  $\beta^{\star}=\beta^{\star}(d,q)$ and $c=c(d,q)>0$ so that for all
  $\beta > \beta^{\star}$ and all $\eps\ge e^{-cn}$ there is an
  algorithm to obtain an $\eps$-relative approximation of the
  partition function and an $\eps$-approximate sampling algorithm both
  running in time polynomial in $n$ and $1/\eps$ for the $q$-state
  Potts model on $\tor$. 
\end{theorem}

\subsubsection{Related results}
\label{sec:related-results}

Recall that an FPRAS is a randomized algorithm that returns an
$\eps$-relative approximation with probability at least $2/3$ and runs
in time polynomial in the instance size and $1/\eps$.  An FPRAS for
the ferromagnetic Ising models on general graphs was given by Jerrum
and Sinclair~\cite{jerrum1993polynomial}. Randall and
Wilson~\cite{randall1999sampling} showed that this algorithm can be
used to sample efficiently from the model.  Recently, Guo and
Jerrum~\cite{guo2018} gave an alternative sampling algorithm, based on
a Markov chain associated to the random cluster model.  For $q \ge 3$,
the complexity of approximating the ferromagnetic Potts model
partition function on general graphs is unknown.  It is \#BIS-hard (as
hard as approximately counting the number of independent sets in a
bipartite graph, see Section~\ref{secIntroHC}) to do so even on
bounded degree
graphs~\cite{goldberg2012approximating,galanis2016ferromagnetic}.

By making use of Theorems~\ref{thmContourModelCount}
and~\ref{thmPolySample} of the present article Jenssen, Keevash, and
Perkins have proven a variant of Theorem~\ref{PottsMainThm} for the
low temperature $q$-state Potts model on bounded degree expander
graphs~\cite{JenssenKeevashPerkins}. Subsequent to the initial posting
of the present article to the arXiv, Barvinok and Regts have given an
algorithm for approximating the partition function of the $q$-state
Potts model at low temperatures on a variety of
graphs~\cite{BarvinokRegts}. Their main hypotheses concerns the
existence of a nice set of generators for the cycle space of the
graph, and for finite simply connected subsets of $\Z^{d}$ they obtain
estimates for $\beta_{0}(q)$ that are better than those implicitly
given by Theorem~\ref{PottsMainThm}.

\subsection{Hard-core model}
\label{secIntroHC}

The \emph{hard-core model} on a finite graph $G$ is a random
independent set $\bI$ from the set $\cI(G)$ of all independent sets of
$G$ according to the distribution
\begin{equation*}
  \mu_{G,\lam}(I) \bydef \pr[ \bI = I] = \frac{\lam ^{|I|}}{ Z_G(\lam) },
\end{equation*}
where $\lam>0$ is the \emph{fugacity} and where the partition function is
\begin{equation*}
  Z_G(\lam) \bydef \sum_{I \in \cI(G)} \lam^{|I|}.
\end{equation*}

Our main result for the hard-core model is that if we take subsets of
$\Z^d$ with appropriate boundary conditions, then there are in fact
efficient counting and sampling algorithms at high fugacities. To
state our results, we recall that a vertex $i \in \Z^d$ is
\emph{even} (resp.\ \emph{odd}) if the sum of its coordinates is
even (resp.\ odd). 
For a finite region $\Lam$, the set of allowed configurations under
\emph{even padded boundary conditions} is
\begin{equation*}
 \cI^{\text{even}}(\Lam) \bydef  \{ I \in \cI(\Lam) :
  d_{\infty}(i,\Lam^{c})\leq 2 \implies \mathbf 1_{i\in
    I}=\mathbf 1_{i\text{ even}}\}\,,
\end{equation*}
and likewise for $\cI^{\text{odd}}(\Lam)$.  See
Figure~\ref{fig:pad-hc}.  The partition function is
\begin{equation*}
Z_\Lam^{\text{even}}(\lam) \bydef \sum_{I \in \cI^{\text{even}}(\Lam)} \lam^{|I|} \,.
\end{equation*}

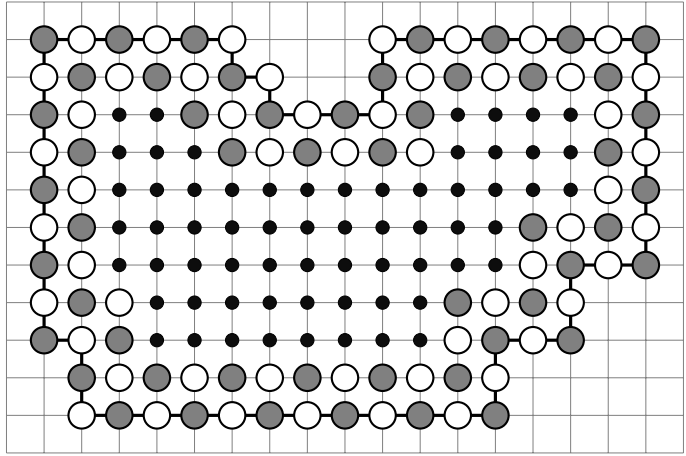
\begin{figure}
  \centering
  \begin{tikzpicture}[scale=.5]
    \draw[boundary] (-1,-1) grid (17,11);
    \draw[lam] (0,2) -- (0,10) -- (5,10) -- (5,9) -- (6,9) -- (6,8) --
    (9,8) -- (9,10) -- (16,10) -- (16,4) -- (14,4) -- (14,2) -- (12,2)
    -- (12,0) -- (1,0) -- (1,2) -- (0,2);

    \node[empv] at (0,0) {};
    \node[empv] at (0,1) {};
    \node[ov] at (0,2) {};
    \node[ev] at (0,3) {};
    \node[ov] at (0,4) {};
    \node[ev] at (0,5) {};
    \node[ov] at (0,6) {};
    \node[ev] at (0,7) {};
    \node[ov] at (0,8) {};
    \node[ev] at (0,9) {};
    \node[ov] at (0,10) {};

    \node[ev] at (1,0) {};
    \node[ov] at (1,1) {};
    \node[ev] at (1,2) {};
    \node[ov] at (1,3) {};
    \node[ev] at (1,4) {};
    \node[ov] at (1,5) {};
    \node[ev] at (1,6) {};
    \node[ov] at (1,7) {};
    \node[ev] at (1,8) {};
    \node[ov] at (1,9) {};
    \node[ev] at (1,10) {};

    \node[ov] at (2,0) {};
    \node[ev] at (2,1) {};
    \node[ov] at (2,2) {};
    \node[ev] at (2,3) {};
    \node[unk] at (2,4) {};
    \node[unk] at (2,5) {};
    \node[unk] at (2,6) {};
    \node[unk] at (2,7) {};
    \node[unk] at (2,8) {};
    \node[ev] at (2,9) {};
    \node[ov] at (2,10) {};

    \node[ev] at (3,0) {};
    \node[ov] at (3,1) {};
    \node[unk] at (3,2) {};
    \node[unk] at (3,3) {};
    \node[unk] at (3,4) {};
    \node[unk] at (3,5) {};
    \node[unk] at (3,6) {};
    \node[unk] at (3,7) {};
    \node[unk] at (3,8) {};
    \node[ov] at (3,9) {};
    \node[ev] at (3,10) {};

    \node[ov] at (4,0) {};
    \node[ev] at (4,1) {};
    \node[unk] at (4,2) {};
    \node[unk] at (4,3) {};
    \node[unk] at (4,4) {};
    \node[unk] at (4,5) {};
    \node[unk] at (4,6) {};
    \node[unk] at (4,7) {};
    \node[ov] at (4,8) {};
    \node[ev] at (4,9) {};
    \node[ov] at (4,10) {};

    \node[ev] at (5,0) {};
    \node[ov] at (5,1) {};
    \node[unk] at (5,2) {};
    \node[unk] at (5,3) {};
    \node[unk] at (5,4) {};
    \node[unk] at (5,5) {};
    \node[unk] at (5,6) {};
    \node[ov] at (5,7) {};
    \node[ev] at (5,8) {};
    \node[ov] at (5,9) {};
    \node[ev] at (5,10) {};

    \node[ov] at (6,0) {};
    \node[ev] at (6,1) {};
    \node[unk] at (6,2) {};
    \node[unk] at (6,3) {};
    \node[unk] at (6,4) {};
    \node[unk] at (6,5) {};
    \node[unk] at (6,6) {};
    \node[ev] at (6,7) {};
    \node[ov] at (6,8) {};
    \node[ev] at (6,9) {};
    \node[empv] at (6,10) {};

    \node[ev] at (7,0) {};
    \node[ov] at (7,1) {};
    \node[unk] at (7,2) {};
    \node[unk] at (7,3) {};
    \node[unk] at (7,4) {};
    \node[unk] at (7,5) {};
    \node[unk] at (7,6) {};
    \node[ov] at (7,7) {};
    \node[ev] at (7,8) {};
    \node[empv] at (7,9) {};
    \node[empv] at (7,10) {};

    \node[ov] at (8,0) {};
    \node[ev] at (8,1) {};
    \node[unk] at (8,2) {};
    \node[unk] at (8,3) {};
    \node[unk] at (8,4) {};
    \node[unk] at (8,5) {};
    \node[unk] at (8,6) {};
    \node[ev] at (8,7) {};
    \node[ov] at (8,8) {};
    \node[empv] at (8,9) {};
    \node[empv] at (8,10) {};

    \node[ev] at (9,0) {};
    \node[ov] at (9,1) {};
    \node[unk] at (9,2) {};
    \node[unk] at (9,3) {};
    \node[unk] at (9,4) {};
    \node[unk] at (9,5) {};
    \node[unk] at (9,6) {};
    \node[ov] at (9,7) {};
    \node[ev] at (9,8) {};
    \node[ov] at (9,9) {};
    \node[ev] at (9,10) {};

    \node[ov] at (10,0) {};
    \node[ev] at (10,1) {};
    \node[unk] at (10,2) {};
    \node[unk] at (10,3) {};
    \node[unk] at (10,4) {};
    \node[unk] at (10,5) {};
    \node[unk] at (10,6) {};
    \node[ev] at (10,7) {};
    \node[ov] at (10,8) {};
    \node[ev] at (10,9) {};
    \node[ov] at (10,10) {};

    \node[ev] at (11,0) {};
    \node[ov] at (11,1) {};
    \node[ev] at (11,2) {};
    \node[ov] at (11,3) {};
    \node[unk] at (11,4) {};
    \node[unk] at (11,5) {};
    \node[unk] at (11,6) {};
    \node[unk] at (11,7) {};
    \node[unk] at (11,8) {};
    \node[ov] at (11,9) {};
    \node[ev] at (11,10) {};

    \node[ov] at (12,0) {};
    \node[ev] at (12,1) {};
    \node[ov] at (12,2) {};
    \node[ev] at (12,3) {};
    \node[unk] at (12,4) {};
    \node[unk] at (12,5) {};
    \node[unk] at (12,6) {};
    \node[unk] at (12,7) {};
    \node[unk] at (12,8) {};
    \node[ev] at (12,9) {};
    \node[ov] at (12,10) {};

    \node[empv] at (13,0) {};
    \node[empv] at (13,1) {};
    \node[ev] at (13,2) {};
    \node[ov] at (13,3) {};
    \node[ev] at (13,4) {};
    \node[ov] at (13,5) {};
    \node[unk] at (13,6) {};
    \node[unk] at (13,7) {};
    \node[unk] at (13,8) {};
    \node[ov] at (13,9) {};
    \node[ev] at (13,10) {};

    \node[empv] at (14,0) {};
    \node[empv] at (14,1) {};
    \node[ov] at (14,2) {};
    \node[ev] at (14,3) {};
    \node[ov] at (14,4) {};
    \node[ev] at (14,5) {};
    \node[unk] at (14,6) {};
    \node[unk] at (14,7) {};
    \node[unk] at (14,8) {};
    \node[ev] at (14,9) {};
    \node[ov] at (14,10) {};

    \node[empv] at (15,0) {};
    \node[empv] at (15,1) {};
    \node[empv] at (15,2) {};
    \node[empv] at (15,3) {};
    \node[ev] at (15,4) {};
    \node[ov] at (15,5) {};
    \node[ev] at (15,6) {};
    \node[ov] at (15,7) {};
    \node[ev] at (15,8) {};
    \node[ov] at (15,9) {};
    \node[ev] at (15,10) {};

    \node[empv] at (16,0) {};
    \node[empv] at (16,1) {};
    \node[empv] at (16,2) {};
    \node[empv] at (16,3) {};
    \node[ov] at (16,4) {};
    \node[ev] at (16,5) {};
    \node[ov] at (16,6) {};
    \node[ev] at (16,7) {};
    \node[ov] at (16,8) {};
    \node[ev] at (16,9) {};
    \node[ov] at (16,10) {};
  \end{tikzpicture}
  \caption{Even padded boundary conditions for the hard-core model on
    a region $\Lambda$. The thick black line passes through the
    interior vertex boundary $\partial^{\text{in}}\Lambda$ of
    $\Lam$. Vertices required to be occupied or unoccupied by the
    boundary conditions are drawn as gray or white circles,
    respectively. Solid black vertices indicate where the
    configuration is not determined by the boundary conditions; note these
    vertices may be required to be unoccupied due to sharing
    an edge with a vertex required to be occupied by the boundary
    conditions.}
  \label{fig:pad-hc}
\end{figure}

\begin{theorem}
\label{HCMainThm}
For $d\ge 2$ there exists a $\lam^{\star}=\lam^{\star}(d)$ such that for
all $\lam > \lam^{\star}$, there is an efficient sampling algorithm
and an FPTAS for the hard-core model on any finite region $\Lam$ of
$\Z^d$ with even or odd padded boundary conditions.
\end{theorem}

We also establish efficient counting and sampling algorithms on $\tor$
when $n$ is even; this ensures the existence of an independent set
that contains half of the vertices of $\tor$.
\begin{theorem}
  \label{HCTorus}
  For $d\ge2$ there exists $\lam^{\star}=\lam^*(d)$ and $c=c(d)>0$ so
  that for all $\lam > \lam^{\star}$ and all $\eps\ge e^{-cn}$ there
  is an algorithm to approximate the partition function to within
  $\eps$-relative error and an $\eps$-approximate sampling algorithm
  both running in time polynomial in $n$ and $1/\eps$ for the hard-core
  model on the torus $\tor$ for even $n$. 
\end{theorem}

The value of $\lam^\star(d)$ we obtain is exponentially large in $d$,
as in the results for slow mixing in~\cite{borgs1999torpid}.  We
expect Theorem~\ref{HCTorus} to hold for much smaller $\lam^\star$, in
particular with $\lam^\star(d) \to 0$ as $d\to \infty$ as in the
proofs of phase coexistence in the hard-core model on
$\Z^d$~\cite{galvin2004phase,peled2014odd}. See
Section~\ref{secConclude}.

\subsubsection{Related results}
\label{sec:related-results-1}

For graphs of maximum degree at most $\Delta$ a clear picture has
emerged about the existence of an FPTAS for computing $Z_G(\lam)$.  A
crucial role is played by the value
$\lam_c(\Delta) \bydef \frac{(\Delta-1)^{\Delta-1}}{(\Delta-2)^\Delta}$,
the uniqueness threshold for the infinite $d$-regular tree.  For
$\lam < \lam_c(\Delta)$, Weitz~\cite{weitz2006counting} gave an FPTAS
for approximating $Z_G(\lam)$ on all graphs of max degree $\Delta$.
Conversely, Sly~\cite{sly2010computational}, Sly and
Sun~\cite{sly2014counting}, and Galanis, \v{S}tefankovi\v{c}, and
Vigoda~\cite{galanis2016inapproximability} showed that for
$\lam > \lam_c(\Delta)$ there is no FPRAS for approximating
$Z_G(\lam)$ unless $\text{NP}=\text{RP}$, where RP is the class of
problems that can be solved in polynomial time by a randomized
algorithm.

The problem of counting independent sets on bipartite graphs is called
\#BIS, and no such hardness result is known for \#BIS. Several
important problems have been shown to be as hard as \#BIS to
approximate, including the problem of approximating the ferromagnetic
Potts model partition function on general
graphs~\cite{goldberg2012approximating,cai2016hardness,galanis2016ferromagnetic}.
The problem \#BIS may be easier than the problem of approximating the
hard-core partition function on general graphs: unlike on general
graphs, finding the size of the \emph{largest} independent set is easy
on bipartite graphs. It is a major open problem in complexity theory
to determine the complexity of \#BIS~\cite{dyer2004relative}.

\subsection{Overview of the algorithms}
\label{secOverview}
The preceding theorems will be proven as applications of more general
results about \emph{polymer models} and \emph{contour models}. We
introduce polymer models in Section~\ref{secTaylor} below, and contour
models in Section~\ref{secContourModels}. In the current section,
which gives an informal overview of our algorithms, we elide the
distinction between polymers and contours, and for simplicity we will
write contour models. The idea behind contour models is introduced in
Section~\ref{secOverviewPolymer}, we outline our approximation
algorithms in Section~\ref{secOverviewApprox}, and lastly we describe
our sampling algorithms in Section~\ref{secOverviewSample}.

\subsubsection{Contour models}
\label{secOverviewPolymer}
For many discrete statistical mechanics models there are regimes in
which the most likely configuration is simple to describe. For
example, in the hard-core model the most likely configuration at low
fugacities is the empty independent set, while at high fugacities the
most likely configurations are the all-even or all-odd occupied independent
sets. Contour models are a geometric way to represent spin models in
terms of their deviations from these most likely configurations, which
we will henceforth call \emph{ground states}.

In the simplest settings such a representation involves re-writing a
partition function as a sum over a suitable class of subgraphs. For
example, this can be done for the high-temperature Ising model. In
more complex situations, Pirogov--Sinai theory provides an appropriate
representation. We defer the details of this to
Section~\ref{secContourModels}. For the purposes of this introduction
it will suffice that the reader has in mind that a contour model
expresses the partition function as a sum over collections of disjoint geometric
objects.

\subsubsection{Approximation algorithms using contour
  models}
\label{secOverviewApprox}

Our algorithm for approximating the partition function will be based
on truncating the Taylor series for $\log Z_{G}$ after a given number
of terms.  There are several components to making this work:
\begin{enumerate} 
\item We write the partition function as an abstract contour model as
  dictated by Pirogov--Sinai
  theory~\cite{pirogov1975phase,pirogov1976phase}.
\item We prove that the partition function, as a function of the
  inverse temperature, does not vanish outside a disc in the complex
  plane.  We do this by using the Peierls' condition and a theorem of
  Borgs and Imbrie~\cite{borgs1989unified} implementing Zahradnik's
  version~\cite{zahradnik1984alternate} of Pirogov--Sinai theory.
\item We use the absence of zeros to write error bounds for
  the truncated Taylor series for the log partition function,
  following Barvinok~\cite{barvinok2017combinatorics,barvinok2016computing}.  
\item We efficiently compute the low-order coefficients of the Taylor
  series.  This is done inductively using the cluster expansion.
\end{enumerate}

None of these components are wholly new -- our main contribution is to
establish the relevance of Pirogov--Sinai theory to the design of
algorithms.  In this paper we strive for simplicity and clarity of the
main ideas, and so we do not try to pursue optimal bounds or maximal
generality in stating theorems. We believe, however, that essentially
any application of Pirogov--Sinai theory to prove phase coexistence or
to prove slow mixing for discrete lattice spin models can be turned
into efficient approximate counting and sampling algorithms with the
ideas of this paper.

\subsubsection{Samping algorithms using contour models}
\label{secOverviewSample}

Often efficient approximation algorithms lead to efficient sampling
algorithms via self-reducibility. The basic idea is that if one can
accurately approximate the partition function $Z_{G}$ for arbitrary
$G$ with \emph{arbitrary boundary conditions}, then one can accurately
estimate the probability of a configuration by expressing it as a telescoping
product of partition functions. The idea is already evident in the
expression for the probability that a vertex $v$ is occupied in
the hard-core model:
\begin{equation*}
  \P_{G,\lambda}[\text{$v$ occupied}] = \lambda\frac{Z_{G\setminus
      N(v)}(\lambda)}{Z_{G}(\lambda)},
\end{equation*}
where $N(v)$ is the union of $\{v\}$ and the set of neighbours of
$v$. This expressions arises as $v$ being occupied implies that no
neighbour of $v$ is occupied. We think of the numerator as being a
partition function with a boundary condition that $N(v)$ is
unoccupied.

The derivation of contour representations in Pirogov--Sinai theory
makes use of particular boundary conditions: the padded boundary
conditions introduced in Sections~\ref{secPottsIntro}
and~\ref{secIntroHC}. This leads to a difficulty in using
self-reducibility to define sampling algorithms, as changing the
boundary conditions may lead to a situation in which we do not have a
contour representation.  We circumvent this difficulty by using the
idea of self-reducibility on the level of contours: instead of
iteratively determining a spin configuration spin by spin, we instead
iteratively determine a contour configuration contour by contour. The
manner in which contours are defined ensures that we are always able
to write the partition functions that arise in terms of contour
representations.

Obtaining a spin configuration from a contour configuration is
straightforward, and we defer a discussion of this point until after
we have defined contour models precisely.

\subsection{Organization and Conventions}
In Section~\ref{secTaylor} we define polymer models and present both
the cluster expansion and Taylor series for the log partition
function.  Under the condition of a zero-free region of the partition
function in the complex plane, we give an efficient algorithm for
approximating the partition function of a polymer model.

In Section~\ref{secContourModels} we define the more sophisticated
contour models from Pirogov--Sinai theory, and show that the algorithm
of Section~\ref{secTaylor} can be applied to approximate the partition
function of a contour model under suitable hypotheses. We discuss how
to verify the main hypothesis, which is the convergence of the cluster
expansion, in Section~\ref{secZeroFree}. By using a theorem of Borgs
and Imbrie~\cite{borgs1989unified} we verify this condition for the
Potts model and the hard-core model.

In Section~\ref{secSample} we prove our main sampling results.
Establishing our results for the torus $\tor$ requires some additional
work and we carry this out in Section~\ref{secTorus}. In
Section~\ref{secConclude} we conclude with some directions for future
work.

We end this section with some notation and conventions that will be
used throughout. All logarithms are natural logarithms. If $G$ is a
graph we write $|G|$ for the size of the vertex set of $G$.

A finite subset $\Lam \subset \Z^d$ is \emph{c-connected} if $\Lam^c$
is connected under the adjacency relation derived from the distance function
$d_{\infty}(x,y) = \max_{i=1}^d |x_i-y_i|$. We also call c-connected
subsets \emph{regions}.  The \emph{interior boundary} of a set
$A \subset \Z^d$ is
$\partial^{\text{in}} A = \{ i \in A: d_{\infty}(i, A^c)=1 \}$.  The
\emph{exterior boundary} of a set $A \subset \Z^d$ is
$\partial^{\text{ex}} A = \{ i \in A^c: d_{\infty}(i, A)=1 \}$.  On
the torus $\tor$, with the vertex set viewed as $\{1,\dots, n \}^d$,
we define the $d_\infty$ distance in the natural way, with
$d_{\infty}(x,y) = \max_{i=1}^d \min\{(x_i-y_i) \mod n , (y_i - x_i)
\mod n \}$.

\section{Cluster expansions, Taylor series, and approximate counting}
\label{secTaylor}

In this section we introduce polymer models and the cluster expansion,
and describe how they can be used algorithmically. To illustrate the
method we recover results of Patel--Regts~\cite{patel2016deterministic}
and Liu--Sinclair--Srivastava~\cite{liu2017ising} on the efficient
approximation of the hard-core and Ising models. The method of this
section is at the heart of the proofs of our main results for more
sophisticated contour models.

\subsection{Polymer models}
\label{secPolymer}
Let $G=(V,E)$ be a finite graph and let $\Omega$ be a finite set of
\emph{spins}.  Define a \emph{polymer} $\gamma$ in $G$ to be a pair
$\gamma=(\overline \gamma, \omega_{\overline \gamma})$ where
$\overline \gamma$, the \emph{support} of the polymer, is a connected
subgraph of $G$ and
$\omega_{\overline \gamma}\colon\overline\gamma\to\Omega$ is an
assignment of a spin from $\Omega$ to each vertex in
$\overline \gamma$.  The \emph{size} of a polymer is
$|\overline \gamma|$.  A \emph{polymer model} consists of a set
$\cC(G)$ of polymers along with \emph{weight functions}
$w(\gamma,\cdot)\colon\mathbb{C}\to\mathbb{C}$ for each polymer
$\gamma$.  We need one assumption about the weight functions:
\begin{assumption}
  \label{asAnalytic}
  The weight functions $w(\gamma,z)$ are analytic functions of $z$ in
  a neighborhood of the origin of the complex plane, and there is an
  absolute constant $\rho>0$ such that for each $\gamma\in\cC(G)$ the
  first non-zero term in the Taylor series expansion of $w(\gamma,z)$
  around zero is of order $k \ge |\overline \gamma| \rho$.
\end{assumption}
Note that Assumption~\ref{asAnalytic} implies $w(\gamma,0)=0$ for all
$\gamma$ with non-empty support.

We say two polymers $\gamma, \gamma'\in\cC(G)$ are \emph{compatible}
if $d(\overline \gamma,\overline \gamma')>1$, where $d(\cdot,\cdot)$
is the graph distance in $G$. Let $\cG(G)$ be the collection of all
finite sets of polymers from $\cC(G)$ that are pairwise compatible,
including the empty set of polymers.

The  partition function associated to the polymer model defined by
$\cC(G)$ is 
\begin{align}
\label{eqPolymerPartition}
Z(G,z) \bydef \sum_{\Gamma \in \cG(G)} \prod_{\gamma \in \Gamma} w(\gamma,z)
\end{align}
where the term corresponding to the empty set of polymers is $1$ by
convention.  We think of $Z(G,z)$ as a function of one complex
variable $z$.

\begin{example}[Hard-core model at low density]
  \label{ex:HC-LD}
  The hard-core model is the simplest model to describe as a polymer
  model.  Polymers are single vertices, i.e., $\cC(G) = V(G)$. The
  spin set, which is superfluous in this simple example, is
  $\Omega = \{1\}$: every polymer receives the same spin $1$, which is
  interpreted as meaning the vertex is `occupied'.  The weight
  function of each polymer is $w(\gamma,z) = z$.  Two polymers are
  compatible if their distance in the graph is more than $1$, and so
  the sets of pairwise compatible polymers are exactly the independent
  sets of $G$, and the polymer partition function is the hard-core
  model partition function at fugacity $z$:
  \begin{equation*}
    Z(G,z) = \sum_{\Gamma \in \cG(G)} \prod_{\gamma \in \Gamma} w(\gamma,z)  
           = \sum_{I \in \cI(G)} z^{|I|} = Z_G(z)  \,.
  \end{equation*}
\end{example}

\begin{example}[Ising model with free boundary conditions and an external field]
  \label{ex:Ising-ExF}
  Consider the Ising model with free boundary conditions and an
  external field $z$.  That is
\begin{equation*}
  Z_G(\beta,z) \bydef \sum_{ \sigma \in \{\pm 1\}^{V(G)}} z^{\sum_{v \in
      V(G)} \sigma(v)} \prod_{\{u,v\} \in E(G)} e^{\beta \sigma(u)
    \sigma(v) } \,. 
\end{equation*}
Assume $|z| <1$, so $-1$ spins are preferred. To obtain a polymer
model representation we can express the partition function in terms of
deviations from the all $-1$ configuration.  That is, a polymer
$\gamma$ is a connected induced subgraph $\overline\gamma$ of
vertices, all labeled $+1$.  Then we can write
\begin{equation*}
  Z_G(\beta,z) =  z^{-|G|} e^{\beta |E(G)|}   \sum_{\Gamma \in \cG(G)}
  \prod_{\gamma \in \Gamma} w(\gamma,z) \, ,
\end{equation*} 
where, letting $\partial _e \overline \gamma = | \{ \{u,v\} \in E(G) :
u\in \overline \gamma, v \notin \overline \gamma \} |$, the weight
function is
\begin{equation*}
w(\gamma,z) = z^{2 |\overline \gamma|} e^{-2 \beta |\partial _e
  \overline \gamma|}.
\end{equation*}
\end{example}

\subsection{The cluster expansion}
The \emph{cluster expansion} is the following formal power series
representation for $\log Z(G,z)$, see,
e.g.,~\cite{kotecky1986cluster,friedli2017statistical}. Under suitable
conditions, see Section~\ref{secZeroFree} below, it is also an
absolutely convergent power series representation.
\begin{equation}
\label{eqClusterExpand}
\log Z(G,z) = \sum_{k \ge 1} \frac{1}{k!}\sum_{(\gamma_1, \dots, \gamma_k)}  \phi (\gamma_1, \dots, \gamma_k) \prod_{i=1}^k w(\gamma_i,z) \,.
\end{equation}
The sum in \eqref{eqClusterExpand} is over ordered $k$-tuples of
polymers from $\cC(G)$, and $\phi$ is the \emph{Ursell function},
which we now define.

Let $H=H(\gamma_1, \dots, \gamma_k)$ be the \emph{incompatibility
  graph} of polymers $\gamma_1, \dots, \gamma_k$, i.e., the graph on
$k$ vertices with an edge between $\gamma_i$ and $\gamma_j$ if and
only if $\gamma_i$ and $\gamma_j$ are not compatible. Then
 \begin{equation*}
   \phi(\gamma_1, \dots, \gamma_k) \bydef \sum_{ \substack{E \subseteq E(H) \\ \text{spanning, connected}}} (-1)^{|E|}  \,. 
\end{equation*} 
The sum is over spanning and connected edge sets of $H$.  Thus
$\phi(\gamma_1, \dots ,\gamma_k)=0$ if $H$ is disconnected. By
definition, the Ursell function depends only on the graph $H$ induced
by the incompatibility relation, and not on the polymers
$\gamma_1, \dots, \gamma_k$ themselves.

It will be convenient for us later to rewrite \eqref{eqClusterExpand}
as a sum over unordered multisets of polymers from $\cC(G)$.  Given a
multiset $M=\{\gamma_1^{m_1},\ldots,\gamma_t^{m_t}\}$,
there are exactly $\binom{k}{m_1\cdots m_t}$ $k$-tuples which have $M$
as underlying multiset. Here the
exponents $m_i$ denote the multiplicities of the elements in $M$, and
$k = \sum_{i=1}^{t}m_{i}$. We can therefore rewrite \eqref{eqClusterExpand} as
\begin{equation}
\log Z(G,z) = \sum_{k \ge 1} \frac{1}{k!}\sum_{\{ \gamma_1^{m_1}, \dots, \gamma_t^{m_t} \}} \binom{k}{m_1\cdots m_t} \phi (\gamma_1^{m_1}, \dots, \gamma_t^{m_t}) \prod_{i=1}^t w(\gamma_i,z)^{m_i} \,, \label{eqClusterExpandM} 
\end{equation}
where $\phi (\gamma_1^{m_1}, \dots, \gamma_t^{m_t})$ is the Ursell
function applied to the incompatibility graph of the collection
  of polymers
$\underbrace{\gamma_1,\ldots,\gamma_1}_{m_1},\underbrace{\gamma_2,\ldots,\gamma_2}_{m_2},\ldots,\underbrace{\gamma_t,\ldots,\gamma_t}_{m_t}$.

\subsection{The Taylor series}

We can also Taylor expand $\log Z(G,z)$ around $z=0$:
\begin{equation}
  \label{eqTaylorSeries}
  \log Z(G,z) = \sum_{k\ge 1}  \frac{z^k}{k!} \frac{\partial^k}{\partial z^k} \log Z(G,0) \,.
\end{equation}

In fact, as observed by Dobrushin~\cite{dobrushin1996estimates}, the
cluster expansion and Taylor series are the same power series in
$z$, though arranged differently.  By our assumptions on the
weight functions, for each $k$ only a finite number of terms in the
cluster expansion contribute to the coefficient of $z^{k}$, and so we
can compute the coefficients of the Taylor series via the cluster
expansion:

\begin{equation}
  \label{eqTayfromCluster}
  \frac{\frac{\partial^k}{\partial z^k}  \log Z(G,0) }{ k!} 
  = 
  \sum_{j =1}^{k} \frac{1}{j!}\sum_{(\gamma_1, \dots, \gamma_j )}  
  \phi (\gamma_1, \dots, \gamma_j) 
  \frac{1}{k!} \frac{\partial^k}{\partial z^k} \left( \prod_{i=1}^j
    w(\gamma_i,z) \right)_{z=0} \,.
\end{equation}

\subsection{Approximate counting for polymer models}

The partial sums of the Taylor series are
\begin{equation*}
  T_m(G,z)  \bydef \sum_{k= 1}^m  \frac{z^k}{k!} \frac{\partial^k}{\partial z^k} \log Z(G,0)\,.
\end{equation*}

If we know $Z(G,z)$ is non-zero in a disc around the origin in the
complex plane, then we can control the error of the truncated Taylor
series approximation for $\log Z(G,z)$.  This is the approach of
Barvinok for devising approximation algorithms
~\cite{barvinok2015computing,barvinok2016computingP,barvinok2016computing,barvinok2017combinatorics}.
The next lemma rephrases \cite[Lemma 2.2]{patel2016deterministic} and
indicates where to truncate the Taylor series to get a good
approximation.  We use the following notion of relative error for
complex numbers.
\begin{defn}
  An \emph{$\eps$-relative approximation} to a complex number
  $Z\neq 0$ is a complex number $\hat Z\neq 0$ so that
  \begin{equation*}
    e^{-\eps} \le \left |\frac{Z}{\hat Z} \right| \le e^{\eps}
  \end{equation*}
  and the angle between $Z$ and $\hat Z$ as vectors in the complex
  plane is at most $\eps$.
\end{defn}
\begin{lemma}
  \label{lemTaylor}
  Suppose the degree of the polynomial $Z(G, z)$ is at most $N$ and
  suppose that $Z(G,z) \ne 0$ for all $|z| \le \del$.  Then for every
  $\eps>0$ and every $|z| < \del$, $\exp[ T_m(G,z)]$ is an
  $\eps$-relative approximation to $ Z(G,z)$ for all
\begin{equation*}
m \ge \frac{\log (N/\eps)}{1-|z|/\del} \,.
\end{equation*}
\end{lemma}
Lemma~\ref{lemTaylor} implies that if we can compute all of the coefficients
$\frac{\partial^k}{\partial z^k} \log Z(G,0)$ for $k=1, \dots, m$ in
time $\exp ( O( m))$, then we obtain an algorithm to produce
$\eps$-relative approximations of $ Z(G,z)$ with a running time
polynomial in $N$ and $1/\eps$ when $|z|<\delta$.

\begin{defn*}
  \label{defComputem}
  We can \emph{compute a function $f(z)$ up to order $m$} if we
  can compute the coefficients of the Taylor series of $f(z)$ around
  $0$ up to order $m$.
\end{defn*}

\begin{theorem}
  \label{thmPolymerApprox}
  Fix $\Delta$ and let $\mathfrak G$ be a set of graphs of degree at
  most $\Delta$.  Suppose:
  \begin{itemize}
  \item There is a constant $C$ so that $Z(G,z)$ is a polynomial in
    $z$ of degree at most $C |G|$ for all $G \in \mathfrak G$.
  \item The weight functions satisfy Assumption~\ref{asAnalytic}, and
    we can compute $w(\gamma,z)$ up to order $m$ for all
    $G \in \mathfrak G$ and all $\gamma \in \cC(G)$ in time
    $\exp(O(m+\log |G|))$.
  \item For every connected subgraph $G' $ of every
    $G \in \mathfrak G$, we can list all polymers $\gamma\in\cC(G) $
    with $\overline \gamma = G'$ in time $\exp(O(|G'|))$.
  \item There exists $\del >0$ so that for all $|z| < \del$ and all
    $G \in \mathfrak G$, $Z(G, z) \ne 0$.
  \end{itemize}
  Then for every $z$ with $|z|<\del$, there is an FPTAS for $Z(G,z)$
  for all $G \in \mathfrak G $.
\end{theorem}

The proof of Theorem~\ref{thmPolymerApprox} requires a few
lemmas.\footnote{To maintain consistency of the numbering of results
  with previous version of the paper, we skip from Theorem~2.2 to Lemma~2.4.}
\stepcounter{theorem}
Let $\cC_m(G) \bydef \{ \gamma \in \cC(G): |\overline \gamma| \le m\}$
be the set of polymers of size at most $m$. If $|G|=n$ the next lemma shows
$\cC_{m}(G)$ can be enumerated in time $\exp(O(m+\log n))$.
\begin{lemma}
  \label{lemPolymerCount}
  Under the assumptions of Theorem~\ref{thmPolymerApprox} we can list
  all polymers $\gamma\in\cC_{m}(G)$ in time $\exp(O(m+\log|G|))$.
\end{lemma}
\begin{proof}
  There are at most $\exp ( O(m + \log |G|))$ such polymers, as (i)
  the support of a polymer is a connected subgraph of a bounded degree
  graph, and by~\cite[Lemma~9]{BorgsChayesKahnLovasz} there are
  $\exp((O(m+\log |G|)))$ of these, and (ii) by assumption we can list
  all polymers with a given support of size at most $m$ in time
  $\exp(O(m))$.  The list can be created in time
  $\exp ( O(m + \log |G|))$ as in~\cite[Lemma
  3.4]{patel2016deterministic}
   \qedhere\end{proof}

\begin{lemma}
  \label{lemPolymerCount2}
  Under the assumptions of Theorem~\ref{thmPolymerApprox}, for any
  polymer $\gamma$ we can list all polymers $\gamma'$ such that
  $\gamma'$ is incompatible with $\gamma$ and
  $|\overline \gamma'|\le m$ in time
  $\exp(O(m+\log|\overline \gamma|))$.
\end{lemma}
\begin{proof}
  For each $v$ such that $d(v, \overline \gamma) \le 1$, we list all
  polymers of size at most $m$ containing $v$, then remove duplicates.
  As in the proof of Lemma~\ref{lemPolymerCount}, this can be done in
  time $\exp ( O(m + \log |\overline \gamma|))$.
\qedhere \end{proof}

The computation of the Ursell function of a graph on $k$
vertices by naively summing over all spanning edge sets would take
$\exp(O(k^{2}))$ time. The next lemma does better. 
\begin{lemma}
  \label{lemUrsellCompute}
  The Ursell function $\phi(H)$ can be computed in time $\exp ( O(|H|))$.
\end{lemma}
\begin{proof}
  Let $\kappa((V,A))$ denote the number of connected components of a
  graph $(V,A)$. The \emph{Tutte polynomial} of a connected graph
  $H=(V,E)$ on $k$ vertices is
  \begin{equation*}
    T_{H}(x,y) \bydef \sum_{A \subseteq E} (x-1)^{\kappa((V,A)) - 1} (y-1)^{\kappa((V,A)) +|A| -k}
  \end{equation*}
  We can express $\phi(H)$ in terms of the Tutte polynomial:
  \begin{equation*}
    \phi(H) = \sum_{A \subseteq E}\mathbf 1_{\kappa((V,A))=1} \cdot (-1)^{|A|} 
              = (-1)^{k-1} T_{H}(1,0) \,.
  \end{equation*}
  The coefficients of the Tutte polynomial $T_{H}(x,y)$ can be
  computed in time $3^k k^{O(1)}$ using an algorithm of Bj{\"o}rklund,
  Husfeldt, Kaski, and Koivisto~\cite{bjorklund2008computing}, and
  hence $T_{H}(1,0)$ can be computed in this time.
\qedhere \end{proof}

Finally, we give a simple lemma about products of weight functions.

\begin{lemma}
  \label{lemRational1}
  Let $w_1(z)$ and $w_2(z)$ be two weight functions.  If we know
  $w_1(z)$ and $w_2(z)$ up to order $m$ then we can compute the
  product $w_1(z)w_2(z)$ up to order $m$ in time $O(m^2)$.
\end{lemma}
\begin{proof}
  It is a simple calculation to express the coefficients of $w_1w_2$
  in terms of those of $w_1$ and $w_2$. This  takes $O(m^2)$
  time.
\qedhere \end{proof}

\begin{proof}[Proof of Theorem~\ref{thmPolymerApprox}]
  Let $n = |G|$ and set
  $m = \lceil \frac{\log (Cn/\eps)}{1-|z|/\del} \rceil $, where
  $C, \del$ are the constants from the hypotheses of the theorem.
  Recall the constant $\rho$ of Assumption~\ref{asAnalytic}, and let
  $m' = \lceil{m/\rho\rceil}$. Note $m' = \Theta(m)$.

  First we create a list of all polymers in $\cC_{m'}(G)$, along with
  the Taylor series coefficients of $w(\gamma,z)$ of order at most $m$
  for all $\gamma \in \cC_{m'}(G)$. These are the polymers and
  coefficients that can contribute to the order $k$ coefficients of
  the Taylor series of $\log Z(G,z)$ for $k\leq m$. The list of
  polymers can be formed in time $\exp(O(m+\log|G|))$ by
  Lemma~\ref{lemPolymerCount}, and we can compute the coefficients of
  the weight functions up to order $m$ in time $\exp(O(m)+\log |G|)$ by
  assumption. 
  Sort this list by $| \overline \gamma|$ and call the sorted list $\mathcal L: \mathcal L(j)$ is the $j$th polymer in the list.

Call a subgraph of $G$ consisting of a set of edges $A \subseteq E(A)$
so that the underlying vertices of $A$ form a connected induced
subgraph a \textit{cluster graph}.  The size of a cluster graph is its
number of vertices.  Each non-zero term of the cluster expansion
corresponds to an unordered multiset of polymers with the property
that the union of edges from these polymers (taken with multiplicity
$1$ each) form a cluster graph.

We can list all cluster graphs of $G$ of size at most $m'$ in time
$\exp ( O(m + \log |G|))$.  For each of these cluster graphs, we can
assign integers $\ge 1$ to its vertices with a sum at most $m'$ in
time $\exp(O(m))$.  Call a cluster graph with such an assignment of
integers to vertices a \emph{cluster multi-graph}. The size of the cluster
multi-graph is the sum of the integers assigned to the vertices.  
Each unordered
multiset of polymers that contributes to the cluster expansion induces
a connected cluster multi-graph.
  
Let $C_2>0$ be a uniform constant such that enumerating all polymers of
size $j$ containing a given vertex can be done in at most
$\exp(C_2 j)$ steps (such a $C_2$ exists since $G$ has bounded degree
and by the assumptions of the theorem). Next we show by induction that
we can enumerate all mutisets of polymers that induce a given
connected cluster multi-graph of size $k$ in time $\exp (C_{1}k)$ for
some $C_{1}>C_{2}$.  We
prove this by induction on $k$, with the base case $k=0$ which
trivially can be done in constant time.  Now assume that we can do
this for sizes at most $k$ in at most $\exp (C_1 k)$ elementary
computational steps.  
Given a cluster multigraph $\Gamma$ of size $k+1$, pick an arbitrary
vertex $v$ in $\Gamma$.  List all polymers $\gamma$ containing the
vertex $v$ whose edges are contained in $\Gamma$.  For each of these,
remove the vertices of $\gamma$ from $\Gamma$ (that is, decrease the
integer labels by one; a label of zero means the vertex is removed
from $\Gamma$). What remains is a possibly disconnected
cluster multigraphs, and by induction, we enumerate all multisets of
polymers that induce these cluster multigraphs. 
The time it takes to do this is at most
 \[ \sum_{j \ge 1} \exp(C_2 j) \exp( C_1(k+1-j)) = \exp(C_1 (k+1)) \sum_{j \ge 1} \exp( (C_2-C_1)j) \,. \]
 This is at most $\exp(C_1 (k+1))$ as desired provided if we choose $C_{1}=C_{2}+1$.

 Let $\mathcal L'$ denote 
 the list of all such multisets of polymers for all connected cluster
 multi-graphs of size at most $m'$. 
 Note that this list has size $\exp(O(m)+\log |G|)$ and can be constructed in the same time.
We record each such a multiset as a non-decreasing sequence of
positive integers; each 
integer corresponds with a polymer via its location in the list $\mathcal{L}$.

  Each of these multisets
  is made up of polymers from $\cC(G)$ and each has the property that
  its corresponding incompatibility graph is connected; that is, each
  corresponds to a \emph{cluster} that contributes to the sum
  in~\eqref{eqClusterExpand}. 
  Moreover, each cluster contributing
  to~\eqref{eqClusterExpand} whose weight is $z^{j}(1+O(z))$ for
  $j\leq m$ appears at least once in this list.  This is because
  Assumption~\ref{asAnalytic} implies that both a cluster of
  $\tilde m >m$ polymers and a cluster containing a polymer of size
  $\tilde m > m'$ are $0$ up to order $m$.  
  From the list $\cL'$ we can obtain a list that
  contains each possible cluster exactly once by removing all
  duplicate clusters from $\cL'$.  This takes time at most quadratic
  in the length of the list, which is $\exp(O(m + \log|G|))$.

  For each cluster in $\cL'$ we can compute its incompatibility graph
  $H$ in time $O(m^{2})$, and we can compute the Ursell function
  $\phi(H)$ in time at most $\exp(O(m))$ by
  Lemma~\ref{lemUrsellCompute}.  
  We can also compute the product of
  the weight functions of the polymers in the cluster up to order $m$:
  since we have already computed the weight functions up to order $m$
  we can do this in time $O(m^3)$ by $m$ applications of
  Lemma~\ref{lemRational1}.

  We then sum the coefficients of order $k$ over all clusters in
  $\cL'$ to obtain the coefficient of $z^k$ in the Taylor series for
  $\log Z(G,z)$ by~\eqref{eqTayfromCluster}. 
  Evaluating $T_m(G,z)$ and exponentiating gives an $\eps$-relative
  approximation to $Z(G,z)$ by Lemma~\ref{lemTaylor}.  The total
  running time of the algorithm is $\exp(O(m + \log|G|))$.
\qedhere \end{proof}

Before applying this theorem to our examples, we record a remark that
will be needed later.
\begin{remark}
  \label{rem:zero}
  In the proof of Theorem~\ref{thmPolymerApprox} we only used
    the fourth hypotheses of the theorem to guarantee the accuracy of
    the approximation $\exp[T_{m}(G,z)]$ to $Z(G,z)$. In particular,
    this hypothesis was not used in the computation of the
    coefficients of $T_{m}(G,z)$.
\end{remark}

\subsection{Examples}
\label{sec:clus-ex}

Theorem~\ref{thmPolymerApprox} allows us to recover the results of
Patel and Regts~\cite{patel2016deterministic}, and independently Harvey, 
Srivastava and Vondr\'ak~\cite{harvey2018computing}, for the hard-core model
and the results of Liu, Sinclair, and Srivastava~\cite{liu2017ising}
for the Ising model with non-zero external field. In both cases we get
an FPTAS for these models on graphs with degree at most $\Delta$. Let
us briefly justify why Theorem~\ref{thmPolymerApprox} applies.
\begin{example}[The hard-core model at low density]
  \label{ex:HC-LD2}
  Recall Example~\ref{ex:HC-LD}.  Let $\mathfrak{G}_\Delta$ be the set
  of graphs of maximum degree $\Delta$.  The first three conditions of
  Theorem~\ref{thmPolymerApprox} are straightforward to verify. For
  the fourth condition, Shearer's bound shows that $Z(G,z)\neq 0$ for
  all $|z|< \frac{(\Delta-1)^{\Delta-1}}{\Delta^\Delta}$ and
  $G\in\mathfrak{G}_\Delta$~\cite{shearer1985problem,scott2005repulsive}.
\end{example}

\begin{example}[The Ising model with free boundary conditions and an
  external field]
  \label{ex:Ising-ExF2}
  Recall Example~\ref{ex:Ising-ExF}. It suffices to approximate the
  polymer model
  \begin{equation}
    \label{eq:IsingExt}
    Z(G,z) = z^{|G|} e^{-\beta|E(G)|} Z_G(\beta,z),
\end{equation}
  and by swapping the roles of the $+1$ and $-1$ spins, it suffices
    to consider $|z|<1$.
The first three conditions of Theorem~\ref{thmPolymerApprox} are
easily verified:
\begin{itemize}
\item $Z(G,z)$ is a polynomial of degree $2|G|$ in $z$.
\item Polymers correspond to connected induced subgraphs of
  $G$. We can compute the weight functions of all polymers up to order
  $m$ as follows. First, list all connected induced subgraphs of $G$
  of size at most $m$; there are at most $\exp(O(m+ \log |G|))$ of
  these and the list can be constructed in this time by
  Lemma~\ref{lemPolymerCount}. For each connected subgraph, the weight
  function can be computed in time $O(m)$ as it suffices to count
  $|\overline\gamma|$ and $|\partial_{e}\overline\gamma|$.
\item For each connected induced subgraph $G'$ of $G$ there is exactly
  one polymer with support $G'$.
\end{itemize}
The fourth condition is provided by the Lee--Yang
theorem~\cite{lee1952statistical}: for any $G$ and any
$\beta>0$, $Z_G(\beta,z) \ne 0$ if $|z| < 1$. By \eqref{eq:IsingExt},
this implies $Z(G,z)\neq 0$ for $|z|<1$ as well.
\end{example}

\subsection{A generalization}
\label{secPolyGeneral}

We can generalize the definitions and results above, and this
generalization will be useful in what follows.  Let
$\cS \subseteq \cC(G)$, and let $\cG(\cS)$ be the collection of all
finite sets of polymers from $\cC(\cS)$ that are pairwise compatible,
including the empty set of polymers.  Abusing notation, we define
\begin{align}
\label{eqPolymerPartGen}
Z(\cS,z) \bydef \sum_{\Gamma \in \cG(\cS)} \prod_{\gamma \in \Gamma} w(\gamma,z) \,. 
\end{align}

If we know $Z(G,z)$ has a zero-free disk about the origin, then we can
efficiently approximate $Z(\cS,z)$ for any $\cS \subseteq \cC(G)$.

\begin{lemma}
  \label{thmPolymerApproxGen}
  Fix $\Delta$ and let $\mathfrak G$ be a set of graphs of degree at
  most $\Delta$.  Suppose:
  \begin{itemize}
  \item There is a constant $C$ so that $Z(\cS,z)$ is a polynomial in
    $z$ of degree at most $C |G|$ for all $G \in \mathfrak G$ and all
    $\cS \subseteq \cC(G)$.
  \item We can compute $w(\gamma,z)$ up to order $m$ for all
    $\gamma \in \cC(G)$ in time $\exp(O(m+\log |G|))$.
  \item For every connected subgraph $G' $ of every
    $G \in \mathfrak G$, we can list all polymers $\gamma $ with
    $\overline \gamma = G'$ in time $\exp(O(|G'|))$.
  \item There exists $\del >0$ so that for all $|z| < \del$ and all
    $G \in \mathfrak G$, the cluster expansion \eqref{eqClusterExpand} 
 is absolutely convergent.
  \end{itemize}
  Then for every $z$ with $|z|<\del$, there is an FPTAS for $Z(\cS,z)$
  for all $G \in \mathfrak G $ and all $\cS \subseteq \cC(G)$.
\end{lemma}

The proof is a repetition of the proof of
Theorem~\ref{thmPolymerApprox} together with one observation:
 for all $\cS \subseteq \cC(G)$,
$Z(\cS,z) \ne 0$ for $|z|<\del$.  This follows since the cluster
expansion for $\log Z(G,z)$ is absolutely convergent, and the cluster
expansion for $\log Z(\cS,z)$ is a subseries so it too must be
absolutely convergent.

\subsection{Related results}
\label{sec:related-results}

The algorithm of Theorem~\ref{thmPolymerApprox} has strong similarity
with the algorithms used in~\cite{patel2016deterministic}
and~\cite{liu2017ising}.  Both of these results use truncation of the
Taylor series for $\log Z$ and the fact that the Taylor series are in
some sense supported on connected graphs.
Theorem~\ref{thmPolymerApprox} makes this notion of connectedness
explicit and illustrates the connection to the cluster expansion. As a
consequence our result uses analyticity of the weight functions, while
the other approaches use more algebraic methods in combination
with the Newton identities (see \eqref{eqPk1} and \eqref{eqEk1} below).  

In the next section we will apply Theorem~\ref{thmPolymerApprox} to
more sophisticated contour models.  It is likely possible to apply the
approach of~\cite{patel2016deterministic} to contour models as
well. We have elected to develop the cluster expansion approach as it
gives us access to well-developed criteria for verifying the fourth
condition of Theorem~\ref{thmPolymerApprox}, as will be explained in
Section~\ref{secZeroFree}.

A more careful analysis of our algorithm allows one to recover the
result from~\cite{patel2017computing} saying that one compute the
number of independent sets of size $m$ in a bounded degree graph of
order $n$ in time $O(nc^m)$.

\section{Contour models}
\label{secContourModels}

A more sophisticated version of a polymer model is a \emph{contour
  model}, and for this we specialize to $\Z^d$, $d\geq 2$. Our setup
will be an amalgamation of those in~\cite{borgs1989unified}
and~\cite[Chapter 7]{friedli2017statistical}. The main result is
Theorem~\ref{thmContourModelCount}. We give examples of contour model
representations of spin models in Section~\ref{sec:examples}.

\subsection{Contour models}
\label{sec:contour-models}

Fix a finite set of spins $\Omega$, and let $\Xi$ be a finite set of
\emph{ground states}.  In spin models ground states correspond to
periodic assignments of spins to $\Z^d$ that minimize energy, e.g.,
monochromatic configurations for the Potts model or the all even/all
odd occupied configurations for the hard-core model, but at this level
of generality they are just labels.

A contour $\gamma$ is a pair $(\overline \gamma, \omega_{\overline \gamma})$;
the \emph{support} $\overline \gamma $ is a finite subset of $\Z^d$
connected under the $d_{\infty}$ distance and
$\omega_{\overline \gamma}\colon \overline\gamma \to \Omega$ is an
assignment of spins to the vertices of $\overline \gamma$.  The
support $\overline \gamma$ of a contour partitions
$\Z^d \setminus \overline \gamma$ into maximal connected components,
and in what follows we denote them by $A_0, A_1, \dots, A_t$, and we
assume $A_0$ is the unique infinite component.  Let
$\text{ext} \gamma \bydef A_0$ denote the \emph{exterior} of $\gamma$
and $\text{int} \gamma \bydef \bigcup_{i=1}^t A_i$ denote the
\emph{interior} of $\gamma$.

A \emph{contour model} is a set of contours $\cC$, a \emph{surface
  energy} $\| \gamma \|\in\mathbb{N}$ for each contour, and a
\emph{labeling function} $\text{lab}_\gamma(\cdot)$ for each
contour. The labeling function $\text{lab}_\gamma$ is a map from the
collection of connected components $\{A_0, \dots, A_t\}$ to $\Xi$, the
set of ground states. We will assume the labelling function is
determined by the contour $\gamma$.

We will always make two basic assumptions on contour models. The first
is about the computability of contours and their surface energies.
\begin{assumption}
  \label{assumCompute}
  For every contour $\gamma$ we can both determine if $\gamma\in\cC$
  and compute the labelling function $\text{lab}_{\gamma}(\cdot)$
  in time $\exp(O(|\overline\gamma|))$. 
  Moreover, for $\gamma\in\cC$ we can compute $\|\gamma\|$ in time
  $\exp(O(|\overline\gamma|))$. 
\end{assumption}

Our second assumption relates the surface energy to the support of a
contour. In applications the upper bound is typically trivial, while the lower
bound is non-trivial and is known as the \emph{Peierls' condition}.
\begin{assumption}
  \label{asSurface} 
  There are constants $\rho,C>0$ such that for all $\gamma\in\cC$
  the surface energy $\| \gamma \|$ is a positive integer satisfying
  the bound
  \begin{align}
    \label{eqContourPeierls}  
    \rho |\overline \gamma| \le \|\gamma \| \le C |\overline \gamma|. 
  \end{align}
\end{assumption}

\subsection{Partition functions of contour models}
\label{sec:part-funct-cont}

There are natural partition functions associated to contour models,
and to introduce them we need a few more definitions. Two contours
$\gamma$ and $\gamma'$ are \emph{compatible} if
$d_\infty( \overline \gamma, \overline \gamma') >1$.  A contour
$\gamma$ is of \emph{type} $\varphi$ if its exterior is labelled
$\varphi$.  The union of all interior regions of $\gamma$ with label
$\varphi$ wil be denoted
\begin{equation*}
  \text{int}_{\varphi} \gamma \bydef \bigcup_{ i \ge1 : \text{lab}_\gamma(A_i) =\varphi} A_i. 
\end{equation*} 
Let $\Gamma$ be a set of compatible contours.
\begin{enumerate}
\item We say $\gamma\in\Gamma$ is \emph{external} if
  $\overline \gamma \subset \text{ext} \gamma'$ for all
  $\gamma ' \in \Gamma$, $\gamma' \ne \gamma$,
\item We say $\Gamma$ is \emph{matching and of type $\varphi$} if (i)
  all external contours have type $\varphi$, and (ii) either
  $|\Gamma|=1$, or for each external contour $\gamma \in \Gamma$ and
  ground state $\varphi'$ the subcollection of contours
  $\Gamma'\subset\Gamma$ whose support is contained in
  $\text{int}_{\varphi'}\gamma$ is matching and of type $\varphi'$.
\end{enumerate}

Let $\cC^\varphi\subset\cC$ be the set of all contours of type $\varphi$, and
for a region $\Lam \subset \Z^d$, let $\cC^\varphi(\Lam)$ be the set
of all contours $\gamma$ of type $\varphi$ so that
$d_{\infty}(\overline \gamma,\Lam^c)>1$. We say these contours are
\emph{in $\Lam$}.  Let $\cG^\varphi_{\text{match}}(\Lam)$ be the
collection of all sets of pairwise compatible contours in
$\Lam$ that are matching and of type $\varphi$. Define
\begin{equation}
  \label{eqContourMatching}
  Z^\varphi(\Lam,z) 
  \bydef 
  \sum_{\Gamma \in \cG^{\varphi}_{\text{match}}(\Lam)} 
  \prod_{\gamma \in \Gamma}  z^{\| \gamma \|}  \,.
\end{equation}
We call this the \emph{contour representation} of the partition
function.  It is clear from~\eqref{eqContourMatching} that $Z(\Lam,z)$
is a polynomial in $z$ with constant term $1$, and by
Assumption~\ref{asSurface} it is of degree at most $C|\Lam|$.
See Figure~\ref{fig:cont-full} for a schematic representation.

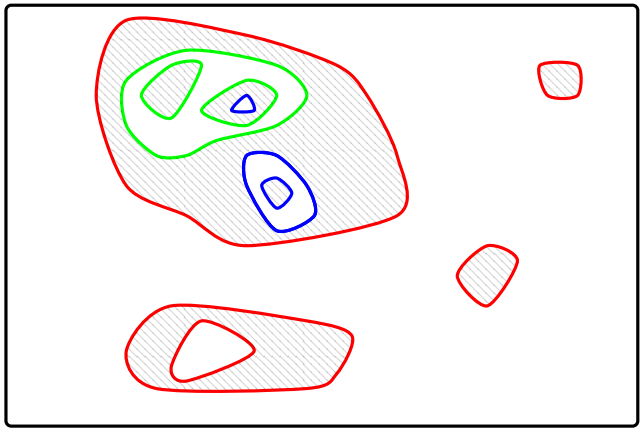
\begin{figure}
  \centering
  \begin{tikzpicture}[scale=.4]
    \draw[lam, rounded corners=2] (-1,1) rectangle (20,15);

    \draw[name path = A] [br] plot [smooth cycle] coordinates {(2,12)
      (3,14.5) (7,14) (10,13) (11,12) (12,10) (12,8) (7,7) (5,8)
      (3,9)};

    \draw[name path = B] [boundary] plot [smooth cycle] coordinates {
      (7,9) (7,10) (8,10) (9,9) (9.25,8) (8,7.5)} plot [smooth cycle]
    coordinates {(3,12.5) (5,13.5) (8,13) (9,12) (8,11) (6,10.5)
      (5,10) (4,10) (3,11)};
 
    \draw [bb] plot [smooth cycle] coordinates { (7,9) (7,10) (8,10)
      (9,9) (9.25,8) (8,7.5)};
      
    \draw [bg] plot [smooth cycle] coordinates {(3,12.5) (5,13.5)
      (8,13) (9,12) (8,11) (6,10.5) (5,10) (4,10) (3,11)};
     
    \tikzfillbetween[of=A and B]{pattern=north west lines,
      opacity=0.4};

      \draw [boundaryfill] plot [smooth cycle] coordinates {(3.5,12) (4.5,13)
        (5.5,13)  (4.5,11.25)};
      \draw [bg] plot [smooth cycle] coordinates {(3.5,12) (4.5,13)
        (5.5,13)  (4.5,11.25)};

      \draw[name path = E] [bg] plot [smooth cycle] coordinates
      {(5.5,11.5)  (7,12.5) (8,12) (7,11)};
      \draw[name path = F] [bb] plot [smooth cycle] coordinates
      {(6.5,11.5) (7,12) (7.25,11.5)};
      \tikzfillbetween[of=E and F]{pattern=north west lines, opacity=0.4};

     \draw [bb] plot [smooth cycle] coordinates { (7,9)
         (7,10) (8,10) (9,9) (9.25,8) (8,7.5)};

       \draw[boundaryfill] plot [smooth cycle] coordinates {(7.5,9)
         (8,9.25) (8.5,8.75) (8,8.25)}; 
       \draw[bb] plot
       [smooth cycle] coordinates {(7.5,9)
         (8,9.25) (8.5,8.75) (8,8.25)};
      
      \draw [boundaryfill] plot [smooth cycle]
      coordinates {(17,12) (16.75,13) (18,13) (18,12)};

      \draw [br] plot [smooth cycle] coordinates {(17,12)
        (16.75,13) (18,13) (18,12)};

      \draw[boundaryfill] plot [smooth cycle] coordinates {(14,6) (15,7)
        (16,6.5) (15,5)}; 
      \draw[br] plot [smooth cycle] coordinates {(14,6) (15,7)
        (16,6.5) (15,5)};

      \draw[name path = C] [br] plot [smooth cycle] coordinates
      {(3,3.5) (4.5,5) (9,4.5) (10.5,4) (10,2.75) (9,2.25) (4,2.25)};

      \draw[name path = D] [br] plot [smooth cycle] coordinates
      { (4.5,3) (5.5,4.5) (7.25,3.5) (5,2.5)};

      \tikzfillbetween[of=C and D]{pattern=north west lines,
        opacity=0.4};
  \end{tikzpicture}
  \caption{A schematic representation of 
the contour partition function. The boundary condition is red, the shaded set
    indicates the contours, and colors indicate the labels of the
    contours. Each connected component of the shaded set is a distinct
    contour.}
  \label{fig:cont-full}
\end{figure}

Let $\cG^\varphi_{\text{ext}}(\Lambda)$ be the collection of all sets
$\Gamma$ of contours from $\cC^\varphi(\Lam)$ so that every
$\gamma \in \Gamma$ is external. By fixing the outer contours
in~\eqref{eqContourMatching} and summing over all possible contours in
their interior, we obtain the following inductive representation of
$Z^{\varphi}$: 
\begin{equation}
  \label{eqContourPartition}
  Z^\varphi(\Lam,z) 
  = 
  \sum_{\Gamma \in \cG^{\varphi}_{\text{ext}}(\Lam)}  \prod_{\gamma \in \Gamma} 
  \left ( z^{\| \gamma \|} \prod_{\varphi' \in \Xi} 
    Z^{\varphi '} (\text{int}_{\varphi'} \gamma ,z) \right) \,,
\end{equation}
which we call the \emph{outer contour representation}.  In obtaining
\eqref{eqContourPartition} we have used that compatibility implies
that the distance between contours is at least two, and hence any
contour $\gamma$ of type $\varphi$ with
$\overline\gamma\subset\text{int}_{\varphi}\gamma'$ belongs to
$\cC^{\varphi}(\text{int}_{\varphi}\gamma')$. The base case in
\eqref{eqContourPartition} is a \emph{thin region} $\Lam$, i.e., one
so that $\cG^{\varphi}_{\text{ext}}(\Lam) = \{ \emptyset \}$, in which
case $Z^\varphi(\Lam,z)=1$. See Figure~\ref{fig:Outer} for a schematic representation.

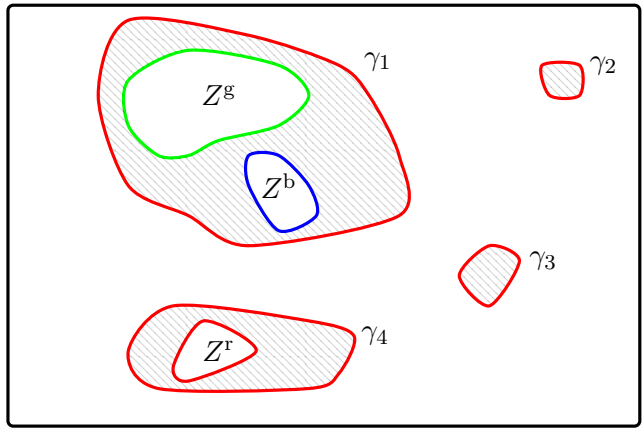
\begin{figure}
  \centering
  \begin{tikzpicture}[scale=.4]
    \draw[lam, rounded corners=2] (-1,1) rectangle (20,15);
    
      \node at (10.5,12.5) [above right] {$\gamma_{1}$};
    
      \draw[name path = A] [br] plot [smooth cycle] coordinates {(2,12) 
        (3,14.5)  (7,14) (10,13) (11,12) (12,10) (12,8) 
        (7,7) (5,8) (3,9)};

      \draw[name path = B] [boundary] plot [smooth cycle] coordinates { (7,9)
         (7,10) (8,10) (9,9) (9.25,8) (8,7.5)} plot [smooth cycle]
      coordinates {(3,12.5) (5,13.5) (8,13) (9,12) (8,11)
        (6,10.5) (5,10) (4,10) (3,11)};
 
     \draw [bb] plot [smooth cycle] coordinates { (7,9)
         (7,10) (8,10) (9,9) (9.25,8) (8,7.5)};
      
      \draw [bg] plot [smooth cycle]
      coordinates {(3,12.5)  (5,13.5) (8,13) (9,12) (8,11)
        (6,10.5) (5,10) (4,10) (3,11)};
     
      \tikzfillbetween[of=A and B]{pattern=north west lines,
        opacity=0.4};
      
      \node at (18,13) [right] {$\gamma_{2}$};
   
      \draw [boundaryfill] plot [smooth cycle]
      coordinates {(17,12) (16.75,13) (18,13) (18,12)};

      \draw [br] plot [smooth cycle] coordinates {(17,12)
        (16.75,13) (18,13) (18,12)};

      \node at (16,6.5) [right] {$\gamma_{3}$};

      \draw[boundaryfill] plot [smooth cycle] coordinates {(14,6) (15,7)
        (16,6.5) (15,5)}; 
      \draw[br] plot [smooth cycle] coordinates {(14,6) (15,7)
        (16,6.5) (15,5)};

      \node at (12,4) [left] {$\gamma_{4}$};

      \draw[name path = C] [br] plot [smooth cycle] coordinates
      {(3,3.5) (4.5,5) (9,4.5) (10.5,4) (10,2.75) (9,2.25) (4,2.25)};

      \draw[name path = D] [br] plot [smooth cycle] coordinates
      { (4.5,3) (5.5,4.5) (7.25,3.5) (5,2.5)};

      \tikzfillbetween[of=C and D]{pattern=north west lines,
        opacity=0.4};

      \node (zg) at (6,12) {$Z^{\text{g}}$};
      \node (zb) at (8,9) {$Z^{\text{b}}$};
      \node (zr) at (6,3.5) {$Z^{\text{r}}$};
  \end{tikzpicture}
  \caption{A schematic representation of a term of the outer contour
    partition function. The boundary condition is red, the shaded set
    indicates the contours, and colors indicate the labels of the
    contours. The contours $\gamma_{1}$ and $\gamma_{4}$ have
    interiors with interior partition functions $Z^{\varphi}$ for
    $\varphi\in\{r,b,g\}$, while $\gamma_{2}$ and $\gamma_{3}$ do not
    have interiors.}
  \label{fig:Outer}
\end{figure}

There are well-known methods to convert discrete statistical physics
models into contour
representations~\cite[Chapter~7]{friedli2017statistical}. For the
convenience of the reader we carry this out in
Section~\ref{sec:examples} for the Potts and hard-core models.

\subsection{Approximating the contour model partition function}
\label{secapproxcontour}

Our main theorem is an algorithm to approximate the contour model
partition function.
\begin{theorem}
  \label{thmContourModelCount}
  Fix $d \ge 2$ and $\varphi\in\Xi$, and suppose that:
  \begin{itemize}
  \item The contour model satisfies Assumptions~\ref{assumCompute}
    and~\ref{asSurface}.
  \item There exists $\del >0$ so that for 
    $|z| < \del$ and all regions $\Lam \subset \Z^d$, 
    $Z^\varphi(\Lam, z) \ne 0$.
  \end{itemize}
  Then for every $z$ with $|z|<\del$, there is an FPTAS for
  $Z^\varphi(\Lam,z)$ for all regions $\Lam \subset \Z^d$.
\end{theorem}

To prove this theorem we will view the outer contour model given
by~\eqref{eqContourPartition} as a polymer model. To make this
precise, define the weight function of $\gamma$ by
\begin{equation}
  \label{eqouterweights}
  w^{\text{ext}}(\gamma,z) =  z^{\| \gamma \|} \prod_{\varphi' \in \Xi} 
                             Z^{\varphi '} (\text{int}_{\varphi'} \gamma ,z) \,.
\end{equation}
The outer contour representation can be rewritten as
\begin{equation}
  \label{eqouteras}
  Z^\varphi(\Lam,z) 
  = \sum_{\Gamma \in \cG^{\varphi}_{\text{ext}}(\Lam)}  \prod_{\gamma \in \Gamma} 
                      w^{\text{ext}}(\gamma,z)\,,
\end{equation}
which matches the form of~\eqref{eqPolymerPartition}, except for the
fact that the compatibility condition for external contours is not the
notion of compatibility that was used for polymer models. We will
address this momentarily. Note that by construction
$w^{\text{ext}}(\gamma,z)$ is a polynomial. By
Assumption~\ref{asSurface} $\|\gamma\|\geq \rho |\overline\gamma|$,
and hence Assumption~\ref{asAnalytic} is satisfied for these weights.

Two contours $\gamma,\gamma'$ are \emph{mutually external} if they are
compatible, $\overline\gamma\subset\text{ext}\gamma'$, and
$\overline\gamma'\subset \text{ext}\gamma$. This mean neither contour
lies in the interior of the other. 
Let
\begin{equation}
  \label{eqcov}
  \cov(\gamma) = \overline \gamma \cup \bigcup_{\varphi \in \Xi}
  \text{int}_\varphi \gamma \, .
\end{equation}
Then two contours $\gamma, \gamma'$ of type $\varphi$ are mutually
external if $d_{\infty} (\cov(\gamma), \cov(\gamma'))>1$. We will use
mutual externality as the notion of compatibility for the outer
contour model; this replaces the notion of compatibility that was used
for polymer models. The cluster expansion \eqref{eqClusterExpand}
holds for this notion of compatibility~\cite{friedli2017statistical},
and the proof of Theorem~\ref{thmPolymerApprox} goes through unchanged
for this notion of compatibility given the following replacement for
Lemma~\ref{lemPolymerCount2}.

\begin{lemma}
  \label{lemOuterCompat}
  Suppose it is possible to determine if $\gamma\in\cC(G)$ in time
    $\exp(O(|\overline\gamma|))$. 
    Then for any contour $\gamma$ we can list all contours $\gamma'\in\cC(G)$
    such that $\gamma, \gamma'$ are not mutually external and
    $|\overline \gamma'|\le m$ in time
    $\exp(O(m+\log|\overline \gamma|))$.
\end{lemma}
\begin{proof}
  We need to list all $\gamma'\in\cC(G)$ of size at most $m$ so that
  $d_\infty(\cov(\gamma),\cov(\gamma'))\le 1$. 

  For each $v$ such that $d_\infty(v, \cov(\gamma))\le 1$, and each
  $u$ such that $d_{\infty}(u,v) \le m$, we list all
  $d_\infty$-connected subgraphs of size at most $m$ containing $u$
  and all assignments of spins from $\Omega$ to these
  subgraphs. This takes time $\exp(O(m))$ by
    \cite[Lemma~9]{BorgsChayesKahnLovasz}. By hypothesis we can
  determine which of these contours are in $\cC$ in time $\exp(O(m))$,
  and hence for each $v,u$ this list can be constructed in time $\exp
  (O(m))$.

  There are at most
  $2 \cdot 3^d | \overline
  \gamma|^{d/(d-1)}$~\cite[Lemma~7.28]{friedli2017statistical} such
  vertices $v$, and for each $v$ at most $(2m+1)^d$ vertices $u$ , and
  so the combination of all lists can be constructed in time
  $\exp (O(m +\log|\overline \gamma| ))$.  Finally for each $\gamma'$
  in the list, we check if
  $d_\infty(\cov(\gamma),\cov(\gamma'))\le 1$. This can be done in
  time polynomial in $|\overline \gamma | \cdot |\overline \gamma'|$.
\qedhere \end{proof}

Theorem~\ref{thmContourModelCount} will follow directly from
Theorem~\ref{thmPolymerApprox} if we can verify the second hypothesis,
i.e., if we can prove that the the weight functions
$w^{\text{ext}}(\gamma,z)$ can be computed up to order $m$ for all
$\gamma \in \cC^\varphi(\Lam)$ in time $\exp(O(m + \log |\Lam|))$.

\begin{lemma}
  \label{lemContourWeightLemma}
  Under the assumptions of Theorem~\ref{thmContourModelCount}, we can
  compute the weight functions $w^{\text{ext}}(\gamma)$ up to order
  $m$ for all $\varphi \in \Xi$ and all contours
  $\gamma \in \cC_m^\varphi(\Lam)$ in time $\exp(O(m+ \log |\Lam|))$.
\end{lemma}

Before we prove Lemma~\ref{lemContourWeightLemma} we need one useful
fact, the Newton identities.  Let $Z(z) = 1+ \sum_{k=1}^N e_k z^k$ be
a polynomial, and let $p_1, p_2, \dots $ be the normalized coefficients of the
Taylor series $\log Z(z) = \sum_{k \ge 1} -p_k/k z^k$ around $0$. The
Newton identities imply the coefficients $p_{i}$ can be expressed
inductively in terms of the coefficients $e_{i}$, and vice-versa (cf.\
\cite{patel2016deterministic}):
\begin{align}
  \label{eqPk1}
  p_k &= - k e_k -\sum_{j=1}^{k-1}e_jp_{k-j} \,, \\
  \label{eqEk1}
  e_k &= -\frac{1}{k} \sum_{j=0}^{k-1}
        e_j p_{k-j} \, .
\end{align}
From this it follows that we can compute $Z$ up to order $m$ in
time polynomial in $m$ given the Taylor series coefficients of
$\log Z$ up to order $m$ and vice versa.

\begin{proof}[Proof of Lemma~\ref{lemContourWeightLemma}]
  We compute the weight functions inductively.  Let
  \begin{align*}
    \cC_m^{\varphi}(\Lambda) &\bydef \left\{ \gamma \in
      \cC^\varphi(\Lam)  :|\overline \gamma| \le m \right \}, \quad
    \text{and}\\
        \cC_m(\Lam) &\bydef \bigcup_{ \varphi' \in \Xi}  \cC_m^{\varphi'}(\Lambda) \,.
  \end{align*}
  
  We first give a polynomial-time algorithm to list and order $\cC_m$
  such that if $\overline \gamma$ lies in the interior of $\gamma'$
  then $\gamma$ comes before $\gamma'$ in the ordering. In particular,
  the contours with thin interiors are at the front of the order. To
  do this we note that by Lemma~\ref{lemPolymerCount} (using
  Assumption~\ref{assumCompute} in place of the third hypothesis of
  Theorem~\ref{thmPolymerApprox}) we can list $\cC_{m}(\Lam)$ in time
  $\exp(O(m+\log|\Lam|))$. For each $\gamma$ we can determine the
  components of $\Lam\setminus\overline\gamma$ in time $|\Lam|$ by
  greedily growing the components of the complement. We can then
  decide how to order a pair $\{\gamma, \gamma'\}$ by checking if each
  $y\in\overline\gamma$ is contained in a single interior component of
  $\gamma'$ or not and vice versa; this takes time $O(|\Lam|m)$. Doing
  this for each pair of contours can be done in time quadratic in the
  length of the list, and hence the list can be ordered in time
  $\exp(O(m+\log|\Lam|))$.

  Given the ordered list, we will compute the weight functions
  $w^{\text{ext}}(\gamma)$ in order. The base cases are the contours
  with thin interiors for which
  $w^{\text{ext}}(\gamma) = z^{\|\gamma\|}$.  By
  Assumption~\ref{assumCompute} these can each be computed in time in
  $\exp(O(|\overline \gamma|)=\exp(O(m))$.

  Now suppose we have computed the weight functions to order $m$ for
  every contour $\gamma'$ that precedes $\gamma$ in the list.  Then we
  can compute
  \begin{equation*}
    w^{\text{ext}}(\gamma,z) = z^{\|\gamma\|} \prod_{\varphi' \in \Xi
    }    Z^{\varphi'}(\text{int}_{\varphi'}(\gamma),z )    
  \end{equation*}
  as follows. The surface energy can be computed in time
  $\exp(O(|\overline \gamma|))$ by Assumption~\ref{assumCompute}.
  Each factor $ Z^{\varphi'}(\text{int}_{\varphi'}(\gamma),z ) $ is a
  polynomial in $z$ whose first $m$ coefficients can be computed in
  time $\exp(O(m+ \log
  |\overline\gamma|))$ 
  as follows. Recalling Remark~\ref{rem:zero}, the proof of
  Theorem~\ref{thmPolymerApprox} (with Lemma~\ref{lemOuterCompat}
  taking the place of Lemma~\ref{lemPolymerCount2}) shows we can
  compute the first $m$ coefficients of the Taylor series for
  $\log Z^{\varphi'}(\text{int}_{\varphi'}(\gamma),z)$ in the claimed
  time. The conditions of the theorem are satisfied since we have
  already written down to order $m$ the weight function of any contour
  that can appear in the interior of $\gamma$. We can then use
  the Newton identities~\eqref{eqEk1} to compute the coefficients of
  $ Z^{\varphi'}(\text{int}_{\varphi'}(\gamma),z ) $ from the Taylor
  series coefficients of
  $\log Z^{\varphi'}(\text{int}_{\varphi'}(\gamma),z )$.  Multiplying
  these factors together, of which there are at most
    $\exp(O(\log |\overline\gamma|))$, and applying
  Lemma~\ref{lemRational1} shows that we can compute
  $w^{\text{ext}}(\gamma)$ to order $m$ in time
    $\exp(O(m+\log|\overline\gamma|))$.

  The time to compute each weight function to order $m$ is therefore
  at most $\exp (O(m+ \log |\Lam|))$, and so the total time to compute
  all weight functions is at most $\exp (O(m+ \log |\Lam|))$ as well.
\qedhere \end{proof}

\begin{proof}[Proof of Theorem~\ref{thmContourModelCount}]
  We apply Theorem~\ref{thmPolymerApprox} with the class of bounded
  degree graphs $\mathfrak G$ being subgraphs of $\Z^d$ with the
  $d_\infty$-distance. The first two hypotheses of the theorem are
  true by the remarks following~\eqref{eqContourMatching} and
  \eqref{eqouteras} and Lemma~\ref{lemContourWeightLemma}. The third
  and fourth hypotheses are the first part of
  Assumption~\ref{assumCompute} and an assumption of the Theorem,
  respectively.
\qedhere \end{proof}

\subsection{A slight generalization}
\label{secContourGen}

As in Section~\ref{secPolyGeneral} we generalize the definitions and
results slightly.  We will use this generalization in the sampling
algorithm of Section~\ref{secSample}.

Let $\cS \subseteq \cC^\varphi(\Lam)$ for some region $\Lam$.  Then
define $\cG^\varphi_{\text{ext}}(\cS)$ as the collection of all sets of
compatible and mutually external contours  from $\cS$.  Define
\begin{equation}
  \label{eqContourPartGenEx}
 Z^\varphi(\cS,z) \bydef 
  \sum_{\Gamma \in \cG^{\varphi}_{\text{ext}}(\cS)} \prod_{\gamma \in \Gamma}  
  w^\text{ext}(\gamma,z) \,.
\end{equation}
Our approximate counting algorithm extends to this generalization.

\begin{lemma}
  \label{thmContourModelCountGen}
  Fix $d \ge 2$ and suppose the following:
  \begin{itemize}
  \item The contour model satisfies Assumptions~\ref{assumCompute}
    and~\ref{asSurface}.
  \item There exists $\del >0$ so that for all $|z| < \del$, all
    regions $\Lam \subset \Z^d$, and all $\varphi \in \Xi$,
    the cluster expansion for $\log Z^\varphi(\Lam, z)$ converges absolutely.
  \end{itemize}
  Then for every $z$ with $|z|<\del$, there is an FPTAS for
  $Z^\varphi(\cS,z)$ for all regions $\Lam \subset \Z^d$ and all
  $\cS \subseteq \cC(\Lam)$.
\end{lemma}
As in Section~\ref{secPolyGeneral}, it is enough to observe that
absolute convergence of the cluster expansion for $\log Z^\varphi(\Lam,z)$
implies absolute convergence of the cluster expansion for
$\log Z^\varphi(\cS,z)$.

\subsection{Examples}
\label{sec:examples}
In this section we introduce the contour representations that will be
used in the proofs of our main theorems.

\begin{example}[The ferromagnetic Potts model]
  \label{ex:FP-Contour}
  For the ferromagnetic Potts model with no external field the set
  of ground states is the set of spins (or \emph{colors})
  $\Xi = \Omega =[q]$.  Recall the padded monochromatic boundary
  conditions from Section~\ref{secPottsIntro}: for a region $\Lam$ and
  a color $\varphi \in [q]$, the set of allowed configurations is
\begin{equation*}
  \Omega_\Lam^\varphi = \{ \omega \in [q]^\Lam : \omega_i =\varphi
  \,\forall\, i\,\mathrm{s.t.\ } d_{\infty}(i, \Lam^c) \le 2 \} .
\end{equation*} 
We say a vertex $i \in \Lam$ is \emph{correct} with respect to
$\omega \in \Omega_\Lam^\varphi$ if there exists $\varphi' \in [q]$ so
that $\omega_j = \varphi'$ for all $j \in \Lam$ such that
$d_\infty(i,j) \le 1$; that is, $i$ and its $d_\infty$ neighbors all
receive the same color.  All other vertices of $\Lam$ are
\emph{incorrect} with respect to $\omega$.  The \emph{boundary}
$\Gamma(\omega)$ is the set of all incorrect vertices with respect to
$\omega$.  See Figure~\ref{fig:bdry}. Each connected component (with
respect to the $d_{\infty}$ distance) of $\Gamma(\omega)$ defines the
support $\overline\gamma$ of a contour $\gamma$, and $\omega_{\gamma}$
is the restriction of $\omega$ to $\gamma$.  By the definition of
$ \Omega_\Lam^\varphi$ we have
$d_{\infty}(\overline\gamma,\Lam^{c})>1$ for all contours.

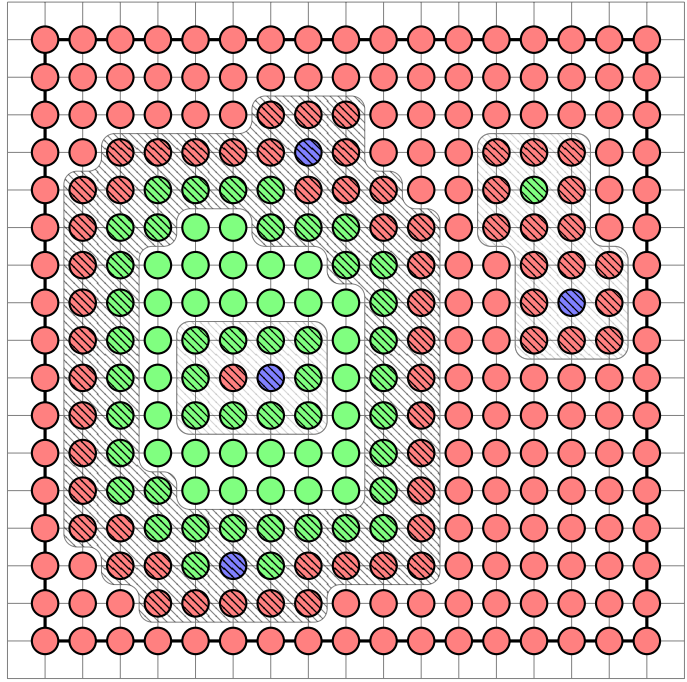
\begin{figure}
  \centering
  \begin{tikzpicture}[scale=.5]
    \draw[boundary] (-1,-1) grid (17,17);
    \draw[lam] (0,0) -- (0,16) -- (16,16) -- (16,0) -- (0,0);

    \node[rv] at (0,0) {};
    \node[rv] at (0,1) {};
    \node[rv] at (0,2) {};
    \node[rv] at (0,3) {};
    \node[rv] at (0,4) {};
    \node[rv] at (0,5) {};
    \node[rv] at (0,6) {};
    \node[rv] at (0,7) {};
    \node[rv] at (0,8) {};
    \node[rv] at (0,9) {};
    \node[rv] at (0,10) {};
    \node[rv] at (0,11) {};
    \node[rv] at (0,12) {};
    \node[rv] at (0,13) {};
    \node[rv] at (0,14) {};
    \node[rv] at (0,15) {};
    \node[rv] at (0,16) {};

    \node[rv] at (1,0) {};
    \node[rv] at (1,1) {};
    \node[rv] at (1,2) {};
    \node[rvb] at (1,3) {};
    \node[rvb] at (1,4) {};
    \node[rvb] at (1,5) {};
    \node[rvb] at (1,6) {};
    \node[rvb] at (1,7) {};
    \node[rvb] at (1,8) {};
    \node[rvb] at (1,9) {};
    \node[rvb] at (1,10) {};
    \node[rvb] at (1,11) {};
    \node[rvb] at (1,12) {};
    \node[rv] at (1,13) {};
    \node[rv] at (1,14) {};
    \node[rv] at (1,15) {};
    \node[rv] at (1,16) {};

    \node[rv] at (2,0) {};
    \node[rv] at (2,1) {};

    \node[rvb] at (2,2) {}; 
    \node[rvb] at (2,3) {}; 
    \node[gvb] at (2,4) {};
    \node[gvb] at (2,5) {};
    \node[gvb] at (2,6) {};
    \node[gvb] at (2,7) {};
    \node[gvb] at (2,8) {};
    \node[gvb] at (2,9) {};
    \node[gvb] at (2,10) {};
    \node[gvb] at (2,11) {};
    \node[rvb] at (2,12) {};
    \node[rvb] at (2,13) {};
    \node[rv] at (2,14) {};

    \node[rv] at (2,15) {};
    \node[rv] at (2,16) {};

    \node[rv] at (3,0) {};
    \node[rvb] at (3,1) {};

    \node[rvb] at (3,2) {};
    \node[gvb] at (3,3) {};
    \node[gvb] at (3,4) {};
    \node[gv] at (3,5) {};
    \node[gv] at (3,6) {};
    \node[gv] at (3,7) {};
    \node[gv] at (3,8) {};
    \node[gv] at (3,9) {};
    \node[gv] at (3,10) {};
    \node[gvb] at (3,11) {};
    \node[gvb] at (3,12) {};
    \node[rvb] at (3,13) {};
    \node[rv] at (3,14) {};

    \node[rv] at (3,15) {};
    \node[rv] at (3,16) {};

    \node[rv] at (4,0) {};
    \node[rvb] at (4,1) {};

    \node[gvb] at (4,2) {};
    \node[gvb] at (4,3) {};
    \node[gv] at (4,4) {};
    \node[gv] at (4,5) {};
    \node[gvb] at (4,6) {};
    \node[gvb] at (4,7) {};
    \node[gvb] at (4,8) {};
    \node[gv] at (4,9) {};
    \node[gv] at (4,10) {};
    \node[gv] at (4,11) {};
    \node[gvb] at (4,12) {};
    \node[rvb] at (4,13) {};
    \node[rv] at (4,14) {};

    \node[rv] at (4,15) {};
    \node[rv] at (4,16) {};

    \node[rv] at (5,0) {};
    \node[rvb] at (5,1) {};

    \node[bvb] at (5,2) {};
    \node[gvb] at (5,3) {};
    \node[gv] at (5,4) {};
    \node[gv] at (5,5) {};
    \node[gvb] at (5,6) {};
    \node[rvb] at (5,7) {};
    \node[gvb] at (5,8) {};
    \node[gv] at (5,9) {};
    \node[gv] at (5,10) {};
    \node[gv] at (5,11) {};
    \node[gvb] at (5,12) {};
    \node[rvb] at (5,13) {};
    \node[rv] at (5,14) {};

    \node[rv] at (5,15) {};
    \node[rv] at (5,16) {};

    \node[rv] at (6,0) {};
    \node[rvb] at (6,1) {};

    \node[gvb] at (6,2) {};
    \node[gvb] at (6,3) {};
    \node[gv] at (6,4) {};
    \node[gv] at (6,5) {};
    \node[gvb] at (6,6) {};
    \node[bvb] at (6,7) {};
    \node[gvb] at (6,8) {};
    \node[gv] at (6,9) {};
    \node[gv] at (6,10) {};
    \node[gvb] at (6,11) {};
    \node[gvb] at (6,12) {};
    \node[rvb] at (6,13) {};
    \node[rvb] at (6,14) {};

    \node[rv] at (6,15) {};
    \node[rv] at (6,16) {};

    \node[rv] at (7,0) {};
    \node[rvb] at (7,1) {};

    \node[rvb] at (7,2) {};
    \node[gvb] at (7,3) {};
    \node[gv] at (7,4) {};
    \node[gv] at (7,5) {};
    \node[gvb] at (7,6) {};
    \node[gvb] at (7,7) {};
    \node[gvb] at (7,8) {};
    \node[gv] at (7,9) {};
    \node[gv] at (7,10) {};
    \node[gvb] at (7,11) {};
    \node[rvb] at (7,12) {};
    \node[bvb] at (7,13) {};
    \node[rvb] at (7,14) {};

    \node[rv] at (7,15) {};
    \node[rv] at (7,16) {};

    \node[rv] at (8,0) {};
    \node[rv] at (8,1) {};

    \node[rvb] at (8,2) {};
    \node[gvb] at (8,3) {};
    \node[gv] at (8,4) {};
    \node[gv] at (8,5) {};
    \node[gv] at (8,6) {};
    \node[gv] at (8,7) {};
    \node[gv] at (8,8) {};
    \node[gv] at (8,9) {};
    \node[gvb] at (8,10) {};
    \node[gvb] at (8,11) {};
    \node[rvb] at (8,12) {};
    \node[rvb] at (8,13) {};
    \node[rvb] at (8,14) {};

    \node[rv] at (8,15) {};
    \node[rv] at (8,16) {};

    \node[rv] at (9,0) {};
    \node[rv] at (9,1) {};

    \node[rvb] at (9,2) {};
    \node[gvb] at (9,3) {};
    \node[gvb] at (9,4) {};
    \node[gvb] at (9,5) {};
    \node[gvb] at (9,6) {};
    \node[gvb] at (9,7) {};
    \node[gvb] at (9,8) {};
    \node[gvb] at (9,9) {};
    \node[gvb] at (9,10) {};
    \node[rvb] at (9,11) {};
    \node[rvb] at (9,12) {};
    \node[rv] at (9,13) {};
    \node[rv] at (9,14) {};

    \node[rv] at (9,15) {};
    \node[rv] at (9,16) {};

    \node[rv] at (10,0) {};
    \node[rv] at (10,1) {};

    \node[rvb] at (10,2) {};
    \node[rvb] at (10,3) {};
    \node[rvb] at (10,4) {};
    \node[rvb] at (10,5) {};
    \node[rvb] at (10,6) {};
    \node[rvb] at (10,7) {};
    \node[rvb] at (10,8) {};
    \node[rvb] at (10,9) {};
    \node[rvb] at (10,10) {};
    \node[rvb] at (10,11) {};
    \node[rv] at (10,12) {};
    \node[rv] at (10,13) {};
    \node[rv] at (10,14) {};

    \node[rv] at (10,15) {};
    \node[rv] at (10,16) {};

    \node[rv] at (11,0) {};
    \node[rv] at (11,1) {};

    \node[rv] at (11,2) {};
    \node[rv] at (11,3) {};
    \node[rv] at (11,4) {};
    \node[rv] at (11,5) {};
    \node[rv] at (11,6) {};
    \node[rv] at (11,7) {};
    \node[rv] at (11,8) {};
    \node[rv] at (11,9) {};
    \node[rv] at (11,10) {};
    \node[rv] at (11,11) {};
    \node[rv] at (11,12) {};
    \node[rv] at (11,13) {};
    \node[rv] at (11,14) {};
    
    \node[rv] at (11,15) {};
    \node[rv] at (11,16) {};

    \node[rv] at (12,0) {};
    \node[rv] at (12,1) {};

    \node[rv] at (12,2) {};
    \node[rv] at (12,3) {};
    \node[rv] at (12,4) {};
    \node[rv] at (12,5) {};
    \node[rv] at (12,6) {};
    \node[rv] at (12,7) {};
    \node[rv] at (12,8) {};
    \node[rv] at (12,9) {};
    \node[rv] at (12,10) {};
    \node[rvb] at (12,11) {};
    \node[rvb] at (12,12) {};
    \node[rvb] at (12,13) {};
    \node[rv] at (12,14) {};

    \node[rv] at (12,15) {};
    \node[rv] at (12,16) {};

    \node[rv] at (13,0) {};
    \node[rv] at (13,1) {};

    \node[rv] at (13,2) {};
    \node[rv] at (13,3) {};
    \node[rv] at (13,4) {};
    \node[rv] at (13,5) {};
    \node[rv] at (13,6) {};
    \node[rv] at (13,7) {};
    \node[rvb] at (13,8) {};

    \node[rvb] at (13,9) {};
    \node[rvb] at (13,10) {};
    \node[rvb] at (13,11) {};
    \node[gvb] at (13,12) {};
    \node[rvb] at (13,13) {};
    \node[rv] at (13,14) {};
    \node[rv] at (13,15) {};
    \node[rv] at (13,16) {};

    \node[rv] at (14,0) {};
    \node[rv] at (14,1) {};

    \node[rv] at (14,2) {};
    \node[rv] at (14,3) {};
    \node[rv] at (14,4) {};
    \node[rv] at (14,5) {};
    \node[rv] at (14,6) {};
    \node[rv] at (14,7) {};
    \node[rvb] at (14,8) {};
    \node[bvb] at (14,9) {};
    \node[rvb] at (14,10) {};
    \node[rvb] at (14,11) {};
    \node[rvb] at (14,12) {};
    \node[rvb] at (14,13) {};
    \node[rv] at (14,14) {};

    \node[rv] at (14,15) {};
    \node[rv] at (14,16) {};

    \node[rv] at (15,0) {};
    \node[rv] at (15,1) {};
    \node[rv] at (15,2) {};
    \node[rv] at (15,3) {};
    \node[rv] at (15,4) {};
    \node[rv] at (15,5) {};
    \node[rv] at (15,6) {};
    \node[rv] at (15,7) {};
    \node[rvb] at (15,8) {};
    \node[rvb] at (15,9) {};
    \node[rvb] at (15,10) {};
    \node[rv] at (15,11) {};
    \node[rv] at (15,12) {};
    \node[rv] at (15,13) {};
    \node[rv] at (15,14) {};
    \node[rv] at (15,15) {};
    \node[rv] at (15,16) {};    

    \node[rv] at (16,0) {};
    \node[rv] at (16,1) {};
    \node[rv] at (16,2) {};
    \node[rv] at (16,3) {};
    \node[rv] at (16,4) {};
    \node[rv] at (16,5) {};
    \node[rv] at (16,6) {};
    \node[rv] at (16,7) {};
    \node[rv] at (16,8) {};
    \node[rv] at (16,9) {};
    \node[rv] at (16,10) {};
    \node[rv] at (16,11) {};
    \node[rv] at (16,12) {};
    \node[rv] at (16,13) {};
    \node[rv] at (16,14) {};
    \node[rv] at (16,15) {};
    \node[rv] at (16,16) {};

    \draw[name path = C,boundary, rounded corners=5] (.5,5.5) --
    (.5,12.5) -- (1.5,12.5) -- (1.5,13.5) -- (5.5,13.5) -- (5.5,14.5)
    -- (8.5,14.5) -- (8.5,12.5) -- (9.5,12.5) -- (9.5,11.5) --
    (10.5,11.5) -- (10.5,1.5) -- (7.5,1.5) -- (7.5,.5) -- (2.5,.5) --
    (2.5,1.5) -- (1.5,1.5) -- (1.5,2.5) -- (.5,2.5) -- (.5,5.5);

    \draw[name path = D,boundary, rounded corners=5] 
    (6.5,10.5) -- (5.5,10.5) --
    (5.5,11.5) -- (3.5,11.5) -- (3.5,10.5) -- (2.5,10.5) -- (2.5,4.5)
    -- (3.5,4.5) -- (3.5,3.5)--(4.5,3.5) -- (8.5,3.5) -- (8.5,9.5) --
    (7.5,9.5) -- (7.5,10.5) -- (6.5,10.5);

    \tikzfillbetween[of=C and D]{pattern=north west lines,
        opacity=0.7};

    \draw[boundary, rounded corners=5] (14.5,7.5) -- (12.5,7.5) --
    (12.5,10.5) -- (11.5,10.5) -- (11.5,13.5) -- (14.5,13.5) -- (14.5,10.5)
    -- (15.5,10.5) -- (15.5,7.5) -- (14.5,7.5);

    \draw[boundaryfill, rounded corners=5] (14.5,7.5) -- (12.5,7.5) --
    (12.5,10.5) -- (11.5,10.5) -- (11.5,13.5) -- (14.5,13.5) -- (14.5,10.5)
    -- (15.5,10.5) -- (15.5,7.5) -- (14.5,7.5);

    \draw[boundary, rounded corners=5] (4.5,5.5) -- (3.5,5.5) --
    (3.5,8.5) -- (7.5,8.5) -- (7.5,5.5) -- (4.5,5.5);

    \draw[boundaryfill, rounded corners=5] (4.5,5.5) -- (3.5,5.5) --
    (3.5,8.5) -- (7.5,8.5) -- (7.5,5.5) -- (4.5,5.5);
  \end{tikzpicture}
  \caption{A $3$-state Potts model configuration with padded red boundary
    conditions. Incorrect vertices and the contours they define are
    indicated by shading.}
  \label{fig:bdry}
\end{figure}

It is a non-trivial fact that for each contour $\gamma$ and each
connected component $A$ of $\Z^{d}\setminus\overline\gamma$, the set of vertices $i\in A$ such that $d_\infty(i,\overline
\gamma)=1$ is connected under the $d_{\infty}$ distance~\cite[Appendix~B.15]{friedli2017statistical} (see also~\cite[Lemma~7.19]{friedli2017statistical}).  This implies  there
exists a $\varphi'$ such that $\omega_{i}=\varphi'$ for all such $i$; the label of $A$
is $\varphi'$. 
This defines the set of contours and their labelling functions.  Note
the set of contours $\Gamma(\omega)$ is matching and of type
$\varphi$.

Conversely, let $\overline \gamma$ be a $d_\infty$-connected subset of
$\Lam$ so that $d(\overline \gamma, \Lam^c) >1$.  Let
$\omega_{\overline \gamma}$ be an assignment of spins to
$\overline \gamma$ so that:
\begin{itemize}
\item For every $i \in \overline \gamma$, there is a
  $j \in \overline \gamma$, $d_\infty(i,j)=1$ so that
  $\omega_{\overline \gamma,i} \ne \omega_{\overline \gamma,j}$.
\item Let $A_0, \dots, A_t$ denote the connected components of
  $\Z^d \setminus \overline\gamma$, with $A_0$ the unique infinite
  component.  For each $i$ there is a spin
  $\varphi' \bydef \text{lab}_\gamma(A_i)$ so that
  $\omega_{\overline \gamma,j}=\varphi'$ for all
  $j \in \overline \gamma$, $d_{\infty}(j, A_i)=1$.  Moreover,
  $\text{lab}_\gamma(A_0) = \varphi$.
\end{itemize}
Any contour satisfying these conditions belongs to the set
$\cC^\varphi(\Lam)$ and can be realized by a configuration
$\omega \in \Omega_\Lam^\varphi$ by setting
$\omega_j = \text{lab}_\gamma(A_i)$ for any $j \in \Lam \cap A_i$ and
$\omega_j = \omega_{\overline \gamma,j}$ for any
$j \in \overline \gamma$. Iterating this construction shows that any
set of matching contours $\Gamma \in \cG^\varphi_{\text{match}}(\Lam)$
of type $\varphi$ gives rise to a Potts configuration
$\omega \in \Omega_\Lam^\varphi$.

Recall that we write $E(H)\subset E(\Z^{d})$ for the set of edges of a
subgraph $H$ of $\Z^{d}$. Define the surface energy of a contour $\gamma$ by
\begin{equation}
  \label{eqPSurf}
  \|\gamma\| = \sum_{\{i,j\}\in E(\overline\gamma)} \mathbf 1_{\omega_{i}\neq\omega_{j}},
\end{equation}
This is a positive integer by the definition of the boundary.  Note
also that we can check whether an assignment satisfies the condition
of a contour and compute $\|\gamma\|$ in time linear in
$|\overline \gamma|$, which shows Assumption~\ref{assumCompute}.

Letting $z=e^{-\beta}$, \eqref{eqPSurf} yields an expression for the
Potts partition function:
\begin{align*}
  Z^\varphi_{q,\Lam}(\beta) &=\sum_{\omega \in \Omega_\Lam^\varphi} e^{\beta \sum_{\{i,j\}\in E(\Lam)} \mathbf 1_{\omega_i = \omega_j} } \\
  &=\sum_{\Gamma \in
    \cG^{\varphi}_{\text{match}}(\Lam)}   z^{-|E(\Lam)|} \prod_{\gamma \in \Gamma}
  z^{\| \gamma \|} \\
  &= z^{-|E(\Lam)|} Z^\varphi(\Lam,z) 
\end{align*}
where $Z^\varphi(\Lam,z)$ is the contour model partition function
defined in~\eqref{eqContourMatching}.

Lastly we must show that Assumption~\ref{asSurface} is satisfied. The
upper bound is immediate, as each vertex has only $2d$ neighbors.  A
crude lower bound can be obtained by noting that for
$v\in\overline\gamma$, there must be a $u$ with $d_{\infty}(u,v)=1$ such
that $\omega_{u}\neq\omega_{v}$. Removing all vertices at $d_{\infty}$
distance at most $1$ from $u$ and $v$, the same holds true for the
remaining vertices of $\overline\gamma$. This implies $\|\gamma\|\geq
\lceil{|\overline\gamma|/(2\cdot 3^{d})\rceil}$. 
\end{example}

\begin{example}[The hard-core model]
  \label{ex:HC-Contour}
  We can express the hard-core model as a contour model in a similar
  way.  We set $\Omega = \{0,1\}$ and
  $\Xi = \{\text{even}, \text{odd} \}$. It will be convenient to
  identify independent sets $I\in\cI(\Lam)$ with their characteristic
  vectors $\omega^{I}\in\{0,1\}^{\Lam}$. In particular we define
  $\omega^{\text{even}}\in \cI(\Z^{d})$ by
  $\omega^{\text{even}} = \mathbf 1_{i\text{ is even}}$, and similarly
  for $\omega^{\text{odd}}$.  The set of valid configurations for the
  even padded boundary conditions is
  \begin{equation*}
    \Omega_{\Lam}^{\text{even}} 
    =  
    \{ \omega \in \{ 0,1\}^\Lam: \omega_i = \omega^{\text{even}}_i
    \text{ if } d_{\infty}(i, \Lam^c) \le 2 \} \,.  
  \end{equation*}
\begin{figure}
  \centering
  \begin{tikzpicture}[scale=.5]
    \draw[boundary] (-1,-1) grid (17,11);
    \draw[lam] (0,0) -- (0,10) -- (16,10) -- (16,0) -- (0,0);

    \node[ov] at (0,0) {};
    \node[ev] at (0,1) {};
    \node[ov] at (0,2) {};
    \node[ev] at (0,3) {};
    \node[ov] at (0,4) {};
    \node[ev] at (0,5) {};
    \node[ov] at (0,6) {};
    \node[ev] at (0,7) {};
    \node[ov] at (0,8) {};
    \node[ev] at (0,9) {};
    \node[ov] at (0,10) {};

    \node[ev] at (1,0) {};
    \node[ov] at (1,1) {};
    \node[ev] at (1,2) {};
    \node[ovb] at (1,3) {};
    \node[evb] at (1,4) {};
    \node[ovb] at (1,5) {};
    \node[evb] at (1,6) {};
    \node[ovb] at (1,7) {};
    \node[ev] at (1,8) {};
    \node[ov] at (1,9) {};
    \node[ev] at (1,10) {};

    \node[ov] at (2,0) {};
    \node[ev] at (2,1) {};

    \node[ovb] at (2,2) {};
    \node[evb] at (2,3) {};
    \node[evb] at (2,4) {};
    \node[evb] at (2,5) {};
    \node[evb] at (2,6) {};
    \node[evb] at (2,7) {};
    \node[ovb] at (2,8) {};

    \node[ev] at (2,9) {};
    \node[ov] at (2,10) {};

    \node[ev] at (3,0) {};
    \node[ovb] at (3,1) {};

    \node[evb] at (3,2) {};
    \node[evb] at (3,3) {};
    \node[ovb] at (3,4) {};
    \node[evb] at (3,5) {};
    \node[evb] at (3,6) {};
    \node[evb] at (3,7) {};
    \node[evb] at (3,8) {};

    \node[ov] at (3,9) {};
    \node[ev] at (3,10) {};

    \node[ov] at (4,0) {};
    \node[evb] at (4,1) {};

    \node[evb] at (4,2) {};
    \node[ovb] at (4,3) {};
    \node[ev] at (4,4) {};
    \node[ovb] at (4,5) {};
    \node[evb] at (4,6) {};
    \node[ovb] at (4,7) {};
    \node[evb] at (4,8) {};

    \node[evb] at (4,9) {};
    \node[ov] at (4,10) {};

    \node[ev] at (5,0) {};
    \node[ovb] at (5,1) {};

    \node[evb] at (5,2) {};
    \node[evb] at (5,3) {};
    \node[ov] at (5,4) {};
    \node[ev] at (5,5) {};
    \node[ov] at (5,6) {};
    \node[evb] at (5,7) {};
    \node[evb] at (5,8) {};

    \node[ovb] at (5,9) {};
    \node[ev] at (5,10) {};

    \node[ov] at (6,0) {};
    \node[evb] at (6,1) {};

    \node[evb] at (6,2) {};
    \node[ovb] at (6,3) {};
    \node[ev] at (6,4) {};
    \node[ov] at (6,5) {};
    \node[ev] at (6,6) {};
    \node[ovb] at (6,7) {};
    \node[evb] at (6,8) {};

    \node[evb] at (6,9) {};
    \node[ov] at (6,10) {};

    \node[ev] at (7,0) {};
    \node[ovb] at (7,1) {};

    \node[evb] at (7,2) {};
    \node[evb] at (7,3) {};
    \node[ovb] at (7,4) {};
    \node[evb] at (7,5) {};
    \node[ovb] at (7,6) {};
    \node[evb] at (7,7) {};
    \node[evb] at (7,8) {};

    \node[ovb] at (7,9) {};
    \node[ev] at (7,10) {};

    \node[ov] at (8,0) {};
    \node[ev] at (8,1) {};

    \node[ovb] at (8,2) {};
    \node[evb] at (8,3) {};
    \node[evb] at (8,4) {};
    \node[evb] at (8,5) {};
    \node[evb] at (8,6) {};
    \node[evb] at (8,7) {};
    \node[ovb] at (8,8) {};

    \node[ev] at (8,9) {};
    \node[ov] at (8,10) {};

    \node[ev] at (9,0) {};
    \node[ov] at (9,1) {};

    \node[evb] at (9,2) {};
    \node[evb] at (9,3) {};
    \node[evb] at (9,4) {};
    \node[ovb] at (9,5) {};
    \node[evb] at (9,6) {};
    \node[ovb] at (9,7) {};
    \node[ev] at (9,8) {};

    \node[ov] at (9,9) {};
    \node[ev] at (9,10) {};

    \node[ov] at (10,0) {};
    \node[ev] at (10,1) {};

    \node[ovb] at (10,2) {};
    \node[evb] at (10,3) {};
    \node[ovb] at (10,4) {};
    \node[ev] at (10,5) {};
    \node[ov] at (10,6) {};
    \node[ev] at (10,7) {};
    \node[ov] at (10,8) {};

    \node[ev] at (10,9) {};
    \node[ov] at (10,10) {};

    \node[ev] at (11,0) {};
    \node[ov] at (11,1) {};

    \node[ev] at (11,2) {};
    \node[ov] at (11,3) {};
    \node[ev] at (11,4) {};
    \node[ov] at (11,5) {};
    \node[ev] at (11,6) {};
    \node[ov] at (11,7) {};
    \node[ev] at (11,8) {};

    \node[ov] at (11,9) {};
    \node[ev] at (11,10) {};

    \node[ov] at (12,0) {};
    \node[ev] at (12,1) {};

    \node[ov] at (12,2) {};
    \node[ev] at (12,3) {};
    \node[ovb] at (12,4) {};
    \node[evb] at (12,5) {};
    \node[ovb] at (12,6) {};
    \node[ev] at (12,7) {};
    \node[ov] at (12,8) {};

    \node[ev] at (12,9) {};
    \node[ov] at (12,10) {};

    \node[ev] at (13,0) {};
    \node[ovb] at (13,1) {};

    \node[evb] at (13,2) {};
    \node[ovb] at (13,3) {};
    \node[evb] at (13,4) {};
    \node[evb] at (13,5) {};
    \node[evb] at (13,6) {};
    \node[ov] at (13,7) {};
    \node[ev] at (13,8) {};

    \node[ov] at (13,9) {};
    \node[ev] at (13,10) {};

    \node[ov] at (14,0) {};
    \node[evb] at (14,1) {};

    \node[evb] at (14,2) {};
    \node[evb] at (14,3) {};
    \node[ovb] at (14,4) {};
    \node[evb] at (14,5) {};
    \node[ovb] at (14,6) {};
    \node[ev] at (14,7) {};
    \node[ov] at (14,8) {};

    \node[ev] at (14,9) {};
    \node[ov] at (14,10) {};

    \node[ev] at (15,0) {};
    \node[ovb] at (15,1) {};
    \node[evb] at (15,2) {};
    \node[ovb] at (15,3) {};
    \node[ev] at (15,4) {};
    \node[ov] at (15,5) {};
    \node[ev] at (15,6) {};
    \node[ov] at (15,7) {};
    \node[ev] at (15,8) {};
    \node[ov] at (15,9) {};
    \node[ev] at (15,10) {};

    \node[ov] at (16,0) {};
    \node[ev] at (16,1) {};
    \node[ov] at (16,2) {};
    \node[ev] at (16,3) {};
    \node[ov] at (16,4) {};
    \node[ev] at (16,5) {};
    \node[ov] at (16,6) {};
    \node[ev] at (16,7) {};
    \node[ov] at (16,8) {};
    \node[ev] at (16,9) {};
    \node[ov] at (16,10) {};

    \draw[name path = A,boundary, rounded corners=5] (.5,3.5) --
    (.5,7.5) -- (1.5,7.5) -- (1.5,8.5) -- (3.5,8.5) -- (3.5,9.5) --
    (7.5,9.5) -- (7.5,8.5) -- (8.5,8.5) -- (8.5,7.5) -- (9.5,7.5) --
    (9.5,4.5) -- (10.5,4.5) -- (10.5,1.5) -- (7.5,1.5) -- (7.5,.5) --
    (2.5,.5) -- (2.5,1.5) -- (1.5,1.5) -- (1.5,2.5) -- (.5,2.5) --
    (.5,3.5);

    \draw[name path = B,boundary, rounded corners=5] (5.5,6.5) --
    (6.5,6.5) -- (6.5,3.5) -- (3.5,3.5) -- (3.5,4.5) -- (4.5,4.5) --
    (4.5,6.5) -- (5.5,6.5);

    \tikzfillbetween[of=A and B]{pattern=north west lines,
        opacity=0.7};

    \draw[boundary, rounded corners=5] (14.5,.5) -- (12.5,.5) --
    (12.5,3.5) -- (11.5,3.5) -- (11.5,6.5) -- (14.5,6.5) -- (14.5,3.5)
    -- (15.5,3.5) -- (15.5,.5) -- (14.5,.5);

    \draw[boundaryfill, rounded corners=5] (14.5,.5) -- (12.5,.5) --
    (12.5,3.5) -- (11.5,3.5) -- (11.5,6.5) -- (14.5,6.5) -- (14.5,3.5)
    -- (15.5,3.5) -- (15.5,.5) -- (14.5,.5);
  \end{tikzpicture}
  \caption{A hard-core model configuration with padded even boundary
    conditions. Incorrect vertices and the contours they define are
    indicated by shading.}
  \label{fig:bdry-HC}
\end{figure}
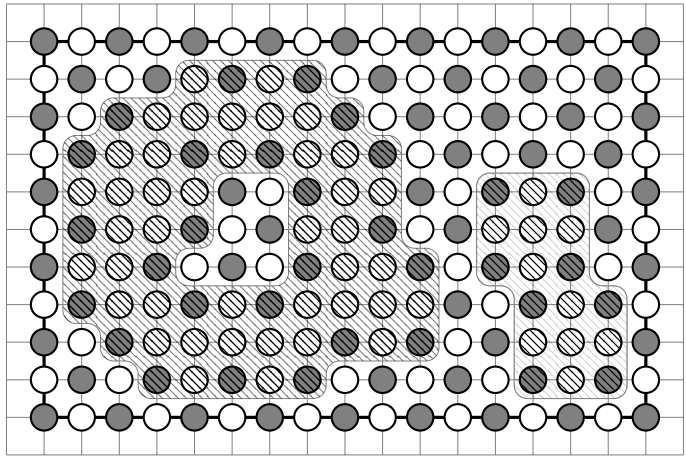

We say a vertex $i \in \Lam$ is \emph{correct} with respect to
$\omega \in \Omega_\Lam^{\text{even}}$ if either
$\omega_j = \omega^{\text{even}}_j$ for all $j \in \Lam$ such that
$d_\infty(i,j) \le 1$ or $\omega_j = \omega^{\text{odd}}_j$ for all
$j \in \Lam$ such that $d_\infty(i,j) \le 1$.  All other vertices of
$\Lam$ are incorrect.  Again $\Gamma(\omega)$ is the set of all
incorrect vertices with respect to $\omega$, and each connected
component (with respect to the $d_\infty$ distance) of
$\Gamma(\omega)$ is the support $\overline\gamma$ of a contour
$\gamma$, and $\omega_{\gamma}$ is the restriction of $\omega$ to
$\gamma$. See Figure~\ref{fig:bdry-HC} for an illustration.
Again we have $d(\overline\gamma,\Lam^{c})>1$ for all
contours $\gamma$.  For each contour $\gamma$ and each connected
component  $A$ of $\Z^{d}\setminus\overline\gamma$ either
$\omega_{i}=\omega_i^{\text{even}}$ for all $i\in A$ such that
$d_\infty(i,\overline \gamma)=1$ or $\omega_{i}=\omega_i^{\text{odd}}$
for all $i\in A$ such that $d_\infty(i,\overline \gamma)=1$; this
again relies on~\cite[Appendix~B.15]{friedli2017statistical} as in Example~\ref{ex:FP-Contour}.  In the
first case, $\text{lab}_\gamma(A) =\text{even}$ and in the second,
$\text{lab}_\gamma(A) =\text{odd}$.  The set $\cC^{\text{even}}(\Lam)$
consists of all possible contours $\gamma$ of type $\text{even}$ with
$d_\infty(\overline \gamma, \Lam^c) >1$.

Analogously to the Potts model, each configuration
$\omega \in \Omega_{\Lam}^{\text{even}}$ corresponds to a matching set
of contours $\Gamma(\omega) $ of even type and each set of matching
contours $\Gamma \in \cG_{\text{match}}^{\text{even}}$ corresponds to
a configuration $\omega \in \Omega_{\Lam}^{\text{even}}$.

Given $A\subset\Lam$, let $A^{\text{even}}$ denote the set of even
vertices of $A$. We define the surface energy of $\gamma$ to be
\begin{equation}
  \label{eqHCenergy}
  \|\gamma \| =   \frac{1}{4d} \sum_{\substack{i \in \overline \gamma \\ \omega_{\overline\gamma,i} =0 }} (2d - \sum_{j \in N(i)} \omega_{\overline\gamma,j})  \,,
\end{equation}
where $N(i)$ is the set of neighbors of $i$ in $\Z^d$.  The surface
energy is completely determined by $\overline \gamma$ and
$\omega_{\overline \gamma}$. Let $\Gamma(\omega^{I})$ denote
  the set of contours determined by the configuration $\omega^{I}$. A double
  counting argument shows that
  \begin{equation}
    \label{eqIndsize}
    |I| = |\Lam^{\text{even}}| - \sum_{\gamma\in\Gamma(\omega^{I})}\|\gamma\|.
  \end{equation}
Since each contour $\gamma$ can arise from a hard-core
configuration, this formula shows $\|\gamma\|$ is integer valued. We
can determine if a given assignment of spins to a
$d_{\infty}$-connected subgraph $\overline\gamma$ satisfies the
definition of a contour, and can compute $\|\gamma\|$ in linear
time. This shows Assumption~\ref{assumCompute} holds.

Let $z=1/\lam$. Using~\eqref{eqIndsize} we obtain
\begin{align*}
Z^{\text{even}}_\Lam(\lam)
&= \sum_{\omega \in \Omega_{\Lam}^{\text{even}}}
  \lam^{|\Lam^{\text{even}}|} \prod_{\gamma \in \Gamma(\omega)}
  \lam^{-\| \gamma \|} \\
&= z^{-|\Lam^{\text{even}}|}  \sum_{\Gamma \in \cG^{\text{even}}_{\text{match}}(\Lam)} \prod_{\gamma \in \Gamma} z^{\| \gamma \|} \\
&= z^{-|\Lam^{\text{even}}|} Z^{\text{even}}(\Lam,z) \, ,
\end{align*}
where $Z^{\text{even}}(\Lam,z)$ is the contour model partition function.

Assumption~\ref{asSurface} is also satisfied. The upper bound follows
as $ \| \gamma \| \le |\overline \gamma|$ as each $i$ can contribute
at most $1$ to the sum in~\eqref{eqHCenergy}. For the lower bound we
have $\|\gamma\|\geq \frac{|\overline\gamma|}{4d\cdot 2\cdot3^{d}}$; 
this is a crude bound obtained by using that for every incorrect vertex
$i$ there must be a $j$ with $d_\infty(i,j) \le1$ so that $j$ is
unoccupied and has an unoccupied
neighbor. 
\end{example}

\section{Convergence of the cluster expansion}
\label{secZeroFree}

To apply Theorems~\ref{thmPolymerApprox} or~\ref{thmContourModelCount}
requires knowing that the partition function is non-zero in a disc
around the origin in the complex plane.  Occasionally, recall
Section~\ref{sec:clus-ex}, this is provided by model-specific
results. More generally, however, there are criteria for polymer and
contour models that guarantee the partition function is non-zero in a
disc around the origin.

The following theorem gives a criterion for the convergence of the
cluster expansion; it is a special case of a result of Koteck{\`y} and
Preiss~\cite{kotecky1986cluster}.  The theorem says that if the
weights decay at fast enough exponential rates, then the partition
function is non-vanishing in some disc. For refined criteria,
see~\cite{FernandezProcacci}.

\begin{theorem}[Koteck{\`y} and Preiss~\cite{kotecky1986cluster}] 
  \label{thmKP}
  Suppose that for every $\gamma \in \cC(G)$,
  \begin{align} 
    \label{eqKPcondition} 
    \sum_{\gamma' \nsim \gamma} |w(\gamma',z)|e^{|\overline \gamma'|} 
    \le |\overline \gamma| \,, 
  \end{align} 
  where the sum is over all polymers $\gamma'$ incompatible with $\gamma$.  Then the cluster expansion for $\log Z(G,z)$ converges absolutely and, in particular, $Z(G,z)\ne 0$.  
\end{theorem}

\begin{example}[Hard-core model at low density] 
  \label{exHC-KP}
  Recall Example~\ref{ex:HC-LD}. We can apply Theorem~\ref{thmKP} to the
  polymer representation of the hard-core model on graphs of maximum
  degree $\Delta$. Polymers have size $1$ and are incompatible with
  at most $\Delta+1$ polymers; the $+1$ accounts for incompatibility
  with itself. Equation~\eqref{eqKPcondition} becomes
  $(\Delta+1) |z| e \le 1 $, or 
  \begin{equation}
    |z| \le \frac{1}{e (\Delta+1)} \,.
  \end{equation}
  This radius of convergence is not sharp; recall Example~\ref{ex:HC-LD2}.
  It is, however, asymptotically sharp, since $\frac{ (\Delta-1)^{\Delta-1} }{\Delta^\Delta } \sim \frac{1}{e \Delta}$ as $\Delta \to \infty$. For more more on zero-free regions of the hard-core partition function see~\cite{scott2005repulsive,peters2017conjecture}.  \end{example}

We cannot apply a result like Theorem~\ref{thmKP} to the outer contour model of Section~\ref{secContourModels} with weights given by~\eqref{eqouterweights}, as these weight functions generally grow exponentially in the size of a contour and its interior.  Instead we use a standard trick in Pirogov--Sinai theory. 

Define the weight function
\begin{equation}
\label{eqWeightsContour}
w^\varphi(\gamma,z) \bydef z^{ \|\gamma\|} \prod_{\varphi' \in \Xi } \frac{ Z^{\varphi'}(\text{int}_{\varphi'}(\gamma),z )  }{  Z^\varphi( \text{int}_{\varphi'}(\gamma),z )  } \,. 
\end{equation}
Then we can rewrite~\eqref{eqContourPartition} as 
\begin{equation}
\label{eqContourPart3}
Z^\varphi(\Lam,z) = \sum_{\Gamma \in \cG^{\varphi}_{\text{ext}}(\Lam)} \prod_{\gamma \in \Gamma}  
\left ( w^\varphi(\gamma,z)  \prod_{\varphi' \in \Xi} Z^{\varphi } (\text{int}_{\varphi'} \gamma ,z) \right) \, ,
\end{equation}
and now the partition function $ Z^{\varphi } (\text{int}_{\varphi'} \gamma ,z)$ inside the product can be written using~\eqref{eqContourPart3} again.  Iterating this yields
\begin{equation}
\label{eqContourRepresentation}
Z^\varphi(\Lam,z) = \sum_{\Gamma \in \cG^{\varphi}(\Lam)} \prod_{\gamma \in \Gamma}  
w^\varphi(\gamma,z)  \,,
\end{equation}
where $\cG^\varphi(\Lam)$ is the collection of all subsets of contours from $\cC^{\varphi}(\Lam)$ that are pairwise compatible (but are no longer required to be mutually external).  We call~\eqref{eqContourRepresentation} the \emph{polymer representation} of the partition function.  Note however, that unlike the outer contour representation, there is not a mapping from the sets of contours appearing in the sum in~\eqref{eqContourRepresentation} to spin configurations.

The polymer representation is of exactly the same form as the polymer partition function~\eqref{eqPolymerPartition}, but with a different weight function and the restriction that all contours have type $\varphi$. Moreover, the weight functions $w^\varphi(\gamma,z)$ satisfy the condition on the weight functions in the polymer model: the first non-zero Taylor series coefficient of $w^\varphi(\gamma,z)$ is of order at least $ |\overline \gamma| \rho$.

 In the remainder of this section we indicate a method for proving the convergence of the cluster expansion for contour models with weight functions given by~\eqref{eqWeightsContour}. The method, due to Borgs and Imbrie~\cite{borgs1989unified}, is based on Zahradn{\'i}k's truncation-based approach to Pirogov--Sinai theory~\cite{zahradnik1984alternate,pirogov1975phase}

\begin{assumption} 
  \label{assumeTransitive} 
  The surface energy function $\| \cdot \|$ and the labeling function are \emph{translation invariant}, i.e., if there is an $a\in \Z^{d}$ such that
  $\overline\gamma'=\overline\gamma+a$ and $\omega_{\overline\gamma'}(i) = \omega_{\overline\gamma}(i-a)$ for all $i\in\overline\gamma'$, then they have the same surface energy and the labelling function respects the translation.  
\end{assumption}

To state the result of Borgs and Imbrie we must define the notion of a
stable contour and a stable ground state. Recall from
Section~\ref{secContourModels} that $\Xi$ denotes the finite set of
ground states. A contour $\gamma$ of type $\varphi$ is \emph{stable}
if
\begin{equation*}
 \frac{ Z^{\varphi'}(\text{int}_{\varphi'}(\gamma),z )  }{  Z^\varphi( \text{int}_{\varphi'}(\gamma),z )  }   \le e^{4 |\partial ^{\text{ex}} \text{int}_{\varphi'}(\gamma)| }
\end{equation*}
for all $\varphi' \in \Xi$.  Let $\cG^{\varphi}_{\text{stab}}(\Lam)$
be the collection of all sets of pairwise compatible, stable contours
from $\cC^\varphi(\Lam)$.  The \emph{truncated partition} is
\begin{equation}
\label{eqtruncatedPart}
Z^\varphi_{\text{trun}}(\Lam,z) \bydef \sum_{\Gamma \in \cG^{\varphi}_{\text{stab}}(\Lam)} \prod_{\gamma \in \Gamma}  
w^\varphi(\gamma,z)  \,.
\end{equation}
If Peierls' condition holds and $|z|$ is small enough then the cluster
expansion for the truncated partition function converges, and hence
the limiting free energy of the truncated partition functions exists
for each ground state $\varphi \in \Xi$, i.e.,
\begin{equation*}
f(\varphi) \bydef \lim_{\Lam \to \Z^d}  \frac{1}{|\Lam|} \log Z^\varphi_{\text{trun}}(\Lam,z) 
\end{equation*}
exists when the limit is taken in the sense of van Hove.\footnote{This means $\frac{|\partial^{\text{in}}\Lam_{n}|}{|\Lam_{n}|}\to 0$, see~\cite[Section~3.2.1]{friedli2017statistical}} A \emph{stable} ground state $\varphi$ is one for which $\text{Re} \, f(\varphi) \ge \text{Re} \, f(\varphi')$ for all $\varphi'\in \Xi$.  In particular, at least one stable ground state exists.

\begin{theorem}[Borgs, Imbrie~\cite{borgs1989unified}] 
\label{thmImbrie} 
Fix $d \ge 2$. Suppose a contour model satisfies Assumptions~\ref{asSurface}
  and~\ref{assumeTransitive}.  Then there exists a constant $\del= \del(d, \rho, \Xi)>0$ so that for all $z \in \mathbb C$ with $|z| < \del$, all regions $\Lam$, and all stable ground states $\varphi$, the weights $w^{\varphi}$ satisfy \eqref{eqKPcondition}. In particular, the cluster expansion for $\log Z^\varphi(\Lam,z)$ converges absolutely, and $Z^\varphi(\Lam,z) \ne 0$.
\end{theorem}
\begin{proof}   
  We must explain why the analysis of~\cite{borgs1989unified} applies
  when $d\geq 2$ and Assumptions~\ref{asSurface}
  and~\ref{assumeTransitive} hold. This is essentially immediate as
  these assumptions
  constitute~\cite[Equation~(2.1)]{borgs1989unified}, which is the
  assumption used in~\cite{borgs1989unified}. Two
  further remarks are in order. First, while the setup discussed in
  the introduction to~\cite{borgs1989unified} takes place in $\R^{d}$,
  the analysis applies to partition functions that can be expressed in
  the algebraic form
  of~\cite[Equation~(2.6)]{borgs1989unified}. Second, contours
  in~\cite{borgs1989unified} are geometric objects with phase labels,
  while our contours additionally have spins assigned to
  vertices. This does not cause any complication as it only modifies
  the exponential growth rate of the number of contours.
\qedhere \end{proof}

For this paper we do not need to go into the details of proving the
stability of particular ground states: for the Potts and hard-core
models symmetry ensures all ground states are stable. Thus by
combining Theorem~\ref{thmImbrie} with
Theorem~\ref{thmContourModelCount} we can prove the FPTAS portions of
Theorems~\ref{HCMainThm} and~\ref{PottsMainThm}.

\begin{proof}[Proof of Theorem~\ref{HCMainThm}, FPTAS part]
  By Example~\ref{ex:HC-Contour} the hard-core model satisfies Assumptions~\ref{assumCompute},~\ref{asSurface}, and~\ref{assumeTransitive}, and hence by Theorem~\ref{thmImbrie} there is a zero-free region for the partition function.  The result then follows by Theorem~\ref{thmContourModelCount}.
\qedhere \end{proof}

  For arbitrary $q$ and $\beta$ sufficiently large, the proof of the FPTAS portion of  Theorem~\ref{PottsMainThm} is exactly analogous to that of Theorem~\ref{HCMainThm}. Example~\ref{ex:FP-Contour} verifies Assumptions~\ref{assumCompute},~\ref{asSurface}, and~\ref{assumeTransitive}, and we obtain a zero-free region from Theorem~\ref{thmImbrie}.

\section{Sampling}
\label{secSample}

This section introduces a notion of self-reducibility based on
polymers and contours. When combined with the approximate counting
algorithms of Theorems~\ref{thmPolymerApprox}
and~\ref{thmContourModelCount} this yields efficient sampling
algorithms.

\subsection{Sampling from a polymer model}
\label{sec:meas-from-polym}

We will first introduce an algorithm to sample from a polymer model, then use a very similar algorithm to sample from a contour model.   

In order to sample from a polymer model we need one further
assumption; we note this assumption is simple to verify in the
examples of polymer models (Examples~~\ref{ex:HC-LD}
and~\ref{ex:Ising-ExF}) that we have seen so far.
\begin{assumption}
  \label{asReal}
  For $z>0$ the weights $w(\gamma,z)$ are non-negative real numbers
  for all polymers $\gamma$.  Moreover, we can compute an $\eps$-relative approximation to $w(\gamma,z)$ in time polynomial in $|\overline \gamma|$ and $1/\eps$. 
\end{assumption}

Throughout this section we will assume Assumption~\ref{asReal}
holds. In this case, given a polymer model on a graph $G$ and a real
number $z>0$ the \emph{probability measure $\mu_{G}$ associated to the
polymer model} is
\begin{equation}
  \label{eq:polymer-meas}
  \mu_{G}(\Gamma) \bydef \frac{\prod_{\gamma\in\Gamma}w(\gamma,z)}{Z(G,z)},
  \qquad \Gamma\in \cG(G),
\end{equation}
where $Z(G,z)$ is the polymer partition function defined
in~\eqref{eqPolymerPartition}. Under the conditions for which we
obtain an FPTAS for $Z(G,z)$ we obtain an efficient sampling
algorithm.
\begin{theorem}
  \label{thmPolySample}
  Under the conditions of Lemma~\ref{thmPolymerApproxGen} and
  Assumption~\ref{asReal}, for any positive real number
  $0<z<\del$ there is an efficient sampling algorithm for $\mu_G$ for
  all $G \in \mathfrak G$.
\end{theorem}

We will begin by describing an idealized sampling algorithm which
returns an exact sample from $\mu_G$ by sampling a configuration
$\Gamma$ one polymer at a time. We will then describe how to turn this
into an efficient approximate sampling algorithm.

For a set of vertices $S \subset V(G)$ and a collection of compatible
polymers $\Gamma \in \cG(G)$, let $\cC_{\Gamma,S}$ be the set of
polymers $\gamma$ given by
\begin{equation}
  \label{eq:polyGammaS}
  \cC_{\Gamma,S} \bydef\{ \gamma \in \cC(G) : \overline \gamma \cap S = \emptyset, \,
  \gamma\cup\Gamma\in \cG(G)\}.
\end{equation}
For a vertex $x \in V(G)$, let
$\cC_{\Gamma,S}(x)\subset \cC_{\Gamma,S}$ be the subset of polymers
$\gamma$ such that $x\in \overline\gamma$.  Note that if
$\gamma,\gamma' \in \cC_{\Gamma,S}(x)$ then they are
incompatible.

Let $\gamma_{\emptyset}$ be the empty polymer, and set
$w(\gamma_{\emptyset},z)\bydef 1$. Let $\mu_{\Gamma,S,x}$ be the
probability measure on $\cC_{\Gamma,S}(x) \cup \gamma_{\emptyset}$
defined by
\begin{equation*}
  \mu_{\Gamma,S,x}(\gamma) \bydef\frac{  w(\gamma,z) Z   ( \cC_{\Gamma \cup \gamma,S \cup x}  ,z)  }{Z (  \cC_{\Gamma,S}  ,z) },
\end{equation*}
where we recall the notation $Z(\cS,z)$ from
Section~\ref{secPolyGeneral}. To verify this is a probability measure,
note the so-called \emph{fundamental
  identity}~\cite{scott2005repulsive}:
\begin{equation*}
  Z (  \cC_{\Gamma,S}  ,z) 
  = \sum_{\gamma\in\cC_{\Gamma,S}(x)\cup\{\gamma_{\emptyset}\}}
  w(\gamma,z) Z   ( \cC_{\Gamma \cup \gamma,S \cup x}  ,z).
\end{equation*}

\begin{algorithm}
  \label{algPolymerSample}
  Set $\Gamma_0 = \emptyset$, $S_0 = \emptyset$,
  and order the vertices of $G$ by $x_1, \dots ,x_n$. Repeat the
  following procedure for $t=0$ to $n-1$:
  \begin{enumerate}
  \item Sample $\gamma$ from the measure $\mu_{\Gamma_{t},S_{t},x_{t+1}}$.
  \item Set $\Gamma_{t+1} = \Gamma_{t} \cup \gamma$.
  \item Set $S_{t+1} = S_{t} \cup x_{t+1}$.
  \end{enumerate}
  Return $\Gamma =\Gamma_n$.
\end{algorithm}

\begin{lemma}
  \label{Lempolysampleexact}
  The distribution of $\Gamma$ returned by
  Algorithm~\ref{algPolymerSample} is exactly $\mu_G(\Gamma)$.
\end{lemma}
\begin{proof}
  By construction the algorithm only outputs collections
  $\Gamma = \{\gamma_1, \dots, \gamma_k\} $ of polymers that belong to
  $\cG(G)$, so it suffices to compute the probability the algorithm
  outputs a particular $\Gamma\in\cG(G)$.  

  We first claim that each $\gamma \in \cC(G)$ has at most one chance
  to be added to the collection $\Gamma$: at the first step $i$ so
  that $x_i \in \overline \gamma$.  For $j<i$,
  $\gamma \notin \cC_{\Gamma_{j-1},S_{j-1}}(x_j)$ since
  $x_j \notin \overline \gamma$. For $j >i$,
  $\gamma \notin \cC_{\Gamma_{j-1},S_{j-1}}(x_j)$ since
  $S_{j-1} \cap \overline \gamma \ne \emptyset$.  With this in mind,
  given a collection
  $\Gamma = \{\gamma_{1},\cdots,\gamma_{k}\}\in\cG(G)$, let
  $i(j)=\min \{ i : x_{i}\in\overline\gamma_{j}\}$.

  Without loss of generality we may assume the $i(j)$'s are
  strictly increasing.  Set $\Gamma_j = (\gamma_1, \dots, \gamma_j)$ and
  $X_j = \{x_1, \dots , x_j\}$.  Using the convention that $i(0) =0$,
  the probability that $\Gamma $ is returned by the sampling algorithm
  is
\begin{align*}
  \mu_{\text{alg}}(\Gamma) 
  &= \prod_{j=1}^k   \left(   \frac{ w(\gamma_j,z)  Z (
    \cC_{\Gamma_{j},X_{i(j)}}  ,z)  }{ Z (
    \cC_{\Gamma_{j-1},X_{i(j)-1}}  ,z)   }\cdot  \prod_{i=i(j-1)+1}^{i(j)-1}
    \frac{Z (  \cC_{\Gamma_{j-1},X_{i}}  ,z)    }{Z (
    \cC_{\Gamma_{j-1},X_{i-1}}  ,z)   }   \right) 
    \\
       &\quad \quad \times  \prod_{i=i(k)+1}^n  \frac{ Z (  \cC_{\Gamma_{k},X_{i}}  ,z)   }{Z (
    \cC_{\Gamma_{k},X_{i-1}}  ,z)     } \\                          
  &= \prod_{j=1}^k \frac{ w(\gamma_j,z)Z (  \cC_{\Gamma_{j},X_{i(j)}}
    ,z)   }{ Z (  \cC_{\Gamma_{j-1},X_{i(j-1)}}  ,z)    } \cdot
    \frac{ Z(\cC_{\Gamma_{k},X_{n}},z)}{Z(\cC_{\Gamma_{k},X_{i(k)}},z)}\\                       
  &= \frac{  \prod_{j=1}^k w(\gamma_j,z)} { Z( G,z) }                 
\end{align*}
which is $\mu_{G}(\Gamma)$, as desired. In the third equality we have
used the fact that $Z ( \cC_{\Gamma_{k},X_{n}} ,z)=1$ and
$Z ( \cC_{\Gamma_{0},X_{i(0)}} ,z) =Z( G,z) $.
\qedhere \end{proof}

To turn Algorithm~\ref{algPolymerSample} into an efficient approximate
sampling algorithm, we will sample approximately from the measures
$\mu_{\Gamma_{t},S_{t},x_{t+1}}$.  To analyze the effect on the output
distribution we need a lemma about total variation distance.

Given a family of probability measures
$\{\mu_{\alpha}\}_{\alpha\in\mathcal{A}}$, a \emph{$\mu$-sequence of
  length $n$} is a sequence of random variables $(X_{i})_{i=1}^{n}$,
where the conditional distribution of $X_{i}$ is 
$\mu_{\alpha_{i}}$ for some $\alpha_{i}\in\mathcal{A}$ that is
  a function of the values of the random variables $X_{j}$, $j<i$.
\begin{lemma}
  \label{lemTV2}
  Let $(\mu_{\alpha})_{\alpha\in \mathcal{A}}$ and
  $(\nu_{\alpha})_{\alpha\in \mathcal{A}}$ be families of probability
  measures on a finite set, and suppose
  $\|\mu_{\alpha}-\nu_{\alpha}\|_{TV}<\eps'$ for all
  $\alpha\in\mathcal{A}$. Then if $\eps'<\eps^{2}/(9n^{2})$ the total
  variation distance between the distributions of $\mu$- and
  $\nu$-sequences of length $n$ is at most $\eps$.
\end{lemma}
\begin{proof}
  The hypothesis $\|\mu_{\alpha}-\nu_{\alpha}\|_{TV}<\eps'$ implies
  the subset $A(\alpha)$ of outcomes such that
  \begin{equation*}
    (1-\sqrt{\eps'})\mu_{\alpha}(a)\leq \nu_{\alpha}(a)\leq
    (1+\sqrt{\eps'})\mu_{\alpha}(a),\qquad a\in A(\alpha),
  \end{equation*}
  has measure $\mu_{\alpha}(A(\alpha))\geq 1-2\sqrt{\eps'}$ for
    all $\alpha\in\mathcal{A}$.  

  Let $(X_{i})_{i=1}^{n}$ and $(Y_{i})_{i=1}^{n}$ be $\mu$-
    and $\nu$-sequences of length $n$, respectively.  Write $\mu$ for the law of
  the $\mu$-sequence and similarly for $\nu$. Let $A$ be
    the event that for each $1\leq j\leq n$ both
    $X_{i}\sim\mu_{\alpha_{i}}$ and $Y_{i}\sim \nu_{\alpha_{i}}$
    take values in $A(\alpha)$.
  By a union bound $A$ has $\mu$-measure at
  least  $1-2n\sqrt{\eps'}$.  Moreover, 
  \begin{equation}
    \label{eq:munu}
    \nu(\bar a) 
    \geq (1-n\sqrt{\eps'})
    \mu(\bar a)
  \end{equation}
  for any $\bar a = (a_1,\dots, a_n)\in A$. Recalling the
  definition of $\eps'$, the claim now follows as
  \begin{equation*}
    \|\mu-\nu\|_{TV} 
    = \sum_{\mu(\bar a)>\nu(\bar a)}  \mu(\bar a)-\nu(\bar a)
    \leq \mu(A^{c}) 
    + \sum_{\mu(\bar a)>\nu(\bar a)}\mu(\bar a)n\sqrt{\eps'}
    <3n\sqrt{\eps'},
  \end{equation*}
  where we have used the estimate~\eqref{eq:munu} to obtain the
  inequality by splitting the sum into those $\bar a\in A$ and those not.
\qedhere \end{proof}

We need a lemma that tells us we only need to consider polymers of
size at most $O(\log (n/\eps))$. Recall $\cC_m(G)\subset\cC(G)$ is the
set of all polymers of size at most $m$, and let $\cG_m(G)$ be the
collection of all sets of compatible polymers from $\cC_m(G)$. Let
$Z_m(G,z)$ denote $Z(\cC_m(G),z)$, where this partition function is
defined according to Section~\ref{secPolyGeneral}. Let $\mu_{G,m}$ be
the corresponding probability measure.  We consider $\mu_{G,m}$ as a
measure on $\cG(G)$ by setting $\mu_{G,m}(\Gamma)=0$ for any
collection $\Gamma\in \cG(G)$ that contains a contour of size larger
than $m$.

\begin{lemma}
  \label{lemPolySampSmall}
  Suppose the polymer model satisfies Assumption~\ref{asAnalytic} with
  constant $\rho$, $Z(G,z)$ is a polynomial of degree at most $C|G|$
  for all $G \in \mathfrak G$, and that the cluster expansion for
  $\log Z(G,z)$ converges absolutely for all 
  $G \in \mathfrak G$ and all $|z|<\del$.  Let
 \begin{align}
 \label{meq1}
   m = \Big\lceil{\frac{\log (2C|G|/\eps)}{\rho(1-|z|/\del)}\Big\rceil} \,.
\end{align}
Then
\begin{equation*}
\| \mu_G- \mu_{G,m} \| _{TV} <e^{2\eps}-1 \,.
\end{equation*}
\end{lemma}
\begin{proof}
  For $\Gamma\in \cG_m(G)$ we have
  \begin{equation*}
    \mu_{G,m}(\Gamma)     = \mu_{G}(\Gamma) \frac{Z(G,z)}{Z_m(G,z)}.
  \end{equation*}
  By Lemma~\ref{lemTaylor} and the remark below
  Lemma~\ref{thmPolymerApproxGen} we have
  \begin{equation*}
    e^{-\eps}Z(G,z)\leq \exp[T_{m}(G,z)] \leq e^\eps Z_m(G,z)
  \end{equation*}
  as the degree of both $Z(G,z)$ and $Z_m(G,z)$ is at most
  $C|G|$.  Thus $\| \mu_G- \mu_{G,m} \| _{TV}$ can be estimated by
  \begin{equation*}
    \sum_{\Gamma:\, \mu_{G,m}(\Gamma)>\mu_G(\Gamma)} 
      |\mu_{G,m}(\Gamma)-\mu_G(\Gamma)|
    \leq \sum_{\Gamma} \left|
      \mu_G(\Gamma)\left(\frac{Z(\Lam,z)}{Z_m(\Lam,z)}-1\right)\right|
      \leq e^{2\eps}-1.
\end{equation*}
\qedhere 
\end{proof}

Extend the notation given in~\eqref{eq:polyGammaS} to polymers of
restricted sizes by setting
$\cC^{m}_{\Gamma,S}=\cC_{\Gamma,S}\cap\cC_{m}$ and
$\cC^{m}_{\Gamma,S}(x)=\cC_{\Gamma,S}(x)\cap\cC_{m}$.
Theorem~\ref{thmPolySample} relies on the following algorithm.

\begin{algorithm}
  \label{algPolymerSample2}
  Let $\eps'$ and $m$ be given.  Set $\Gamma_0 = \emptyset$,
  $S_0 = \emptyset$, and order the vertices of $G$ by
  $x_1, \dots ,x_n$. Repeat the following procedure for $t=0$ to $n-1$:
  \begin{enumerate}
  \item Create the list of polymers
    $ \cC^m_{\Gamma_{t},S_{t}}(x_{t+1})$.
  \item For each $\gamma \in \cC^m_{\Gamma_{t},S_{t}}(x_{t+1})$,
    compute $ Y(\gamma)$, an $\eps'$-relative approximation to \\
    $w(\gamma,z) Z ( \cC_{\Gamma_t \cup \gamma,S_t \cup x_{t+1}} ,z)$.
    Do the same for the empty polymer $\gamma_{\emptyset}$.
  \item Sample $\gamma$ from the measure
    $\hat \mu_{\Gamma_{t},S_{t},x_{t+1}}$ defined by
    \begin{equation*}
    \hat \mu_{\Gamma_{t},S_{t},x_{t+1}} (\gamma) = \frac{ Y(\gamma)}{Y(\gamma_{\emptyset})+\sum_{\gamma \in    \cC^m_{\Gamma_{t},S_{t}}(x_{t+1})}  Y(\gamma)} \,.
    \end{equation*}
  \item Set $\Gamma_{t+1} = \Gamma_{t} \cup \gamma$.
  \item Set $S_{t+1} = S_{t} \cup x_{t+1}$.
  \end{enumerate}
  Return $\Gamma =\Gamma_n$.
\end{algorithm}

\begin{proof}[Proof of Theorem~\ref{thmPolySample}]
  Let $n = |G|$, and let $m$ be as in the statement of
    Lemma~\ref{lemPolySampSmall}.  By Lemma~\ref{lemPolySampSmall} it
  is enough to show that Algorithm~\ref{algPolymerSample2} produces an
  $\eps$-approximate sample from $\mu_{G,m}$  in time polynomial in $n$ and $1/\eps$.  By Lemma~\ref{lemTV2},
  Algorithm~\ref{algPolymerSample2} will output an $\eps$-approximate
  sample from $\mu_{G,m}$ if each approximation in step (2) is an
  $\eps' = O(\eps^2/n^2 )$-relative approximation.  

  Since there are only $n$ steps in Algorithm~\ref{algPolymerSample2},
  what remains is to show that each step of the algorithm takes time
  polynomial in $n$ and $1/\eps$. The creation of the list in step (1)
  can be done in polynomial time by Lemma~\ref{lemPolymerCount}: first
  we list all polymers in $\cC_{m}$ in polynomial time. We can then
  determine which polymers are in $\cC_{\Gamma_{t},S_{t}}(x_{t+1})$ by
  checking, for each $\gamma'\in\cC_{m}$, (i) if
  $x_{t+1}\in \overline \gamma'$ and (ii) if there is any $s\in S$ so
  that $s\in \overline \gamma'$ or any
  $v\in\bigcup_{\gamma\in\Gamma}\overline \gamma$ so that
  $d(v, \overline \gamma')\le 1$.  Since there are at most $n$
  vertices to check, this last step takes time at most $O(n^{2})$.

  Computing the approximations in step (2) can be done in polynomial
  time, as (i) $\cC^{m}_{\Gamma_{t},S_{t}}(x_{t+1})$ has size at most
  polynomial in $n$ and $1/\eps$ by the definition of $m$, (ii) we can
  obtain $\eps'$-relative approximations to the partition functions in
  polynomial time by Lemma~\ref{thmPolymerApproxGen}, and (iii) we can
  obtain $\eps'$-relative approximations to the weight functions
  $w(\gamma,z)$ by Assumption~\ref{asReal}.
\qedhere \end{proof}

\subsection{Sampling from a contour model}
\label{sec:sampl-from-cont}
Sampling from a contour model is almost the same as sampling from a
polymer model, but we must be precise about which probability measure
we sample from and the notions of incompatibility used.

For $z>0$ define the probability measure $\mu_{\Lam}^{\varphi}$ on
$\cG^\varphi_{\text{ext}}(\Lam)$ by
\begin{equation}
\label{eqOuterContourMeasure}
\mu_\Lam^{\varphi}(\Gamma) \bydef \frac{\prod_{\gamma \in \Gamma} \left ( z^{\| \gamma \|} \prod_{\varphi' \in \Xi} Z^{\varphi '} (\text{int}_{\varphi'} \gamma ,z) \right) 
}{Z^\varphi(\Lam,z)} \,.
\end{equation}

\begin{theorem}
  \label{mainThmSample}
  Fix $d\ge 2$ and suppose the conditions of
  Lemma~\ref{thmContourModelCountGen} hold for $\varphi \in \Xi$.
  Then for any $0<z<\del$ there is an efficient
  sampling algorithm for the measure $\mu^\varphi_\Lam$ given
  in~\eqref{eqOuterContourMeasure} for any region $\Lam \subset \Z^d$.
\end{theorem}

The algorithm we use to prove Theorem~\ref{mainThmSample} will be a
version of Algorithm~\ref{algPolymerSample2} suited to contour
models. The following definitions are analogues of those in
Section~\ref{sec:meas-from-polym}. The main difference is that
compatibility of polymers now becomes compatibility and mutual
externality of contours of the same type, and so instead of attempting
to add a polymer such that $x \in \overline \gamma$, we attempt to add
a contour such that $x \in \cov(\gamma)$, where we recall that
$\cov(\gamma)$ was defined in~\eqref{eqcov}.
 
For $S \subset \Lam$ and a collection of compatible external contours
$\Gamma \in \cG^\varphi(\Lam)$, let
\begin{equation}
  \label{eq:contGammaS}
  \cC^{\varphi}_{\Gamma,S} = \cC^{\varphi}_{\Gamma,S}(\Lam) \bydef \{
  \gamma\in\cC^{\varphi}(\Lam) :
  \cov (\gamma)\cap S = \emptyset,\,
  \gamma\cup\Gamma\in\cG^{\varphi}(\Lam)\}.
\end{equation}
For $x\in\Lam$, set $\cC^{\varphi}_{\Gamma,S}(x)$ denote the subset of
contours in $\cC^{\varphi}_{\Gamma,S}$ such that $x\in \cov(\gamma)$.
Note that if $\gamma, \gamma' \in \cC^{\varphi}_{\Gamma,S}(x)$ then
$\gamma$ and $\gamma'$ are not mutually external.

Let $\gamma_{\emptyset}$ be the empty contour, and set
$w^{\text{ext}}(\gamma_{\emptyset},z)\bydef 1$. Let
$\mu^\varphi_{\Gamma,S,x}$ be the probability measure on
$\cC^\varphi_{\Gamma,S}(x) \cup \gamma_{\emptyset}$ defined by
\begin{equation*}
\mu^\varphi_{\Gamma,S,x}(\gamma) = \frac{  w^{\text{ext}}(\gamma,z) Z^\varphi( \cC^\varphi_{\Gamma \cup \gamma,S \cup x}  ,z)  }{Z^\varphi (  \cC^\varphi_{\Gamma,S}  ,z) },
\end{equation*}
where $Z^{\varphi}(\cC^{\varphi}_{\Gamma,S})$ is defined as in
Section~\ref{secContourGen}.

\begin{algorithm}
  \label{algContourSample}
  Let $n = |\Lam|$. Set $\Gamma_0 = \emptyset$,
  $S_0 = \emptyset$, and order the vertices of $\Lam$ by
  $x_1, \dots, x_n$.\footnote{We could consider only vertices
    $x$ such that $d_\infty(x, \Lam^c)>1$, but it does no harm to
    include the others.} Repeat the following
  procedure for $t=1$ to $n$:
\begin{enumerate}
\item Sample $\gamma$ from the measure $\mu^\varphi_{\Gamma_{t-1},S_{t-1},x_t}$.
\item Set $\Gamma_t = \Gamma_{t-1} \cup \gamma$.
\item Set $S_t = S_{t-1} \cup x_t$.  
\end{enumerate}
Return $\Gamma =\Gamma_n$.
\end{algorithm}

\begin{lemma}
  The output $\Gamma$ of Algorithm~\ref{algContourSample} has
  distribution $\mu^\varphi_\Lam$.
\end{lemma}
\begin{proof}
  The proof is the same as that of Lemma~\ref{Lempolysampleexact}.
\qedhere \end{proof}

To turn Algorithm~\ref{algContourSample} into an efficient approximate
sampling algorithm we follow the recipe used in obtaining
Algorithm~\ref{algPolymerSample2}: we only consider contours of
size $O(\log(n/\eps))$ and we approximate the weight functions and
partition functions involved in the probability measures
$\mu^\varphi_{\Gamma_{t-1},S_{t-1},x_t}$. The details follow.

We let
$\cC_m^\varphi(\Lam) \bydef \{ \gamma \in \cC^\varphi(\Lam) :
|\overline \gamma| \le m \}$, $\cG^\varphi_{\text{ext},m} (\Lam)$ be
the collection of all sets of mutually external contours from
$\cC_m^\varphi(\Lam)$,
$ \cC^{\varphi,m}_{\Gamma,S} \bydef \cC^{\varphi}_{\Gamma,S} \cap
\cC_m^\varphi(\Lam)$, and lastly
$\cC^{\varphi,m}_{\Gamma,S}(x) \bydef \cC^{\varphi}_{\Gamma,S}(x) \cap
\cC_m^\varphi(\Lam)$. With these definitions, we can present the
approximate sampling algorithm.

\begin{algorithm}
  \label{algContourSampleEfficient}
 Let $\eps'$ and $m$ be given and let $n= |\Lam|$.  Set $\Gamma_0 = \emptyset$,
 $S_0 = \emptyset$, and order the vertices of $\Lam$ by $x_1, \dots,
 x_n$. Repeat the following procedure for $t=0$ to $n-1$:
 \begin{enumerate}
 \item Create the list of contours
   $ \cC^m_{\Gamma_{t},S_{t}}(x_{t+1})$.
 \item For each $\gamma \in \cC^m_{\Gamma_{t},S_{t}}(x_{t+1})$, compute
   $ Y(\gamma)$, an $\eps'$-relative approximation to\\
   $w^{\text{ext}}(\gamma,z)  Z^\varphi( \cC^\varphi_{\Gamma_t \cup \gamma,S_t \cup x_{t+1}}  ,z)$. 
  Do the same for the empty contour $\gamma_{\emptyset}$. 
\item Sample $\gamma$ from the measure
  $\hat \mu_{\Gamma_{t},S_{t},x_{t+1}}$ defined by
  \begin{equation*}
    \hat \mu_{\Gamma_{t},S_{t},x_{t+1}} (\gamma) = \frac{ Y(\gamma)}{Y(\gamma_{\emptyset})+\sum_{\gamma \in    \cC^m_{\Gamma_{t},S_{t}}(x_{t+1})}  Y(\gamma)} \,.
\end{equation*}
\item Set $\Gamma_{t+1} = \Gamma_{t} \cup \gamma$.
\item Set $S_{t+1} = S_{t} \cup x_{t+1}$.  
\end{enumerate}
Return $\Gamma =\Gamma_n$.
\end{algorithm}

We now sketch the proof of Theorem~\ref{mainThmSample}; it is
essentially  
that of Theorem~\ref{sec:meas-from-polym}.
\begin{proof}[Proof of Theorem~\ref{mainThmSample}]
  By Lemma~\ref{lemPolySampSmall}, it suffices to sample a
  configuration of outer contours from
  $\cG^\varphi_{\text{ext},m} (\Lam)$ with $m=O(\log(n/\eps))$.  We
  then implement Algorithm~\ref{algContourSampleEfficient} with
  $\eps' = O(\eps^2/n^2)$, where the $\eps'$-relative approximations
  $Y(\gamma)$ can be computed in time polynomial in $n$ and $1/\eps$
  by Lemma~\ref{thmContourModelCountGen}. Finally we use
  Lemma~\ref{lemTV2} to say that the output of the approximate
  algorithm is a close approximation to the truncated contour
  probability measure.
 \qedhere \end{proof}

\subsection{Applications of Theorem~\ref{mainThmSample}}
\label{secAppsamples}

The algorithm of Theorem~\ref{mainThmSample} returns a collection of
contours $\Gamma$ approximately distributed according to the outer contour measure
$\mu^\varphi_\Lam$. If the contour model arises from a spin system
such as the Potts model or hard-core model it is straightforward to
recover a spin configuration from inductive calls to this algorithm.
We show how to partially determine a configuration
$\omega \in \Omega_\Lam^\varphi$ given a set of outer contours
$\Gamma \in \cG^\varphi(\Lam)$.  

For each $\gamma \in \Gamma, i \in \overline \gamma$, set $\omega_i$
to the spin indicated by $\omega_{\overline \gamma}$.  For each
$i \in \Lam$ so that
$i \in \bigcap _{\gamma \in \Gamma} \text{ext} \gamma$, set $\omega_i$
to the spin indicated by the ground state $\varphi$ (e.g., for Potts
set $\omega_i =\varphi$, and for hard-core set
$\omega_i =\omega^\varphi_i$, where we recall
$\omega^{\text{even}}_{i}$ is the all-even occupied configuration).
This leaves $\omega_i$ so that
$i \in \text{int}\gamma, \gamma \in \Gamma$ unset.  To determine these
spins, call the algorithm again for $\text{int}_\varphi' \gamma$ for
each $\gamma \in \Gamma$ and each $\varphi' \in \Xi$.

Using the correspondence between spin configurations and contour configurations given in Examples~\ref{ex:FP-Contour} and~\ref{ex:HC-Contour}, this proves the
sampling portions of Theorems~\ref{PottsMainThm} and~\ref{HCMainThm}.

\section{The torus}
\label{secTorus}

In this section we give counting and sampling algorithms for contour
models on the torus $\tor = \Z^{d}/(n\Z)^{d}$. We first explain how
contour models are defined in this context, and in the subsequent
sections we indicate how our previous algorithms can be extended to
this setting.  For the rest of this section we will fix $d\ge 2$ and
write $\torr$ for $\tor$ to simplify the notation.

\subsection{Contour models on $\torr$}
\label{sec:contour-models-tor}

Contour models on $\torr$ are defined almost exactly as for regions
$\Lam\subset\Z^{d}$, but some additional care is needed as the change
in topology affects the notion of the exterior of a contour. In the
approach below we will largely circumvent topological complications by
distinguishing contours that are `large', i.e., those that can detect
the change in topology. Large contours make negligible contributions
in the cases we are interested in. 

A contour $\gamma$ on the torus $\torr$ is a pair
$(\overline \gamma, \omega_{\overline \gamma})$ consisting of a subset
of vertices $\overline \gamma \subset \torr$ and an assignment
$\omega_{\overline \gamma}$ of spins from $\Omega$ to
$\overline\gamma$.  Letting $x_{i}$ denote the $i$th coordinate of
$x\in\torr$, the diameter of a set $A \subset \torr$ is
\begin{equation*}
  \text{diam}(A) \bydef \max_{x,y \in A} d_\infty ( x,y) \,.
\end{equation*}
Following~\cite{borgs1989unified}, we distinguish between two types of
contours, those that are `small' and those that are `large'. A
\emph{small contour} is a contour $\gamma$ for which
$\text{diam}(\overline\gamma)<n/2$ and $\overline\gamma$ is
$d_{\infty}$-connected. A \emph{large contour} is a contour $\gamma$
for which each $d_{\infty}$-connected component of $\overline\gamma$
has diameter at least $n/2$; note that the support of a large contour need
not be connected.  Each contour $\gamma$ partitions
$\torr \setminus \overline \gamma$ into $d_{\infty}$-connected
components $A_0, \dots, A_t$.  Since each small contour is a subset of
a $d_{\infty}$-ball of radius less than $n/2$, we can define the
\emph{exterior of a small contour} $\gamma$ to be the unique region
with diameter at least $n/2$, and without loss of generality we can
denote this region by $A_{0}$.  The regions $A_{i}$,
$i\geq 1$, are \emph{interior regions}. For a large contour $\gamma$ we set
$A_{0}=\emptyset$ and refer to all connected components of
$\torr\setminus\overline\gamma$ as interior regions. 

A \emph{contour model} on $\torr$ consists of a set of contours $\cC$,
a surface energy $\|\gamma\|$, and labelling function
$\text{lab}_\gamma(\cdot)$ taking values in $\Xi$ for each contour
$\gamma \in \cC$. The label of the exterior $A_{0}$ of a small contour
is called the \emph{type} of the contour.  Large contours $\gamma$
have no exteriors and hence no type, but are still equipped with a
surface energy $\|\gamma\|$ and a labelling function
$\text{lab}_\gamma(\cdot)$ from the connected components of
$\torr \setminus \overline \gamma$ to $\Xi$.

Two contours $\gamma, \gamma'$ are \emph{compatible} if
$d_\infty(\overline \gamma,\overline \gamma')>1$, and two compatible
small contours $\gamma, \gamma'$ of the same type are \emph{mutually
  external} 
if $\overline \gamma \subset \text{ext} \gamma'$ and
$\overline \gamma' \subset \text{ext} \gamma$.

Let $\cC^\varphi(\torr)$ be the set of all small contours of type
$\varphi$. Let $\cG^\varphi_{\text{ext}}(\torr)$ be the collection of
all sets of contours from $\cC^\varphi(\torr)$ that are compatible and
mutually external.  Let $\cC_{\text{large}}(\torr)$ be the set of all
large contours.

For $\cc S\subseteq \cC^{\varphi}(\torr)$, let
$\cG^{\varphi}_{\text{ext}}(\cc S)$ be the collection of all sets of
compatible and mutually external contours of type $\varphi$ from
$\cc S$, and define
\begin{equation}
  \label{eq:TorSmallSubset}
  Z^{\varphi}(\cc S,z) \bydef
  \sum_{\Gamma\in\cG^{\varphi}_{\text{ext}}(\cc S)} \prod_{\gamma\in\Gamma} 
  \left( z^{\| \gamma \|} \prod_{\varphi'} 
    Z^{\varphi'} (\text{int}_{\varphi'} \gamma,z)  \right) \,.
\end{equation}

We now define the \emph{partition function of a contour model on
  $\torr$} by
\begin{equation*}
  Z(\torr, z) \bydef Z^{\text{big}}(\torr, z) +\sum_{\varphi \in \Xi}Z^\varphi(\torr,z), 
\end{equation*}
where
\begin{equation*}
  Z^{\text{big}}(\torr, z) 
  \bydef \sum_{\gamma \in \cC_{\text{large}}(\torr)}  z^{\| \gamma \|}
  \prod_{\varphi \in \Xi} Z^\varphi (\text{int}_{\varphi} \gamma,z), 
\end{equation*}
and where $Z^{\varphi}(\torr,z)$ is shorthand for $Z^{\varphi}(\cc S,z)$
with $\cc S = \cC^{\varphi}(\torr)$. Note that each contour
configuration contributing to $Z(\torr,z)$ contains at most one large
contour.

We can also write the partition function in an expanded form involving
a matching condition. A small contour $\gamma$ in a collection
$\Gamma$ of compatible small contours is \emph{external} if
$\overline \gamma \subset \text{ext}\gamma'$ for all
$\gamma ' \in \Gamma$, $\gamma'\ne \gamma$. As in
Section~\ref{secContourModels}, a set $\Gamma$ of compatible small
contours is matching if (i) all external contours have the same type
and (ii) for each external contour $\gamma$ and ground state $\varphi$
the subcollection of contours $\Gamma'$ whose support is contained in
$\text{int}_{\varphi}\gamma$ is matching and of type $\varphi$.  A set
of compatible contours $\Gamma$ containing exactly one large contour
$\gamma$ is matching if for each ground state $\varphi\in\Xi$ the
subcollection of contours $\Gamma'$ whose support is contained in
$\text{int}_{\varphi}\gamma$ is matching and of type $\varphi$.  Let
$\cG_{\text{match}}(\torr)$ be the collection of all sets of matching
contours.  Then
  \begin{align}
  \label{eqContourRepresentationTor}
Z(\torr, z) &= (|\Xi| -1) + \sum_{\Gamma \in \cG_{\text{match}}(\torr)}  \prod_{\gamma \in \Gamma} z^{\| \gamma \|} \,.
\end{align}
The term $(|\Xi| -1)$ is due to the fact that for each $\varphi\in\Xi$
there is a contribution of $1$ to $Z^{\varphi}(\torr,z)$ from the
empty collection of contours.

Moreover, for each $\varphi$ and $\cc S\subset \cC$, let
$\cG^\varphi_{\text{match}}(\cc S)$ be the collection of sets of
matching small contours from $\cc S$ whose external contours are all
of type $\varphi$.  Then
\begin{equation*}
Z^\varphi(\cc S,z) = \sum_{\Gamma \in \cG^\varphi_{\text{match}}(\cc S)} \prod_{\gamma \in \Gamma} z^{\|\gamma\|} \,.
\end{equation*}
and again we let $Z^{\varphi}(\torr,z) = Z^{\varphi}(\cc S,z)$ with
$\cc S = \cC(\torr)$.

Borgs and Imbrie show that under the Peierls condition, for small
enough $z$ the relative weight of $Z^{\text{big}}(\torr,z)$ in
$Z(\torr,z)$ is exponentially small in $n$. More precisely, and noting
that the definition of stable ground states $\Ostab$ from
Section~\ref{secZeroFree} applies equally well to the partition
functions of small contours on $\torr$, they prove:
\begin{theorem}[Borgs, Imbrie~\cite{borgs1989unified}]
  \label{thmTorusBorgs}
  Suppose the contour model satisfies Assumption~\ref{asSurface} for
  some $\rho,C>0$. Then there exists a constant
  $\del= \del(d, \rho, \Xi)>0$ and constants $N,c' >0$ so that for
  $n > N$, and real $0<z < \del$,
  \begin{equation*}
    \frac{ |Z(\torr, z) - \sum_{\varphi \in \Ostab}
      Z^\varphi(\torr,z)|  }{ |Z(\torr,z)|  } 
    \le e^{-c' n} \,.
  \end{equation*}
  Moreover, for all complex $|z| <\del$ and all $\varphi \in \Ostab$,
  $Z^\varphi (\torr,z) \ne 0$.
\end{theorem}

Using this result, we prove our main counting result for the torus. In
Section~\ref{secsampleTorus} below we apply the result to prove
Theorems~\ref{PottsTorus} and~\ref{HCTorus}.
\begin{theorem}
  \label{mainThmTorCount}
  Fix $d\ge 2$, suppose the contour model satisfies
  Assumptions~\ref{assumCompute} and~\ref{asSurface} and that all
  ground states $\varphi \in \Xi$ are stable.  Then there exists a
  constant $\del= \del(d, \rho, \Xi)>0$ and a constant
  $c=c(d,\rho,\Xi)>0$ so that for all real $0 <z<\del$ and all
  $\eps \ge e^{-cn}$, there is an algorithm to obtain an
  $\eps$-relative approximation to $Z(\torr,z)$ in time polynomial in
  $n$ and $1/\eps$.
\end{theorem}

The conclusion of Theorem~\ref{mainThmTorCount} is slightly weaker
than that of Theorem~\ref{thmContourModelCount}, e.g., it does not
allow $\eps$ to be exponentially small in $n^{d-1}$.  See
Section~\ref{secConclude} for comments on obtaining a full FPTAS.

Note that we require $z$ to be positive in
Theorems~\ref{thmTorusBorgs} and~\ref{mainThmTorCount}.  This is
because for complex or negative $z$ there could be cancellations in
the sum of partition functions associated to the stable ground states.
For models with a symmetric set of ground states, like the Potts and
hard-core models, we can take $|z| < \del$ complex in both theorems as
these cancellations cannot occur.

\begin{proof}[Proof of Theorem~\ref{mainThmTorCount}]
  Let $c = c'/2$ where
  $c'$ is the constant from Theorem~\ref{thmTorusBorgs}, and choose
  $\eps=\eps(n)\ge e^{-cn}$. By Theorem~\ref{thmTorusBorgs}, for $ 0< z
  <\del$ and $n$ large enough we know $\sum_{\varphi\in\Xi}Z^{\varphi}(\torr,z)$
  is an $\eps/2$-relative approximation to
  $Z(\torr, z)$. Hence it suffices to compute
  $\eps/2$-relative approximations to $Z^\varphi(\torr,
  z)$ for each $\varphi \in \Xi$.

  We can compute an $\eps/2$ approximation to $Z^\varphi(\torr, z)$
  almost exactly as in the proof of
  Theorem~\ref{thmContourModelCount}.  Lemma~\ref{lemPolymerCount}
  applies to $\Lam = \torr$ as $\torr$ is a graph of bounded degree.
  The proof of Lemma~\ref{lemContourWeightLemma} carries through as
  before; we can still order small contours so that $\gamma$ precedes
  $\gamma'$ if $\gamma$ can appear in the interior of
  $\gamma'$. Moreover, we can inductively compute the weights exactly
  as before, since $\overline\gamma$ has diameter $<n/2$ and so can be
  embedded in $\Z^d$.
\qedhere \end{proof}

\subsection{Sampling on the torus}
\label{secsampleTorus}
Define the following probability measure associated to the matching
contour representation~\eqref{eqContourRepresentationTor}
\begin{equation*}
  \mu_{\torr}^{\text{match}}(\Gamma) 
  \bydef  \frac{  \prod_{\gamma \in \Gamma} 	z^{\| \gamma \|}
       }{ Z(\torr,z) } \, ,\qquad \Gamma \in \cG_{\text{match}}(\torr).
\end{equation*}

Under the conditions of Theorem~\ref{mainThmTorCount} we obtain an
efficient approximate sampling algorithm for $\mu_{\torr}^{\text{match}}$.

\begin{theorem}
  \label{mainSampleTorus}
  Fix $d\ge 2$, suppose the contour model satisfies
  Assumptions~\ref{assumCompute} and~\ref{asSurface} and that all
  ground states $\varphi \in \Xi$ are stable.  Then there exists a
  constant $\del= \del(d, \rho, \Xi)>0$ and a constant
  $c=c(d,\rho,\Xi)>0$ so that for all real $0 <z<\del$ and all
  $\eps \ge e^{-cn}$, there is an $\eps$-approximate sampling
  algorithm for $ \mu_{\torr}^{\text{match}}$ that runs in time
  polynomial in $n$ and $1/\eps$.
\end{theorem}

To prove Theorem~\ref{mainSampleTorus} we need some auxiliary
probability measures. The measure $ \mu_{\torr}^{\text{match}}$
conditioned on $\Gamma \in \cG^\varphi_{\text{match}}(\torr)$ is
\begin{equation*}
  \mu_{\torr}^{\varphi,\text{match}}(\Gamma) 
  \bydef  \frac{  \prod_{\gamma \in \Gamma} 	z^{\| \gamma \|}
       }{ Z^\varphi(\torr,z) } \, ,\qquad \Gamma \in \cG^\varphi_{\text{match}}(\torr).
\end{equation*}

We define the probability measure associated to the outer contour
representation of $Z^\varphi(\torr,z) $ as
\begin{equation*}
  \mu_{\torr}^{\varphi, \text{ext}}(\Gamma) 
  \bydef  \frac{  \prod_{\gamma \in \Gamma}  \left(	z^{\| \gamma \|}
      \prod_{\varphi'} Z^{\varphi'} (\text{int}_{\varphi'} \gamma,z)
    \right) }{ Z^\varphi(\torr,z) } \, ,\qquad \Gamma \in \cG^\varphi_{\text{ext}}(\torr).
\end{equation*}

\begin{lemma}
  \label{lemTorusSample}
  Under the assumptions of Theorem~\ref{mainSampleTorus}, for any
  $0<z <\del$, there is an efficient sampling algorithm for
  $ \mu_{\torr}^{\varphi, \text{ext}}$.
\end{lemma}

The algorithm and proof of Lemma~\ref{lemTorusSample} are exactly the
same as for Theorem~\ref{mainThmSample}. We now
prove Theorem~\ref{mainSampleTorus} using Theorem~\ref{thmTorusBorgs},
Theorem~\ref{mainThmTorCount}, and Lemma~\ref{lemTorusSample}.

\begin{proof}[Proof of Theorem~\ref{mainSampleTorus}]
  With $c =c'/2$, where $c'$ is the constant from
  Theorem~\ref{thmTorusBorgs} and $\eps \ge e^{-cn}$, to obtain an
  $\eps$-approximate sample from $\mu_{\torr}^{\text{match}}$ it
  suffices to select $\varphi \in \Xi$ with probability approximately
  proportional to $Z^\varphi(\torr,z)$ and return an $\eps/2$
  approximate sample from $\mu_{\torr}^{\varphi, \text{match}}$.  We
  can choose the ground state $\varphi$ within total variation distance
  $\eps/2$ by approximating $Z^{\varphi'}(\torr,z)$ within relative
  error $\eps/2$ for each $\varphi'$ using
  Theorem~\ref{mainThmTorCount}.

  To obtain an $\eps/2$ approximate sample from
  $\mu_{\torr}^{\varphi, \text{match}}$, we obtain an
  $\eps/(2n)$-approximate sample from
  $\mu_{\torr}^{\varphi, \text{ext}}$ using Lemma~\ref{lemTorusSample}
  and then proceed inductively on the interior regions, as in
  Section~\ref{secAppsamples}.  To sample approximately from
  $\mu^{\varphi'}_{\text{int}_{\varphi'} \gamma}$ we can use
  Theorem~\ref{mainThmSample} as
  $\text{diam}(\text{int}_{\varphi'} \gamma) < n/2$ and so it can be
  embedded in $\Z^d$.  We return the collection $\Gamma$ of all
  contours sampled at each step which is by definition a set of
  matching contours.
\qedhere \end{proof}

\subsection{Applications}
Theorem~\ref{mainSampleTorus} immediately implies
Theorems~\ref{PottsTorus} and~\ref{HCTorus} by the same mapping of a
set of matching contours to a spin configuration given in
Section~\ref{secAppsamples}.  Note that we must take $n$ even in
Theorem~\ref{HCTorus} so that we can properly define the contour
models.

\section{Conclusions}
\label{secConclude}

We conclude by describing some open problems.

\subsection{Extending the region of applicability}
\label{sec:extend-regi-appl}

It would be interesting to optimize the ranges of parameters for which
our algorithms work.  The proofs of Theorems~\ref{PottsMainThm}
and~\ref{PottsTorus} for the Potts model use techniques from
mathematical physics~\cite{laanait1991interfaces,borgs1991finite} that
have also been used to prove slow mixing of the Swendsen--Wang dynamics
at $\beta_c$ when $q$ is sufficiently
large~\cite{borgs1999torpid,borgs2012tight}.  For large $q$ we
therefore expect that we can take $\beta^\star(d,q) = \beta_c(d,q)$.
In fact, for large $q$ the techniques of this paper yields an
efficient counting algorithm for $\beta> \beta_c$ and
quasi-polynomial-time counting and sampling algorithms for all
$ \beta\ge \beta_c$.

For the hard-core model, it is known that there is phase coexistence
on $\Z^d$ for
$\lam \ge C \frac{\log^2
  d}{d^{1/3}}$~\cite{galvin2004phase,peled2014odd}. It would be of
interest to understand how small $\lam^*$ could be taken to obtain an
efficient sampling algorithm for the hard-core model on $\Z^d$.
\begin{openprob}
  Can Theorems~\ref{HCMainThm} and~\ref{HCTorus} be extended to
  $\lam^*(d) = \tilde \Theta(d^{-1/3})$?
\end{openprob}

A related direction would be to use more geometrically sophisticated
notions of contours to improve the range of parameters for which the
condition~\eqref{eqKPcondition} holds.
\begin{openprob}
  Find an FPTAS and efficient sampling algorithm for the hard-core
  model on $\mathbb T^2_n$ for $\lam >5.3506$, the region of
  coexistence for the hard-core model on $\Z^2$ proved
  in~\cite{blanca2016phase}.
\end{openprob}

With more sophisticated contours, one could hope to find
algorithms for models whose ground states consist of collections of
configurations, e.g., the $q$-coloring model. 

\begin{openprob}
  Find an FPTAS and efficient sampling algorithm for proper
  $q$-colorings of $\tor$ when $d = d(q)$ is sufficiently large.
\end{openprob}

\subsection{An FPTAS for the torus}
\label{sec:tor-ref}

The obstacle to applying Theorem~\ref{mainThmTorCount} to obtain a
genuine FPTAS for the torus is that if $\eps = \exp(- \omega(n))$,
then the bound of Theorem~\ref{thmTorusBorgs} on the contribution from
large contours is not small enough to ignore.  However, by using much
more sophisticated topological tools, Borgs, Chayes, and
Tetali~\cite{borgs2012tight} showed a bound of
$\exp (- \Theta(n^{d-1}))$ for the contributions to the Potts model
partition function due to configurations containing an `interface' of
non-zero winding number on the torus.  This upper bound is matched by
an upper bound of $\exp( \Theta(n^{d-1}))$ on the mixing time of the
Glauber dynamics for the Potts model on the torus in the same paper.

Remarkably, these two ingredients together with the techniques of this
paper can give a true FPTAS and efficient sampling algorithm on the
torus.  If $\eps = \exp ( - o(n^{d-1}))$, then we safely ignore
contributions to the partition function from configurations with
interfaces and run our counting and sampling algorithms.  But if
$\eps = \exp ( - \Omega(n^{d-1}))$ then the Glauber dynamics provide a
sampling algorithm that runs in time polynomial in $n$ and $1/\eps$.
The idea is straightforward, but the topological details are rather
complicated, and so we leave this for future work.

\subsection{Markov chains}
\label{sec:MCMC}

The algorithms we have presented run in time $(n/\eps)^{O( \log d)}$,
which is polynomial in $n$ and $1/\eps$ for fixed $d$ but far from
linear time.  A more efficient approach would be to use a Markov
chain.  While the Glauber dynamics is known to mix slowly at low
temperature in models of the type we consider
here~\cite{borgs1999torpid}, the definition of mixing time is rather
strict and slow mixing does not rule out an efficient sampling
algorithm based on the Glauber dynamics.

For spin models with finitely many stable and symmetric ground
states, like the Potts or hard-core models, we suggest a Markov chain
algorithm to sample on the torus $\tor$.
\begin{enumerate}
\item Pick a ground state $\varphi \in \Xi$ uniformly at random.
\item Run the Glauber dynamics with the ground state configuration corresponding to $\varphi$  as the initial configuration (i.e. a monochromatic initial configuration for the Potts model; all even or all odd occupied for the hard-core model). 
\end{enumerate}
We conjecture that at sufficiently low temperatures (sufficiently high
fugacities) in such models the distribution is close to stationary
after $O(n \log n)$ steps of the Markov chain; we include the
randomness from the choice of the ground state.

\begin{openprob} 
  Prove that the above algorithm is an efficient sampling algorithm
  for the Potts model below the critical temperature or the hard-core
  model at sufficiently high fugacity.
\end{openprob}

For the $2$-dimensional Ising model on a box with all plus boundary
conditions, Glauber dynamics starting from the all plus configuration
does in fact converge rapidly to the stationary distribution for
$\beta>\beta_c$~\cite{lubetzky2013quasi}.

\subsection{Beyond $\Z^d$ and beyond lattices}
We have restricted ourselves to the lattices $\Z^d$ for simplicity, and
because some geometric lemmas about the connectivity of boundaries in
$\Z^d$ have been proved for us
(e.g., \cite[B.15]{friedli2017statistical} and~\cite{timar2013boundary}).  
Similar lemmas can presumably be proved
for general lattices of dimension at least $2$, but we leave this for
future work.  In particular, Theorem~\ref{HCMainThm} can likely be
extended to the entire class of non-sliding models considered by Jauslin
and Lebowitz~\cite{jauslin2017high}.

A related challenge would be to apply these methods to the \emph{hard
  hexagon model} (i.e., the hard-core model on the triangular
lattice) for which it is known that the free energy is analytic for
all real non-critical
fugacities~\cite{baxter1980hard,joyce1988hard,tracy1987modular}.
\begin{openprob} 
  Find efficient counting and sampling algorithms for the hard hexagon
  model for real $\lam \ne \lam_c$.
\end{openprob}

The fact that the underlying graph is a lattice does not seem to be
entirely necessary. Given the interest in the complexity class \#BIS,
it would be interesting to investigate contour representations of the
hard-core model on more general families of bipartite graphs.
See~\cite{liu2015fptas,cai2016hardness,galanis2016ferromagnetic} for
more about \#BIS. A cautionary note in this respect is that
Bez\'akov\'a, Galanis, Goldberg, and
\v{S}tefankovi\v{c}~\cite{bezakova2017inapproximability} have shown
\#P-hardness of approximating $Z_G(\lam)$ on bipartite graphs for
\emph{any} complex $\lam$ with large real part.

\subsection{Approximating the free energy}
\label{sec:appr-free-energy}

A computational problem related to the problems considered in this
paper is to approximate the limiting free energy
$f_d(\lam) \bydef \lim_{n \to \infty} \frac{1}{n} \log
Z_{\tor}(\lam)$.  The objective is an algorithm which, for any
$\eps>0$, outputs a number
$\eta \in [ f_d(\lam) -\eps, f_d(\lam) + \eps]$, and whose running
time grows as slowly as possible as a function of $1/\eps$.  Gamarnik
and Katz~\cite{gamarnik2009sequential} gave a polynomial time
algorithm for the hard-core model for $\lam$ small enough that
\emph{strong spatial mixing} holds. This condition implies the
hard-core model is in the uniqueness regime.  Adams, Brice{\~n}o,
Marcus, and Pavlov~\cite{adams2016representation} gave a
polynomial-time algorithm for several models (including the hard-core
model) on $\Z^{2}$ in a subset of the uniqueness regime. Their results
also apply to the hard-core and Widom--Rowlinson models on $\Z^2$ in a
subset of the non-uniqueness regime. This last result is of a similar
spirit to the results of this paper, and it would be interesting to
understand if our results have any bearing on this problem.

\section*{Acknowledgements}
WP and GR thank Ivona Bez{\'a}kov{\'a}, Leslie Goldberg, and Mark
Jerrum for organizing the 2017 Dagstuhl seminar on computational
counting and Jan Hladk{\`y} for organizing the 2018 workshop on graph
limits in Bohemian Switzerland. Both meetings provided essential
inspiration and discussion leading to this work. TH thanks Roman
Koteck\'{y} for helpful discussions.  We thank Eric Vigoda, Matthew
Jenssen, and Reza Gheissari for detailed comments on a draft of the
paper. We moreover thank Ryan Mann for spotting an error in the proof of Theorem 2.2 in an earlier version of the paper.
TH supported by EPSRC grant EP/P003656/1.  WP supported by in part by EPSRC grant EP/P009913/1 and NSF Career award DMS-1847451. GR supported by an NWO Veni grant.

\end{document}